\documentclass[12pt]{article}
\usepackage{amssymb}
\usepackage{amsmath,bm}
\usepackage{graphics,mathrsfs}
\usepackage{graphicx,epsfig}
\usepackage{subfigure}
\usepackage{placeins}
\usepackage{indentfirst}
\usepackage{setspace}
\usepackage{enumerate}
\usepackage{enumitem}
\usepackage{caption}
\usepackage{varioref}
\usepackage{booktabs,tabularx}
\usepackage{color}
\usepackage{natbib}
\usepackage{epstopdf}
\setcitestyle{authoryear,open={(},close={)}}
\usepackage{amsthm}
\makeatletter \oddsidemargin  -.1in \evensidemargin -.1in
\begin{document}
	\newcommand{\bea}{\begin{eqnarray}}
		\newcommand{\eea}{\end{eqnarray}}
	\newcommand{\nn}{\nonumber}
	\newcommand{\bee}{\begin{eqnarray*}}
		\newcommand{\eee}{\end{eqnarray*}}
	\newcommand{\lb}{\label}
	\newcommand{\nii}{\noindent}
	\newcommand{\ii}{\indent}
	\newtheorem{theorem}{Theorem}[section]
	\newtheorem{example}{Example}[section]
	\newtheorem{cor}{Corollary}[section]
	\newtheorem{definition}{Definition}[section]
	\newtheorem{lemma}{Lemma}[section]
	\newtheorem{remark}{Remark}[section]
	\newtheorem{proposition}{Proposition}[section]
	\numberwithin{equation}{section}
	\renewcommand{\theequation}{\thesection.\arabic{equation}}
	%\bibpunct[, ]{(}{)}{;}{a}{,}{,}
	\renewcommand\bibfont{\fontsize{10}{12}\selectfont}
	\setlength{\bibsep}{0.0pt}
	\title{\bf Statistical inference for dependent competing risks data under adaptive Type-II progressive hybrid censoring }
	\author{Subhankar Dutta\thanks { Subhankar Dutta (Corresponding author : subhankar.dta@gmail.com)}~ and Suchandan Kayal
		%\thanks{Corresponding author : kayals@nitrkl.ac.in,~suchandan.kayal@gmail.com}
	}
	%author{ Subhankar Dutta 1\thanks {Email address: subhankar.dta@gmail.com}~, Farha Sultana 1\thanks{Email adress : farhasultana18@gmail.com}~, Suchandan Kayal 2\thanks {Email address (corresponding author):
	%                    kayals@nitrkl.ac.in,~suchandan.kayal@gmail.com}
	
	%\doublespacing
	
	\maketitle
	%\date{}
	\maketitle 
	\noindent {\it Department of Mathematics, National Institute of		Technology Rourkela, Rourkela-769008, India} \\
	%{\it $^{2}$ Department of Mathematics and Statistics, Indian Institute of Technology, Kanpur-208016, India.}
	%{\it $^{3}$ Applied Statistics Unit, Indian Statistical Institute, Kolkata-700108, India.}
	\begin{center}
		\textbf{Abstract}
	\end{center}

In this article, we consider statistical inference based on dependent competing risks data from  Marshall-Olkin bivariate Weibull distribution. The maximum likelihood estimates of the unknown model parameters have been computed by using Newton-Raphson method under adaptive Type II progressive hybrid censoring with partially observed failure causes. Existence and uniqueness of maximum likelihood estimates are derived. Approximate confidence intervals have been constructed via the observed Fisher information matrix using asymptotic normality property of the maximum likelihood estimates. Bayes estimates and highest posterior density credible intervals have been calculated under gamma-Dirichlet prior distribution by using Markov chain Monte Carlo technique. Convergence of Markov chain Monte Carlo samples is tested. In addition, a Monte Carlo simulation is carried out to compare the effectiveness of the proposed methods.  Further, three different optimality criteria have been taken into account to obtain the most effective censoring plans. Finally, a real-life data set has been analyzed to illustrate the operability and applicability of the proposed methods.
\\
\textbf{MATHEMATICS SUBJECT CLASSIFICATIONS 2010:} 62F10; 62F15; 62E15.

\section{Introduction}
In life-testing experiments, it is very common that due to complexity of the internal structure and  external environment, failure of any product occurs due to various competing causes. In the literature, such types of models are dubbed as the competing risk models. Competing risks data frequently appear in many fields, such as engineering, biology, social science and medical science. In reliability analysis, to get around the issue of model identifiability, many researchers assume that the causes of failure are independent. See, for example \cite{panwar2015competing}, \cite{ashour2017inference}, \cite{koley2017generalized}, \cite{wang2018inference}, \cite{wang2019inference}, \cite{mahto2021inference}, \cite{dutta2022inference}, \cite{dutta2022bayesian} and \cite{dutta2023inference}. However, in general, it is not easy to verify the assumption that the causes of failure are independent to each other. Further, due to complex system design and operating conditions in real-life experiments, the causes of failure affect each other and become inter-dependent. Thus, in practical situations, the assumption of dependent causes of failure is more reasonable than the assumption of independent causes of failure. When analysing data to address dependent competing risks models, it has been shown that the Marshall-Olkin bivariate distributions are frequently used. For details, one may refer to \cite{kundu2013bayes} in this matter. In recent years, statistical inference of dependent competing risks model has gained much attention from many researchers. \cite{feizjavadian2015analysis} discussed the classical estimation of the parameters of Marshall-Olkin bivariate Weibull (MOBW) distribution in the presence of a dependent competing risks model under progressive hybrid censored sample. \cite{liang2019inference} considered inference of the Marshall-Olkin bivariate exponential (MOBE) distribution for dependent competing risks model under generalized progressive hybrid censoring scheme. \cite{du2021statistical} considered the inference for dependence competing risk model by using MOBE distribution under adaptive Type-II progressive hybrid censoring.\\

Censoring has become a commonly used phenomenon in survival analysis and reliability studies. Compared to complete data, if one gets the exact failure time of only a few experimental units in a lifetime experiment, such data are called censored. Nowadays, the quality and lifetime of any industrial product have become higher because of modern science and technology. Thus, it has become more challenging to increase product qualities with reliability testing in a short period with low expenditure. In this context, censoring is sensible and important to study the failure data  in a reliability experiment. A variety of censoring schemes have already been put forth by several researchers in order to improve the efficiency of an experiment. In survival analysis, Type-I and Type-II censoring schemes are most commonly used.  The mixture of these two censoring schemes is known as hybrid censoring schemes. These conventional censoring schemes do not allow removal of experimental units at intermediate points except at the termination point. To overcome such disadvantage, progressive Type-II censoring scheme is introduced, which allows an experimenter to take out units from an experiment at different phases. To improve the efficiency of the experiments, \cite{kundu2006analysis} introduced progressive hybrid censoring scheme (PHCS), which is a combination of the conventional hybrid and progressive censoring schemes. The drawback of PHCS is that the effective sample size may become zero, which is as similar as Type-I censoring scheme. To get beyond such obstacle, \cite{ng2009statistical} introduced adaptive Type-II progressive hybrid censoring scheme (AT-II PHCS), where the experimental time may exceed the predetermined time $T$ and the effective sample size $m$ is prefixed. In AT-II PHCS, a progressive scheme $(R_1,R_2,\cdots,R_m)$ is provided. Consider that the time of $i$-th failure of units is represented as $X_{i:m:n}$. In AT-II PHCS, if $X_{m:m:n} < T$ then the experiment is terminated at $X_{m:m:n}$ with the censoring scheme $(R_1,R_2,\cdots,R_m)$. Otherwise, if $X_{m:m:n}> T,$ then there exists  $d$ such that $X_{d:m:n}<T<X_{d+1:m:n}$ and the experiment terminates at $X_{m:m:n}$ with no removal by setting $R_i=0$, where $i=d+1,\cdots,m-1$ and $R_{m}= n-m-\sum_{i=1}^{d}R_{i}$. This adaptation ensures that the experiment will end once the predetermined number of failures has been obtained. In recent years, AT-II PHCS has received extensive consideration in the statistical literature. \cite{hemmati2011bayesian} discussed the competing risks model based on exponential distributions under AT-II PHCS. \cite{ismail2014inference} considered the classical estimation of parameters of Weibull distribution under step-stress partially accelerated life test model based on AT-II PHCS. \cite{nassar2017estimation} discussed frequentist and Bayesian estimation for the parameters of inverse Weibull (IW) distribution based on AT-II PHCS. \cite{panahi2020estimation} discussed the estimation of parameters of inverted exponentiated Rayleigh distribution under AT-II PHCS.\\

We recall that the AT-II PHCS can be applied in quality inspections of manufacturing plants that require enough failure information in a brief time period. For instance, we cannot observe an adequate number of failures from ten thousand or more electronic product keystrokes in between the specific time and the test price increases there. AT-II PHCS is relevant to apply in such cases. Sometimes the failure risks may be correlated and affects each other to make an individual failure early. Hence the competing risks are usually dependent. For example, \cite{lin1999nonparametric} investigated colon cancer data in which cancer recurrence or death were the causes of failure. Obviously, these two causes are dependent and correlated with each other. Note that the MOBE distribution may not be applicable when bivariate data shows a non-constant hazard rate function or a unimodal marginal probability density function. The MOBW distribution is more appropriate to fit such models. Furthermore, when dependent competing risks exist, a MOBW distribution has a correlation control parameter that may be applied. the duration of paired organ (eyes, lungs and kidneys) failure can be analyzed by using MOBW model in survival analysis. Also in reliability analysis, this model fits to the duration of paired components (aircraft engines) failure. Motivated by such reasons, we have considered a dependent competing risks model using MOBW distributions to discuss statistical inferences based on AT-II PHCS. \\

In this study, our main goal is to analyze and investigate the statistical inference of the parameters of MOBW distribution using both frequentist and Bayesian approaches under AT-II PHC dependent competing risks model. In the main results, maximum likelihood estimates (MLEs) are observed to be inaccessible in closed form. In order to obtain MLEs, a suitable numerical method has been applied. Using the asymptotic normality properties of the MLEs, approximate confidence intervals (ACIs) have been derived. Bayes estimators have been obtained under two different types of loss functions. Since the Bayes estimates cannot be derived in closed form, Markov Chain Monte Carlo (MCMC) method has been implemented.  Highest posterior density (HPD) credible intervals are also constructed using MCMC samples. MCMC convergence is also discussed. Further, a Monte Carlo simulation study has been conducted to assess the effectiveness of the proposed various estimates based on the absolute bias (AB), mean squared error (MSE), average width (AW), and coverage probability (CP). In addition, we have proposed optimal life-testing plans based on three different optimality criteria. \\

The remaining portion of the paper is structured as follows. In Section 2, the model description has been provided. The maximum likelihood estimates of the unknown parameters with their corresponding approximate confidence intervals have been derived in Section 3. In Section 4, the MCMC approach for Bayesian analysis has been demonstrated, and the associated HPD credible intervals have been constructed. MCMC convergence is also tested there. A Monte Carlo simulation study has been carried out in Section 5. In Section 6, three different optimality criteria have been discussed to choose an optimal progressive censoring plan. In Section 7, a real data set has been analyzed to verify the effectiveness of the proposed methods in practical situations. Finally, some concluding remarks have been added in Section 8.

\section{Model and data description}
\subsection{MOBW distribution}
%Here, we present the model which has been studied in this paper. Recall that Marshal-Olkin bivariate exponential distribution was proposed by \cite{marshall1967multivariate}
In this paper, MOBW distribution has been addressed for statistical inference. First, we describe this model very briefly. \cite{marshall1967multivariate} proposed a multivariate lifetime model known as the Marshall-Olkin model, that can be explained as follows. For $i=0,~1,~2$, assume that $V_i$ $\sim$ $Weibull$ $(\alpha, \lambda_{i})$ be independent to each other. The cumulative distribution function (CDF), probability density function (PDF), and survival function of $V_i$ are given by
\begin{align}
	\nonumber F_{WE}(v;\alpha, \lambda_i)= 1- e^{-\lambda_i v^{\alpha}}, ~~ f_{WE}(v;\alpha, \lambda_i)= \alpha \lambda_i v^{\alpha-1} e^{-\lambda_i v^{\alpha}},~~
	\nonumber \mbox{and} ~~S_{WE}(v;\alpha, \lambda_i) = e^{-\lambda_i v^{\alpha}},
\end{align}
where $v>0,~ \alpha,~ \lambda_i>0,$ $i=0,~1,~2.$
Define $Y_{1}= min \{V_0,V_1\}$ and $Y_{2}= min \{V_0,V_2\}$, then $(Y_1,Y_2)$ has MOBW distribution with shape parameter $\alpha$ and scale parameters $\lambda_0$, $\lambda_1$ and $\lambda_2$. Henceforth, we denote $(Y_1,Y_2) \sim MOBW (\lambda_0,\lambda_1,\lambda_2,\alpha)$. Here, $Y_1$ and $Y_2$ are independent if $\lambda_0=0$; otherwise, they are dependent. Hence $\lambda_0$ is the parameter describing correlation between these two variables $Y_1$ and $Y_2$. Let us assume that $\lambda_{012}= \lambda_0+\lambda_1+\lambda_2$ and $\lambda_{ab}= \lambda_a+\lambda_b$, for $a,~b=0,~1,~2$ with $a \neq b$.
\begin{theorem}\label{th2.1}
	If we consider, $(Y_1,Y_2)\sim MOBW(\lambda_0,\lambda_1,\lambda_2,\alpha)$, then the joint survival function of $(Y_1,Y_2)$ is given by
	\begin{align}
		 S_{Y_1,Y_2}(y_1,y_2)
		= \begin{cases}
			S_{WE}(y_1;\alpha, \lambda_1)~ S_{WE}(y_2;\alpha, \lambda_{02}), ~~~\mbox{if}~~ y_1<y_2\\
			S_{WE}(y_1;\alpha, \lambda_{01})~ S_{WE}(y_2;\alpha, \lambda_2),~~~\mbox{if}~~ y_2<y_1 \\
			S_{WE}(y;\alpha, \lambda_{012}),~~~~~~~~~~~~~~~~~~~~~~~\mbox{if}~~y=y_1=y_2.
		\end{cases} \label{2.1}	
	\end{align} 
\end{theorem}
\begin{proof}
	See Appendix $A.1$.
\end{proof}
\begin{cor} \label{cor1}
	If $(Y_1,Y_2)\sim MOBW(\lambda_0,\lambda_1,\lambda_2,\alpha)$, then the joint PDF of $(Y_1,Y_2)$ is given by
	\begin{align}
		f_{Y_1,Y_2}(y_1,y_2)= \begin{cases}
			f_{WE}(y_1;\alpha, \lambda_1)~ f_{WE}(y_2;\alpha, \lambda_{02}), ~~~\mbox{if}~~ y_1<y_2\\
			f_{WE}(y_1;\alpha, \lambda_{01})~ f_{WE}(y_2;\alpha, \lambda_2),~~~\mbox{if}~~ y_2<y_1 \\
			\frac{\lambda_{0}}{\lambda_{012}}f_{WE}(y;\alpha, \lambda_{012}),~~~~~~~~~~~~~~~~~~\mbox{if}~~y=y_1=y_2.
		\end{cases} \label{2.2}
	\end{align}
\end{cor}
\begin{proof}
	See Appendix $A.2$. 
\end{proof}
\begin{figure}[htbp!]
	\begin{center}
		\subfigure[]{\label{c1}\includegraphics[height=1.5in,width=2.2in]{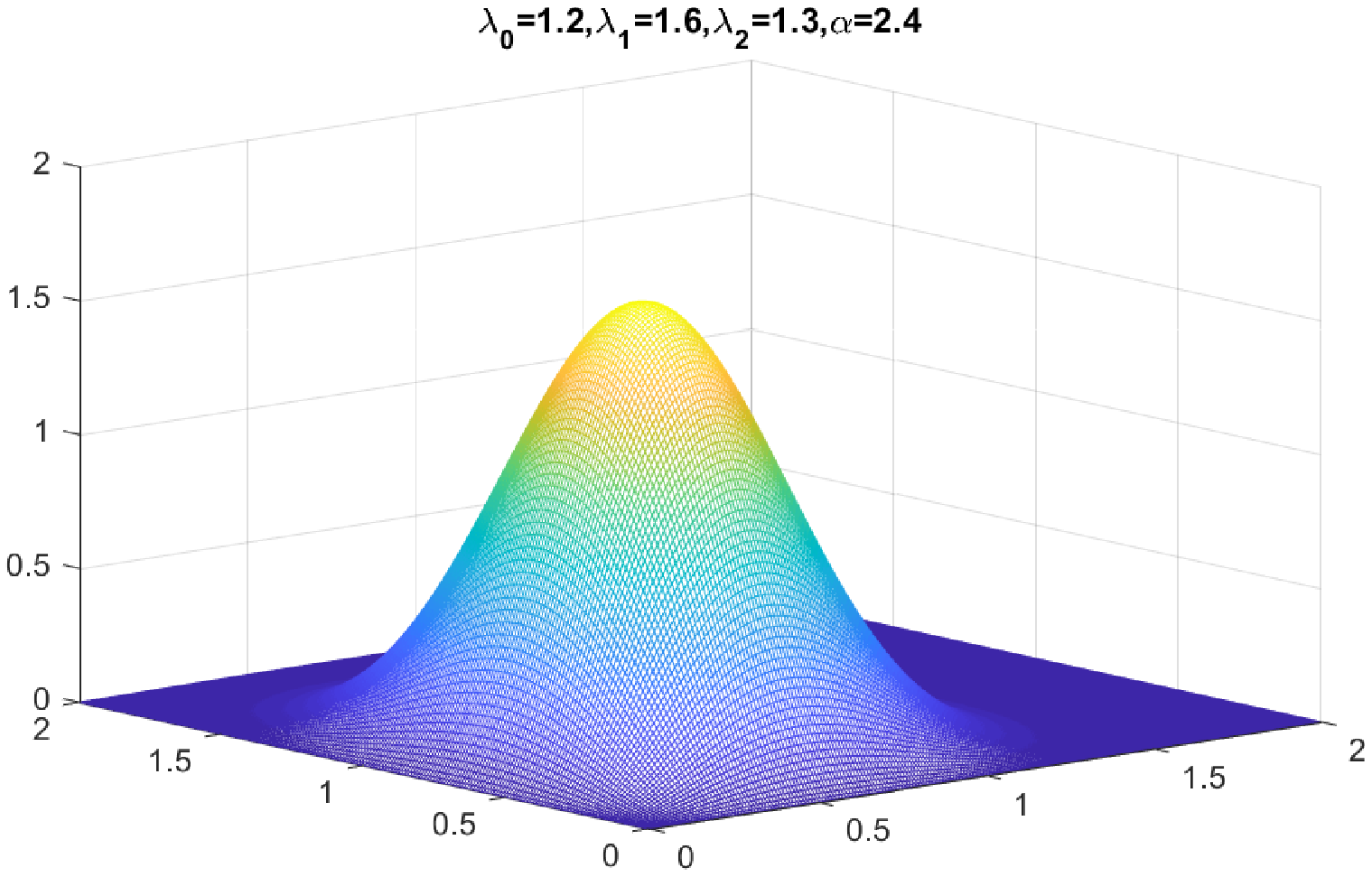}}
		\subfigure[]{\label{c1}\includegraphics[height=1.5in,width=2.2in]{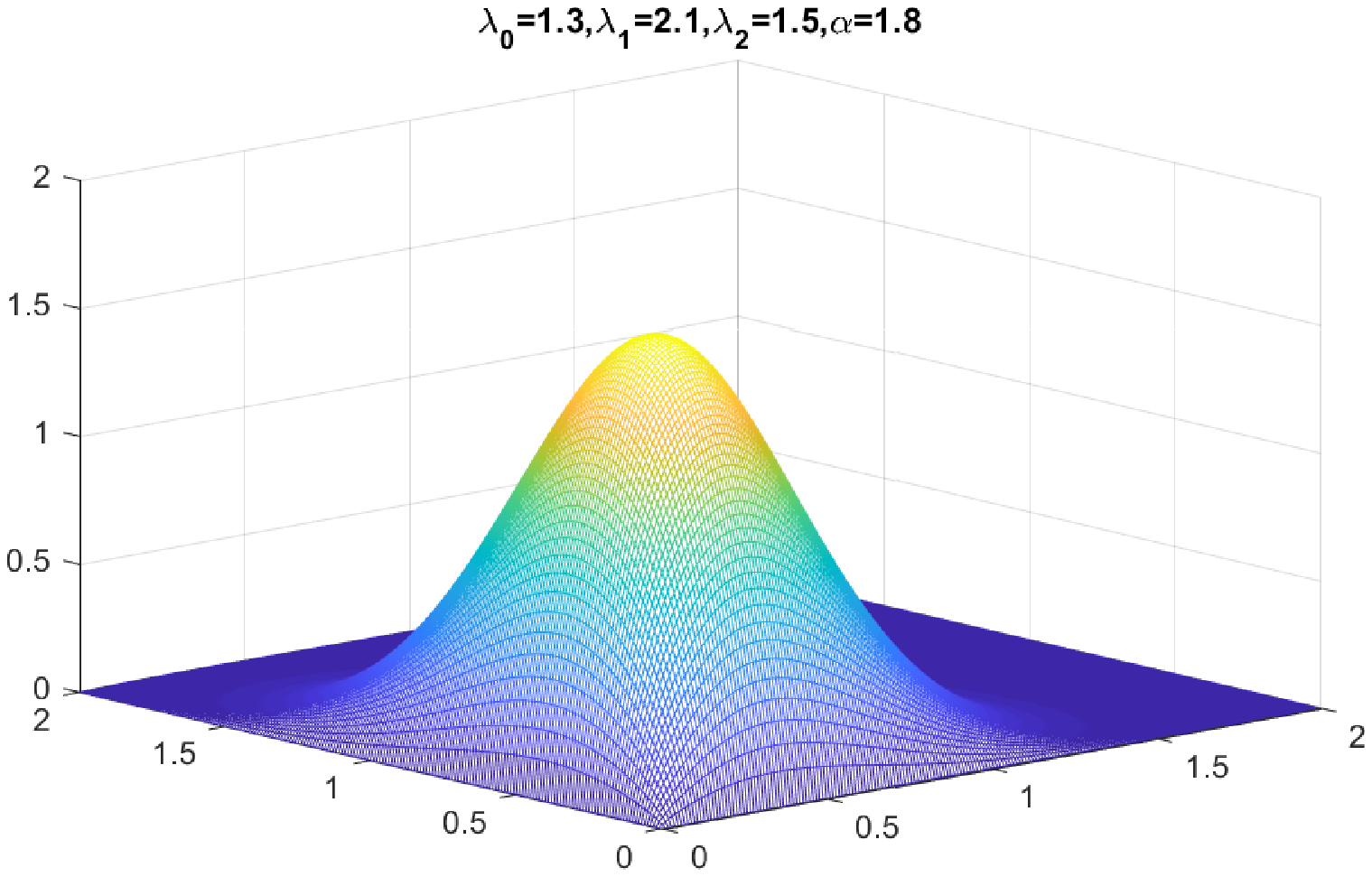}}
		\subfigure[]{\label{c1}\includegraphics[height=1.5in,width=2.2in]{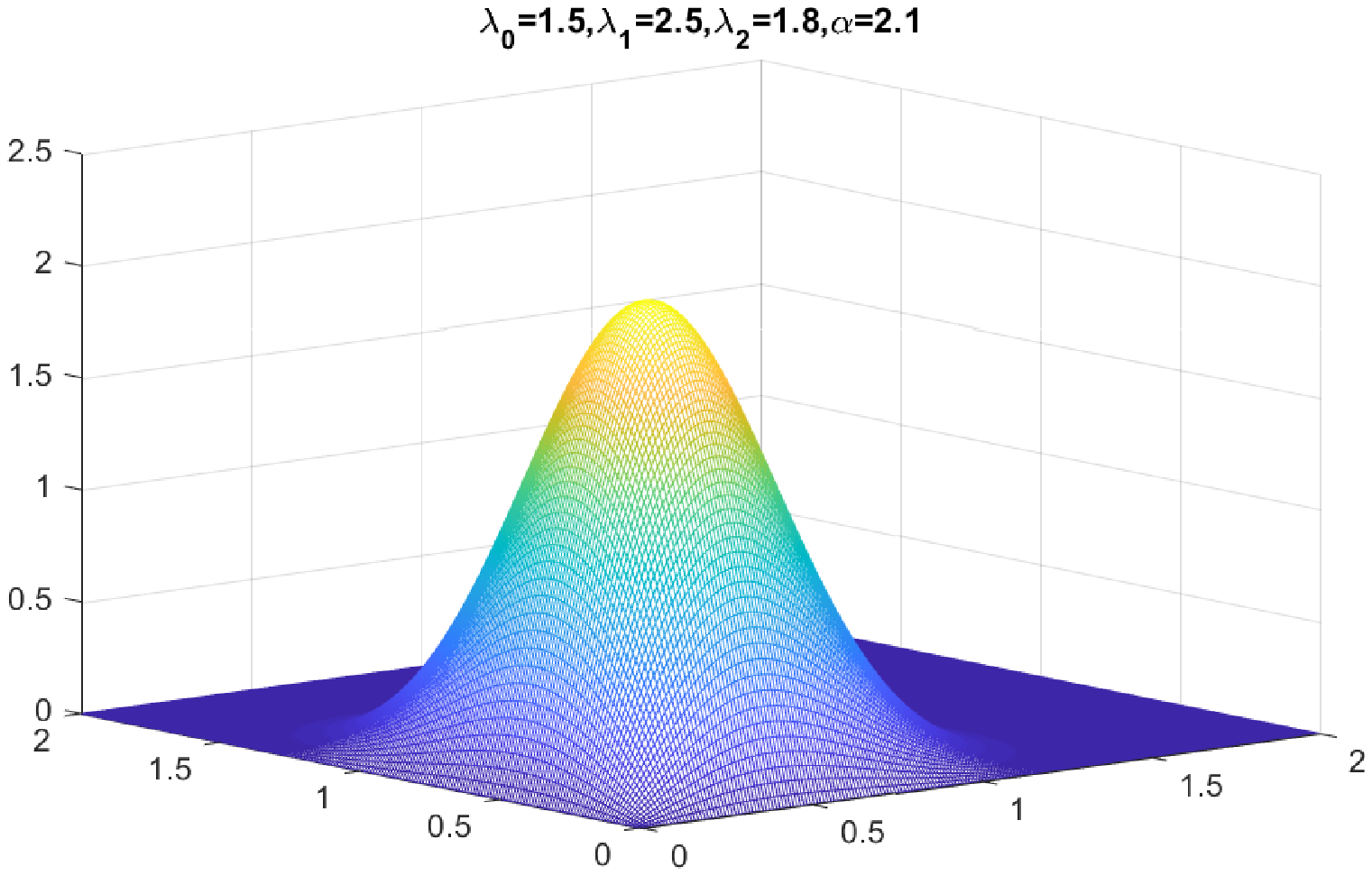}}
		\caption{Surface plots for joint PDF of MOBW. }
	\end{center}
\end{figure}
MOBW is one of the important Marshall-Olkin type distributions due to having singular and absolute continuous parts. The joint PDF $f_{Y_1,Y_2}(y_1,y_2)$ given by $(\ref{2.2})$ yields that densities with two dimensional Lebesgue measure are represented by first two expressions and a density with single dimensional Lebesgue measure is represented by third one. The unimodality property of the bivariate density function of  MOBW distribution can be easily observed from the surface plots given in Fig. $1$. In addition, for $\alpha=1$, MOBW distribution becomes Marshall-Olkin bivariate exonential distribution. Further, MOBW distribution has Weibull models as marginals which have monotone hazard rate functions and decreasing or unimodal PDFs. Due to having such statistical properties  and simpler mathematical structure, analysis of MOBW distribution under AT-II PHCS has been taken into account here may be considered as a proper model of discussion.  
\subsection{Data description}
Suppose that $n$ identical individuals are placed on a life test with lifetime $Y_{i}$ and assume that two causes of failures are known with $Y_i=min\{Y_{1i},Y_{2i}\}$, where $Y_{ki}$ is the latent failure time of $i$-th individual because of $k$-th ($k=1,2$) cause of failure. Then competing risk data under AT-II PHCS can be expressed as
\begin{center}
	\textbf{Case I:} $(y_{1:m:n},\delta_1,R_1), \cdots, (y_{m:m:n},\delta_m,R_m)$, where $y_{m:m:n} <T$~~~~~~~~~~~~~~\\
	\textbf{Case II:} $(y_{1:m:n},\delta_1,R_1), \cdots, (y_{d:m:n},\delta_d,R_d)$, where $y_{d:m:n}<T<y_{d+1:m:n}$,
\end{center}
where $y_{i:m:n}$ is the AT-II PHC sample of $Y_i$ and  the associated failure indicator $\delta_i$ can be expressed as
\begin{align}
	\delta_i = \begin{cases}
		0, ~~\mbox{if}~y_{1i}=y_{2i}\\
		1, ~~\mbox{if}~y_{1i}<y_{2i}\\
		2, ~~\mbox{if}~y_{1i}>y_{2i}\\
		3, ~~\mbox{unknown cause of failure}
	\end{cases} \label{2.3}
\end{align}
and $d$ is a number for which $y_{d:m:n}<T<y_{d+1:m:n}$.

\begin{remark}
	Nowadays the cause of failures can be identified in most cases due to the improvement of quality in research. However, due to complex external and internal interferences, there are some practical situations where each and every failure mode can not be discovered. Therefore, it is more reasonable to consider the unknown cause of failures in the study of the competing risks model.
\end{remark}

\section{Classical estimation}
\subsection{Maximum likelihood estimation}
Suppose that $Y_{i:m:n}$ is the $i$-th failure time of experimental units from a sample of $n$ units which are produced from two Weibull populations that are dependent on each other and have competing risks. For convenience, $y_{i:m:n}$ has been expressed as $y_i$ and the joint PDF is given by
\begin{align}
	\nonumber L(data) = C \prod_{i=0}^{m}& \big[f_{Y_1,Y_2}(y_i,y_i)\big]^{\delta_{i0}} \bigg[-\frac{\partial S_{Y_{1},Y_{2}}(y_1,y_2)}{\partial y_1}\bigg]_{(y_i,y_i)}^{\delta_{i1}} \bigg[-\frac{\partial S_{Y_{1},Y_{2}}(y_1,y_2)}{\partial y_2}\bigg]_{(y_i,y_i)}^{\delta_{i2}}\\
	 \times & f(y_i;\alpha,\lambda_{012})^{\delta_{i3}}	\prod_{i=1}^{J} S_{Y_{1},Y_{2}}(y_1,y_2)^{R_i} S_{Y_{1},Y_{2}}(y_m,y_m)^{R^{*}} \label{3.1}
\end{align}
where $C= \prod_{i=1}^{m}\prod_{j=1}^{m}(R_{j}+1)$, $\delta_{ij}=I(\delta_i=j)(j=0,1,2,3)$ such that $\sum_{j=0}^{3}\delta_{ij}=1$,
\begin{align}
\nonumber	J=\begin{cases}
		m, ~~\mbox{if}~~ y_m <T \\
		d, ~~~\mbox{if}~~ y_d<T<y_{d+1}
	\end{cases}	
~~\mbox{and}~~
	\nonumber	R^{*}=\begin{cases}
		0, ~~~~~~~~~~~~~~~~~~~~~~~\mbox{if}~~ y_m <T \\
		n-m-\sum_{i=1}^{d} R_{i}, ~~\mbox{if}~~ y_d<T<y_{d+1}.
	\end{cases}	
\end{align}
Suppose that $(Y_{i1},Y_{i2})\sim MOBW(\lambda_0,\lambda_1,\lambda_2,\alpha)$, $i=1,\cdots,n$ are independent and identically distributed. Using equations $(\ref{2.2})$ and $(\ref{3.1})$, the likelihood function is expressed as
\begin{align}
	\nonumber L(\lambda_0,\lambda_1,\lambda_2,\alpha)=&~ C \alpha^{m} \prod_{i=1}^{m} \lambda_0^{\delta_{i0}} \lambda_1^{\delta_{i1}} \lambda_2^{\delta_{i2}} \lambda_{012}^{\delta_{i3}}~ y_i^{\alpha-1} e^{-\lambda_{012}\big(\sum_{i=1}^{m}y_i^{\alpha}+\sum_{i=1}^{J}R_i y_i^{\alpha} + R^{*}y_m^{\alpha}\big)} \\
	=&~ C\alpha^{m} \lambda_0^{m_0} \lambda_1^{m_1} \lambda_2^{m_2} \lambda_{012}^{m_3} \prod_{i=1}^{m} y_i^{\alpha-1} e^{-\lambda_{012}\big(\sum_{i=1}^{m}y_i^{\alpha}+\sum_{i=1}^{J}R_i y_i^{\alpha} + R^{*}y_m^{\alpha}\big)}, \label{3.2}
\end{align}
where $ m_j=\sum_{j=0}^{3} \delta_{ij}$ represents the number of failures due to four different cases given in $(\ref{2.3})$ such that $\sum_{j=0}^{3} m_j=m$. Then the log likelihood function can be expressed as
\begin{align}
 \nonumber	\log L = & ~m \log \alpha + m_0 \log \lambda_0 + m_1 \lambda_1 + m_2 \log \lambda_2 + m_3 \log \lambda_{012}\\
	&~ + (\alpha-1) \sum_{i=1}^{m} \log y_i - \lambda_{012}~ A(\alpha), \label{3.3}
\end{align}
where $A(\alpha)= \sum_{i=1}^{m}y_i^{\alpha}+\sum_{i=1}^{J}R_i y_i^{\alpha} + R^{*}y_m^{\alpha}$.
Now, taking derivative of $(\ref{3.3})$ with respect to $\alpha$, $\lambda_0$, $\lambda_1$ and $\lambda_2$ and equating to zero, we get
\begin{align}
	\begin{cases}
	\frac{\partial \log L}{\partial \lambda_0}= &~ \frac{m_0}{\lambda_0} + \frac{m_3}{\lambda_{012}} -A(\alpha) =0,\\
	\frac{\partial \log L}{\partial \lambda_1}= &~ \frac{m_1}{\lambda_1} + \frac{m_3}{\lambda_{012}} -A(\alpha) =0,\\
	\frac{\partial \log L}{\partial \lambda_2}= &~ \frac{m_2}{\lambda_2} + \frac{m_3}{\lambda_{012}} -A(\alpha) =0,\\
	\frac{\partial \log L}{\partial \alpha} =&~ \frac{m}{\alpha} + \sum_{i=1}^{m} \log y_i - \lambda_{012}~ A^{\prime}(\alpha)=0,
\end{cases} \label{3.4}
\end{align}
where $A^{\prime}(\alpha)= \sum_{i=1}^{m} y_i^{\alpha} \log y_i + \sum_{i=1}^{J}R_i y_i^{\alpha} \log y_i + R^{*} y_m^{\alpha} \log y_m$. The MLEs of $\lambda_j$, $j=0,~1,~2$ for known shape parameter $\alpha$, are provided by the subsequent theorem.

\begin{theorem} \label{th3.1}
	Suppose $m_j >0~(j=0,1,2)$. Then the MLE of $\lambda_j$ given $\alpha$ is obtained as
	\begin{align}
		\widehat{\lambda}_j= \frac{m_j}{m_{012}} \frac{m}{A(\alpha)}, ~~~j=0,1,2 \label{3.5}
	\end{align}
where $m_{012}=m_0+m_1+m_2$.
\end{theorem}
\begin{proof}
	See Appendix $A.3$.
\end{proof}
Using $(\ref{3.4})$ and $(\ref{3.5})$, we have
\begin{align}
	\frac{m}{\alpha} + \sum_{i=1}^{m} \log y_i - (\widehat{\lambda}_0+\widehat{\lambda}_1+\widehat{\lambda}_2 )~ A^{\prime}(\alpha)=0. \label{3.11}
\end{align}
Since  $(\ref{3.11})$ can not be solved explicitly, an iteration method as discussed in \cite{kundu2007hybrid} will be employed to obtain the MLE of $\alpha$. Using $(\ref{3.11})$, the MLE $\widehat{\alpha}$ can be obtained as follows
\begin{align}
	\alpha= h(\alpha),\label{3.12}
\end{align}
where
\begin{align}
	\nonumber h(\alpha)= \frac{m}{(\widehat{\lambda}_0+\widehat{\lambda}_1+\widehat{\lambda}_2 )~ A^{\prime}(\alpha)-\sum_{i=1}^{m} \log y_i}.
\end{align}
Assume an initial value for $\alpha$ as $\alpha^{(0)}$. Substitute $\alpha^{(0)}$ in $(\ref{3.12})$ to obtain $\alpha^{(1)}= h(\alpha^{(0)})$. This procedure will be continued until we get an $\alpha^{(k)}$ such that $|\alpha^{(k)}-\alpha^{(k-1)}| < \epsilon$, where $k$ is a positive integer and $\epsilon$ is very small positive real number. This $\alpha^{(k)}$ will be considered as the MLE of $\alpha$. Then replacing $\widehat{\alpha}$ in  $(\ref{3.5})$ we get
\begin{align}
	\nonumber \widehat{\lambda}_j= \frac{m_j}{m_{012}} \frac{m}{A(\widehat{\alpha})}, ~~~j=0,~1,~2 .
\end{align}

\subsection{Approximate confidence interval}
In this subsection, the $100(1-\gamma)\%$ approximate confidence intervals (ACIs) of unknown parameters $\alpha$, $\lambda_0$, $\lambda_1$ and $\lambda_2$ are constructed based on large sample approximation. Under some mild regularity conditions, it can be shown that
\begin{align}
	\begin{pmatrix}
		\widehat{\lambda}_{0} \\
		\widehat{\lambda}_1 \\
		\widehat{\lambda}_2\\
		\widehat{\alpha}\\
	\end{pmatrix}-\begin{pmatrix}
	{\lambda}_{0} \\
	{\lambda}_1 \\
	{\lambda}_2\\
	{\alpha}\\
\end{pmatrix} \rightarrow N\big(0,I^{-1}(\widehat{\lambda}_{0},\widehat{\lambda}_{1},\widehat{\lambda}_{2},\widehat{\alpha})\big), \label{3.13}
\end{align}
where the observed Fisher information matrix $I(\widehat{\lambda}_{0},\widehat{\lambda}_{1},\widehat{\lambda}_{2},\widehat{\alpha})= \bigg[- \frac{\partial^2 \log L}{\partial \Theta_i \partial \Theta_j}\bigg]_{\Theta=\widehat{\Theta}},$ $i,~j=0,~1,~2,~3$ and $\Theta= (\Theta_1,\Theta_2,\Theta_3,\Theta_4)= (\lambda_{0},\lambda_{1},\lambda_{2},\alpha)$.  All these elements $\frac{\partial^2 \log L}{\partial \Theta_i \partial \Theta_j}$ can be expressed as \\
$\frac{\partial^2 \log L}{\partial \Theta_i^2}=-\frac{m_i}{\lambda_{i}^2}-\frac{m_3}{\lambda_{012}^2}$, $\frac{\partial^2 \log L}{\partial \Theta_i \partial \Theta_j}= -\frac{m_3}{\lambda_{012}^2}$, $\frac{\partial^2 \log L}{\partial \Theta_i \partial \Theta_3}= \frac{\partial^2 \log L}{ \partial \Theta_3 \partial \Theta_i}= -A^{\prime}(\alpha)$, for $i,j=0,1,2$ and $i \neq j$; \\
$\frac{\partial^2 \log L}{\partial \Theta_3^2}=-\frac{m}{\alpha^2}-\lambda_{012} A^{\prime\prime}(\alpha)$, where $A^{\prime\prime}(\alpha)=\sum_{i=1}^{m} y_i^{\alpha} (\log y_i)^2 + \sum_{i=1}^{J}R_i y_i^{\alpha} (\log y_i)^2 + R^{*} y_m^{\alpha} (\log y_m)^2$.
For $0<\gamma<1$, the ACIs  of the parameters with significance level $\gamma$  are constructed as
\begin{align}
	 \widehat{\alpha}\underline{+}z_{\gamma/2}\sqrt{Var(\widehat{\alpha})}~~\mbox{and}~~ \widehat{\lambda}_j\underline{+}z_{\gamma/2}\sqrt{Var(\widehat{\lambda}_j)},~\mbox{for}~ j=0,1,2 \label{3.14}
\end{align}
where $z_{\gamma/2}$ is the upper $\gamma/2$-th quantile of the standard normal distribution.
% Since there is no assurance in getting always positive lower bounds for the ACIs given in $(\ref{3.14})$, then the ACIs can be constructed by using delta method and the logarithmic transformation of MLEs such that $\frac{\log \widehat{\Theta}-\log \Theta}{Var(\log \widehat{\Theta})}\rightarrow N(0,1)$. Then the log-transformed $100(1-\gamma)\%$ ACIs are given by
%\begin{align}
%	\nonumber\bigg(&\widehat{\alpha}\cdot exp\big(-\frac{z_{\gamma/2}\sqrt{Var(\widehat{\alpha})}}{\widehat{\alpha}}\big), \widehat{\alpha}\cdot exp\big(\frac{z_{\gamma/2}\sqrt{Var(\widehat{\alpha})}}{\widehat{\alpha}}\big)\bigg)\\
%	\nonumber~\mbox{and}~&\\
%	\nonumber \bigg(&\widehat{\lambda}_j\cdot exp\big(-\frac{z_{\gamma/2}\sqrt{Var(\widehat{\lambda}_j)}}{\widehat{\lambda}_j}\big), \widehat{\lambda}_j\cdot exp\big(\frac{z_{\gamma/2}\sqrt{Var(\widehat{\lambda}_j)}}{\widehat{\lambda}_j}\big)\bigg), ~~\mbox{for}~ j=0,1,2.
%\end{align}
\section{Bayesian estimation}
 In the study of reliability analysis, sometimes it has been observed that the classical estimation through the MLE approach may fail when the data don't provide enough sampling details. In order to overcome this problem, prior information can be used in conjunction with Bayesian analysis. This section deals with Bayesian approach to estimate unknown parameters and the accompanying HPD credible intervals.
\subsection{Prior information}
 In order to obtain Bayes estimates, prior knowledge is crucial, particularly when the data from sampling-based information about the unknown parameters is insufficient. As similar as \cite{kundu2013bayes}, let us consider that $\lambda_{012} \sim $ Gamma $(a,b)$ with PDF as follows
 \begin{align}
 	\pi_{0}(\lambda_{012}|a,b)= \frac{b^a}{\Gamma (a)} \lambda_{012}^{a-1} e^{-b\lambda_{012}} ,~~ a>0,b>0,\lambda_{012}>0. \label{4.1}
 \end{align}
According to \cite{cai2017analysis}, let us consider that for given $\lambda_{012}$, $(\frac{\lambda_{1}}{\lambda_{012}},\frac{\lambda_{2}}{\lambda_{012}})$ follows a Dirichlet prior, say $\pi_{1}(\cdot|d_0,d_1,d_2)$ and that is given by
\begin{align}
	\pi_1(\frac{\lambda_{1}}{\lambda_{012}},\frac{\lambda_{2}}{\lambda_{012}}|\lambda_{012},d_0,d_1,d_2)=\frac{\Gamma (d_{012})}{\Gamma (d_0) \Gamma (d_1) \Gamma (d_2)} \prod_{i=0}^{2} \bigg(\frac{\lambda_{i}}{\lambda_{012}}\bigg)^{d_i}, ~~\lambda_{0}>0,\lambda_{1}>0,\lambda_{2}>0, \label{4.2}
\end{align}
  where $d_0$, $d_1$ and $d_2>0$ are the Dirichlet distribution's parameters with $d_{012}=d_0+d_1+d_2$ and known as the gamma-Dirichlet distribution, which is an extension to higher dimension of a gamma-beta distribution. Through the Jacobian transformation, the joint density function of $\lambda_{0}$, $\lambda_{1}$ and $\lambda_{2}$ can be expressed as
  \begin{align}
  \pi_{1}(\lambda_{0},\lambda_{1},\lambda_{2}) = \frac{\Gamma (d_{012})}{\Gamma (a)} \prod_{i=0}^{2} \frac{b^{d_i}}{\Gamma (d_i)} \lambda_{i}^{d_i-1} e^{-b\lambda_{i}} (\lambda_{012}b)^{a-d_{012}}. \label{4.3}
  \end{align}
  This type of prior distribution can be explained different situations depending on the proper hyperparameter values whether the parameters are independent or not. Hence it becomes highly flexible. When $a=d_{012}$, then the parameters $\lambda_{0}$,  $\lambda_{1}$ and $\lambda_{2}$ in $(\ref{4.3})$ are independent. Since the shape parameter $\alpha$ is not known, then the conjugate prior of the joint prior of $(\lambda_{0},\lambda_{1},\lambda_{2},\alpha)$ never existed and there is no way to presume any particular variation of the prior on $\alpha$, say $\pi_2(\alpha)$. Only the density function $\pi_{2}(\alpha)$ with a non-negative support $(0,\infty)$ may possibly be log-concave. Moreover, the joint prior of $(\lambda_{0},\lambda_{1},\lambda_{2})$ is not affected by $\pi_2(\alpha)$. Based on these assumptions, the joint prior of the model parameters, say $\pi(\lambda_{0},\lambda_{1},\lambda_{2},\alpha)$ is obtained as
  \begin{align}
  \pi(\lambda_{0},\lambda_{1},\lambda_{2},\alpha)=\pi_1(\lambda_{0},\lambda_{1},\lambda_{2})\cdot \pi_2(\alpha), \label{4.4}
  \end{align}
   where it should be noted that, for the purpose of data analysis the choices of the hyperparameters is important  for any specific form of $\pi_2(\alpha)$.

   \subsection{Posterior analysis}
   After combining $(\ref{3.2})$ and $(\ref{4.4})$, it is possible to express the joint posterior density function of $\lambda_{0}$, $\lambda_{1}$, $\lambda_{2}$ and $\alpha$ as 
   \begin{align}
  \nonumber \pi(\Theta|data)&= \frac{L(\Theta;data)\pi(\Theta)}{\int_{0}^{\infty}\cdots\int_{0}^{\infty}L(\Theta;data)\pi(\Theta)~d\Theta}\\
    & \propto~ \alpha^m \bigg(\prod_{j=0}^{2} \lambda_{j}^{m_j+d_j-1}\bigg) \lambda_{012}^{m_3+a-d_{012}} \bigg(\prod_{i=1}^{m}y_{i}^{\alpha-1}\bigg) e^{-\lambda_{012}(b+A(\alpha))}\cdot \pi_2(\alpha).\label{4.5}
   \end{align}
   In order to obtain Bayes estimates, the posterior expected value of a loss function has to be minimized. In such case, the choice of different loss function reflects different estimation error effects. Thus, we choose two different loss functions and these are squared error loss function (SELF) and LINEX loss function (LLF). Here, LLF is asymmetric in nature, whereas SELF is symmetric. The loss functions SELF and LLF are defined as
   \begin{align}
   L_{SE}(W,\widehat{W})=\big(\widehat{W}-W\big)^2,  \label{4.6}
   \end{align}
   and \\
   \begin{align}
   L_{LI}(W,\widehat{W})= e^{p\big(\widehat{W}-W\big)}-p\big(\widehat{W}-W\big)-1,~~p\neq 0,\label{4.7}
   \end{align}
   respectively, where $\widehat{W}$ is the estimate of a parametric function $W$. Under these SELF and LLF, the Bayes estimates of $W(\theta)$ can be obtained as
   \begin{align}
   \widehat{W}_{SE}(\Theta)= \frac{\int_{0}^{\infty}\cdots\int_{0}^{\infty} W(\Theta) \pi(\Theta|data) d\Theta}{\int_{0}^{\infty}\cdots\int_{0}^{\infty}\pi(\Theta|data) d\Theta},\label{4.8}
   \end{align}
   and
   \begin{align}
   \widehat{W}_{LI}(\Theta)= -\bigg(\frac{1}{p}\bigg) \log \bigg(\frac{\int_{0}^{\infty}\cdots\int_{0}^{\infty}e^{-p W(\Theta)}\pi(\Theta|data) d\Theta}{\int_{0}^{\infty}\cdots\int_{0}^{\infty}\pi(\Theta|data) d\Theta}\bigg),\label{4.9}
   \end{align}
   respectively. These above equations $(\ref{4.8})$ and $(\ref{4.9})$ can not be solved explicitly. Thus, the Gibbs sampling technique can be used to generate MCMC samples which can be used to obtain the Bayes estimates and the corresponding HPD credible intervals. The following results are useful to generate the MCMC samples with the help of Gibbs sampling technique. 
   \begin{theorem}\label{th4.1}
   	If $m>1$ and $y_1,\cdots,y_m$ are AT-II PHC competing risks samples, then the marginal posterior density function of $\alpha$ is
   	\begin{align} \label{4.10}
   	\pi_{1}^{*}(\alpha|data) \propto ~ \pi_2(\alpha) \alpha^m \bigg(\prod_{i=1}^{m}y_{i}^{\alpha-1}\bigg) \bigg(b+A(\alpha)\bigg)^{-(a+m)},
   	\end{align}
   	and the joint posterior density of $(\lambda_{0},\lambda_{1},\lambda_{2})$ for the provided $\alpha$ and data, is
   	\begin{align} \label{4.11}
   	\pi_{2}^{*}(\lambda_{0},\lambda_{1},\lambda_{2}|\alpha,data) \propto ~ \lambda_{012}^{m_3+a-d_{012}}\bigg(\prod_{j=0}^{2}Gamma\big(\lambda_{j};m_j+d_j,b+A(\alpha)\big)\bigg),
   	\end{align}
   	which follows gamma-Dirichlet distribution, denoted by $GD(a+m,b+A(\alpha),m_0+d_0,m_1+d_1,m_2+d_2)$.
   \end{theorem}
 \begin{proof}
 	See Appendix $A.4$.
 \end{proof}

  Since it is assumed that $\pi_{2} (\alpha)$ is log-concave, then it is considered as
  \begin{align} \label{4.12}
  	\pi_2(\alpha) \propto \alpha^{a_{1}-1} e^{-b_{1}\alpha},
  \end{align}
   where $a_{1}>0,b_{1}>0$. Now, combining $(\ref{4.10})$ and $(\ref{4.12})$, the marginal posterior density of $\alpha$ is given by
   \begin{align} \label{4.13}
   	\pi_1^{*}(\alpha|data) \propto \alpha^{m+a_1-1} e^{-\alpha\big(b_1-\sum_{i=1}^{m}\log y_i\big)} \big(b+A(\alpha)\big)^{-(a+m)}.
   \end{align}
\begin{theorem}\label{th4.2}
	 The marginal posterior density $\pi_1^{*}(\alpha|data)$ expressed by $(\ref{4.13})$ is log-concave.
	 \end{theorem}
 \begin{proof}
 	See Appendix $A.5$.
 \end{proof}
MCMC samples of the unknown parameters can be generated to compute Bayes estimates and associated HPD credible intervals by using Theorem ${\ref{th4.1}}$ and ${\ref{th4.2}}$. To generate samples the following algorithm is proposed : \\\\
\textbf{Step 1:} Generate $\alpha$ by using the method proposed by \cite{devroye1984simple} from  $(\ref{4.13})$. \\
\textbf{Step 2:} For each given $\alpha$, generate the samples for $(\lambda_{0},\lambda_{1},\lambda_{2})$ from $GD(a+m,b+A(\alpha),m_0+d_0,m_1+d_1,m_2+d_2)$. \\
\textbf{Step 3:} Repeat  Steps 1-2, upto $M$ times to get the samples $\{(\alpha^{1},\lambda_{0}^{1},\lambda_{1}^{1},\lambda_{2}^{1}), \cdots, (\alpha^{M},\lambda_{0}^{M},\lambda_{1}^{M},\lambda_{2}^{M})\}$.\\\\
Using the MCMC samples, Bayes estimates of $W(\Theta)$ under SELF and LLF are expressed as
\begin{align*}
	\widehat{W}_{SE}(\Theta)= \frac{1}{M} \sum_{i=1}^{M} W(\Theta^{i}),
\end{align*}
and
\begin{align*}
	\widehat{W}_{LI}(\Theta)= -\frac{1}{p} \log \bigg(\frac{1}{M} e^{-pW(\Theta^{i})}\bigg),
\end{align*}
respectively, where $\Theta^{i}=(\alpha^{i},\lambda_{0}^{i},\lambda_{1}^{i},\lambda_{2}^{i})$. Replacing $W(\Theta)$ by the parameters, Bayes estimates of the parameters can be obtained.

\subsection {Highest posterior density credible interval}
In order to obtain HPD credible intervals with $\gamma$ level of significance for the model parameters, the following steps are employed using the MCMC samples generated in Section $4.2$. \\\\
\textbf{Step 1:}  Generated MCMC samples of the parameters have been arranged in ascending order as $(\alpha^{(1)},\cdots,\alpha^{(M)})$, $(\lambda_{0}^{(1)},\cdots,\lambda_{0}^{(M)})$, $(\lambda_{1}^{(1)},\cdots,\lambda_{1}^{(M)})$. \\
\textbf{Step 2:} HPD credible intervals of $\Theta$ with $\gamma$ significance level are given by
\begin{align*}
	(\Theta^{(j)},\Theta^{(j+(1-\frac{\gamma}{2})M)}),
\end{align*}
 where $j$ can be obtained from the equation $\Theta^{(j+(1-\frac{\gamma}{2})M)})-\Theta^{(j)}=min_i ~(\Theta^{(i+(1-\frac{\gamma}{2})M)})-\Theta^{(i)})$, where $1 \leq i \leq M-(1-\frac{\gamma}{2})M$.
 
 \subsection{MCMC convergence criterion}
In this subsection, a convergence test suggested by \cite{gelman1992inference} has been used to discuss the convergence of MCMC samples. The following steps have been used to check whether the MCMC samples are convergent or not.\\\\ 
\textbf{Step I:} $M$ number of samples have been generated for the parameters $\Theta=(\theta_1,\theta_2,\theta_3,\theta_4)=(\alpha,\lambda_0,\lambda_1,\lambda_2)$ as $(\theta_1^{1},\cdots,\theta_1^M)$, $(\theta_2^{1},\cdots,\theta_2^M)$, $(\theta_3^{1},\cdots,\theta_3^M)$ and $(\theta_4^{1},\cdots,\theta_4^M)$. Then consider these samples as $M$ number of chains of the parameters as $\{(\theta_1^1,\theta_2^1,\theta_3^1,\theta_4^1), \cdots, (\theta_1^M,\theta_2^M,\theta_3^M,\theta_4^M)\}$ \\
\textbf{Step II:} Calculate the mean for each chain of parameters as 
\begin{align*}
	\widehat{\theta}_j= \frac{1}{4} \sum_{i=1}^{4} \theta_{i}^{j}, 
\end{align*}
where $i\in \{1,2,3,4\}$ and $j \in \{1,\cdots,M\}$. \\ 
\textbf{Step III:} For each chain of parameters, the intra-chain variance can be computed as
\begin{align*}
	\sigma_{j}^2= \frac{1}{3} \sum_{i=1}^{4} \big(\theta_{i}^j-\widehat{\theta}_j\big)^2.
\end{align*} 
\textbf{Step IV:} Compute mean as $\widehat{\theta}= \frac{1}{M}\sum_{j=1}^{M}\widehat{\theta}_j$ and variance of all chains as $W=\frac{1}{M}\sum_{j=1}^{M}\sigma_{j}^2$. \\
\textbf{Step V:} Calculate the variance of the average mean over all chains as
\begin{align*}
	B= \frac{4}{M-1} \sum_{j=1}^{M} \big(\widehat{\theta}_j-\widehat{\theta}\big)^2.
\end{align*} 
\textbf{Step VI:} Calculate $V=\frac{3}{4}W+\frac{M+1}{4M}B$ and $G=\sqrt{V/W}$. \\\\
 Finally, the MCMC sample converges if $G \approx 1$, otherwise not.  
\section{Simulation study}
In this section, some results based on Monte Carlo simulations have been provided to evaluate the performance of the proposed methods given in the preceding sections. Based on the following considerations, the performance of the estimates has been compared.
\begin{itemize}
	\item \textbf{Absolute Bias (AB)}: $\frac{1}{N} \sum_{i=1}^{N}|\Theta_i-\widehat{\Theta}_i|$, where $\Theta_i$ represents parameters, whereas $\widehat{\Theta}_i$ represents their estimates and $N$ is the number of iterations. The lower value of AB suggests that the experimental data and prediction model are more accurately correlated.
	\item \textbf{Mean squared error (MSE)}: $\frac{1}{N} \sum_{i=1}^{N}\big(\Theta_i-\widehat{\Theta}_i\big)^2$. The greater performance of the estimations is indicated by the smaller value of MSE.
	\item \textbf{Average width (AW)}: Average width of $100(1-\gamma)\%$ interval estimates has been assessed. Better performance of the interval estimates is correlated with narrower width.
	\item\textbf{Coverage probability (CP)}: The probability of containing the true values of the parameters in between the interval estimates.
\end{itemize}
 To generate AT-II PHCS with competing risks,
the following algorithm has been constructed.\\	
\\
------------------------------------------------------------------------------------------------------------------------------
\textbf{Algorithm}\\
------------------------------------------------------------------------------------------------------------------------------
1. According to inverse distribution function method min $\{Y_1,Y_2\}\sim MOBW (\lambda_{0},\lambda_{1},\lambda_{2},\alpha)$, then the progressive Type-II censored samples following MOBW distribution have been generated.\\
2. According to the given time threshold $T$, the value of $J$ has been obtained such that $Y_{j:m:n}<T<Y_{j+1:m:n}$. \\
3. Then failure risks $\delta_i= 0,1,2,3$, where $i=1,\cdots,m$ have been assigned to the obtained censored data.\\
------------------------------------------------------------------------------------------------------------------------------
In that case when only partial failure reasons have been observed during an experiment, then it has been considered that each observation is independent of each other with the probability of the unknown reason for each failure as $P(\delta_i=3)=q$. The probability of failure risks can be obtained as 
\begin{align}
	\nonumber P(\delta_i=0)= \frac{\lambda_{0}}{g(\lambda)},~~P(\delta_i=1)= \frac{\lambda_{2}\lambda_{01}}{g(\lambda)}, ~ \mbox{and}~ P(\delta_i=2)= \frac{\lambda_{1}\lambda_{02}}{g(\lambda)},
\end{align}
 where $g(\lambda)= \lambda_{0}+\lambda_{1}\lambda_{02}+\lambda_{2}\lambda_{01}$. 
 
 The Monte Carlo simulation study has been performed with the initial guess of the parameters as $\lambda_{0}=0.5,~ \lambda_{1}=1,~ \lambda_{2}=1.5,$ and $\alpha=1$. Set $(n,m)=(50,30),~ (50,40),~(80,40),~(80,60)$ with different values of $T=(0.5,1)$ and $q=(0.1,0.3)$. Three following censoring schemes are considered here.\\\\
 \textbf{Scheme I:} $R_1=\cdots=R_{m-1}=0, R_m=(n-m)$.\\
 \textbf{Scheme II:} $R_1=(n-m),R_{2}=\cdots=R_{m}=0$.\\
 \textbf{Scheme III:} $R_{2m-n+1}=\cdots=R_{m}=1$, and $R_{i}=0$, others.
\\

  In order to obtain Bayes estimates two different loss functions such as SELF and LLF have been considered. Gamma distribution is thought to be the prior distribution of $\alpha$. Here, all the hyper-parameters are considered as $0.001$. To obtain Bayes estimates, an MCMC algorithm is simulated with 10,000 repetitions. Also the value of the statistic $G$ has been computed to check the convergence of the MCMC sample. Further, $95\%$ ACIs and HPD credible intervals for the unknown model parameters are also constructed. This simulation study has been evaluated using $R~4.0.4$ software via three recommended packages. These packages are `nleqslv' introduced by \cite{hasselman2018package} to compute classical estimates using Newton-Raphson method, `coda' introduced by \cite{plummer2015package} to obtain Bayes estimates using MCMC samples, and `matlib' introduced by \cite{friendly2021package}. The simulation results of point and interval estimates are tabulated in Tables $\ref{T1}$-$\ref{T4}$. From these tables, some observations can be summarized as follows: 
 \begin{itemize}
 	\item With the increase of $(n,m)$ and $T$ when $q$ is fixed, the AB and MSE of point estimates decrease. That is, the estimated results are more accurate when the effective sample size becomes larger.
 \item According to ABs and MSEs, Bayes estimates outperform than MLEs. Further, the performance of Bayes estimates under LLF is superior to that under SELF.
 \item With the increase of $(n,m)$ and $T$ when $q$ is fixed, the AW of the intervals decrease, which yields that the estimation has more accuracy. But no trends of CP have been observed throughout these simulation results.
 \item  HPD intervals perform better than ACIs based on AWs and CPs.
 \item With an increase of $q$, the point and interval estimation results based on AB, MSE, and AW become less accurate since the increasing value of $q$ indicates less available information.
 \item The values of $G$- statistic are almost $1$, which yields that the MCMC samples are convergent.
 \end{itemize}

 From the above results it has been concluded that the performance of point and interval estimates are quite satisfactory. Bayes estimates has been computed based on non-informative prior,and the results of the point estimates based on Bayesian estimation are more effective. For interval estimation one can prefer to HPD intervals over ACIs to get more accurate results.

\begin{table}[htbp!]
	\begin{center}
		\caption{Simulated ABs and MSEs (in parentheses) of the MLE, Bayes estimates under SELF and LLF for different values of $T$ with $q=0.1$.}
		\label{T1}
		\tabcolsep 7pt
		\small
		\scalebox{0.65}{
			\begin{tabular}{*{16}c*{15}{r@{}l}}
				\toprule
				\multicolumn{3}{c}{} &
				\multicolumn{4}{c}{MLE} & \multicolumn{4}{c}{SELF} & \multicolumn{4}{c}{LLF}  \\
				\cmidrule(lr){4-7}\cmidrule(lr){8-11} \cmidrule(lr){12-15}
				\multicolumn{1}{c}{$(n,m)$}& \multicolumn{1}{c}{$T$} & \multicolumn{1}{c}{Scheme} &
				 \multicolumn{1}{c}{$\lambda_0$} & \multicolumn{1}{c}{$\lambda_{1}$}& \multicolumn{1}{c}{$\lambda_{2}$}  & \multicolumn{1}{c}{$\alpha$} & \multicolumn{1}{c}{$\lambda_0$} & \multicolumn{1}{c}{$\lambda_1$}& \multicolumn{1}{c}{$\lambda_{2}$}  & \multicolumn{1}{c}{$\alpha$} & \multicolumn{1}{c}{$\lambda_0$} & \multicolumn{1}{c}{$\lambda_{1}$}& \multicolumn{1}{c}{$\lambda_{2}$}  & \multicolumn{1}{c}{$\alpha$} \\
				\midrule
				(50,30)& 0.5& I& 0.2284& 0.6347& 0.4445& 0.1482& 0.1090& 0.5578& 0.2442& 0.1171& 0.1014& 0.4556& 0.2149& 0.1106  \\
				& & & (0.0605)& (0.8188)& (0.4862)& (0.0382)& (0.0179)& (0.4086)& (0.1119)& (0.0203)& (0.0178)& (0.2752)& (0.1062)& (0.0196) \\
				& & II& 0.2319& 0.4854& 0.3017& 0.1160& 0.1119& 0.4250& 0.2740& 0.1102& 0.1044& 0.3553& 0.2686& 0.1077  \\
				& & & (0.0589)& (0.3670)& (0.1438)& (0.0230)& (0.0186)& (0.2862)& (0.1260)& (0.0203)& (0.0185)& (0.2031)& (0.1092)& (0.0192) \\
				& & III& 0.2345& 0.5686& 0.4120& 0.1360& 0.1118& 0.4593& 0.3269& 0.1138&  0.1043& 0.3778& 0.3077& 0.1116  \\
				& & & (0.0625)& (0.6372)& (0.3660)& (0.0317)& (0.0187)& (0.3545)& (0.1818)& (0.0212)& (0.0181)& (0.2383)& (0.1453)& (0.0202)  \\ [0.2 cm]
				& 1& I & 0.2259& 0.6277& 0.4259& 0.1476& 0.1243& 0.3910& 0.2385& 0.1250& 0.1176& 0.3483& 0.3054& 0.1222  \\
				& & & (0.0547)& (0.6246)& (0.3936)& (0.0358)& (0.0165)& (0.3668)& (0.0951)& (0.0192)& (0.0158)& (0.2407)& (0.0868)& (0.0187)  \\
				& & II& 0.2294& 0.4757& 0.2779& 0.1077& 0.1107& 0.4058& 0.2659& 0.1087& 0.1032& 0.3464& 0.2618& 0.1041  \\
				& & & (0.0551)& (0.3523)& (0.1385)& (0.0222)& (0.0185)& (0.2793)& (0.1261)& (0.0204)& (0.0194)& (0.1969)& (0.1104)& (0.0190) \\
				& & III& 0.2302& 0.5446& 0.3889& 0.1335& 0.1114& 0.4426& 0.3095& 0.1115&  0.1019& 0.3676& 0.2885& 0.1108 \\
				& & & (0.0602)& (0.6102)& (0.3595)& (0.0313)& (0.0186)& (0.3457)& (0.1693)& (0.0209)& (0.0179)& (0.2262)& (0.1314)& (0.0200) \\
				\midrule
				(50,40)& 0.5& I& 0.2098& 0.5228& 0.2813& 0.1373& 0.0999& 0.3203& 0.2123& 0.1160& 0.1010& 0.2756& 0.2043& 0.0986   \\
				& & & (0.0487)& (0.4148)& (0.1361)& (0.0232)& (0.0160)& (0.2243)& (0.1052)& (0.0198)& (0.0154)& (0.2181)& (0.1032)& (0.0193)  \\
				& & II& 0.2089& 0.4689& 0.2443& 0.1046& 0.1092& 0.4160& 0.2384& 0.1013& 0.1017& 0.3373& 0.2256& 0.0991  \\
				& & & (0.0475)& (0.3316)& (0.1036)& (0.0183)& (0.0175)& (0.2822)& (0.0944)& (0.0189)& (0.0184)& (0.2011)& (0.0935)& (0.0181)  \\
				& & III& 0.2245& 0.4630& 0.3309& 0.1188&  0.1053& 0.4261& 0.3060& 0.1073& 0.1018& 0.3671& 0.3010& 0.1061 \\
				& & & (0.0561)& (0.3658)& (0.2720)& (0.0269)& (0.0184)& (0.2925)& (0.1416)& (0.0198)& (0.0180)& (0.2158)& (0.1394)& (0.0191) \\ [0.2 cm]
				& 1& I& 0.2054& 0.5138& 0.2783& 0.1176&  0.0926& 0.3040& 0.1820& 0.0936& 0.1053& 0.3723& 0.1638& 0.0915 \\
				& & & (0.0497)& (0.4268)& (0.1454)& (0.0235)& (0.0158)& (0.2213)& (0.0866)& (0.0176)& (0.0154)& (0.1645)& (0.0848)& (0.0171)  \\
				& & II& 0.2042& 0.4659& 0.2412& 0.1038& 0.1089& 0.3931& 0.2137& 0.1012& 0.1004& 0.3245& 0.2107& 0.0979 \\
				& & & (0.0449)& (0.3275)& (0.0983)& (0.0181)& (0.0174)& (0.2765)& (0.0912)& (0.0180)& (0.0172)& (0.1932)& (0.0903)& (0.0176)  \\
				& & III& 0.2106& 0.4237& 0.2803& 0.1153&  0.1025& 0.4182& 0.2863& 0.1040& 0.1012& 0.3489& 0.2653& 0.1055 \\
				& & & (0.0491)& (0.3458)& (0.2358)& (0.0245)& (0.0175)& (0.2670)& (0.1144)& (0.0194)& (0.0173)& (0.2064)& (0.1121)& (0.0183) \\	
				\midrule
				(80,40)& 0.5& I& 0.2044& 0.5165& 0.2620& 0.1273& 0.0961& 0.2884& 0.1829& 0.1072& 0.0915& 0.2617& 0.1613& 0.0954	\\
				& & & (0.0478)& (0.4081)& (0.1283)& (0.0227)& (0.0152)& (0.2101)& (0.1044)& (0.0182)& (0.0151)& (0.2049)& (0.1017)& (0.0178)  \\
				& & II& 0.2059& 0.4517& 0.2394& 0.0975& 0.1074& 0.4052& 0.2186& 0.0955& 0.1012& 0.3264& 0.2064& 0.0936  \\
				& & & (0.0457)& (0.3225)& (0.1027)& (0.0178)& (0.0170)& (0.2726)& (0.0933)& (0.0176)& (0.0180)& (0.1952)& (0.0916)& (0.0173)  \\
				& & III& 0.2147& 0.4345& 0.3217& 0.1091& 0.1022& 0.4172& 0.2941& 0.0976&  0.0948& 0.3461& 0.2891& 0.0959 \\
				& & & (0.0525)& (0.3449)& (0.2014)& (0.0227)& (0.0181)& (0.2769)& (0.1385)& (0.0183)& (0.0176)& (0.2076)& (0.1288)& (0.0178)  \\ [0.2 cm]
				& 1& I& 0.2014& 0.4957& 0.2534& 0.1148& 0.0915& 0.2573& 0.1678& 0.0917& 0.0902& 0.2371& 0.1591& 0.0912  \\
				& & & (0.0459)& (0.3984)& (0.1225)& (0.0216)& (0.0156)& (0.2052)& (0.0868)& (0.0172)& (0.0151)& (0.1595)& (0.0824)& (0.0169) \\	
				& & II& 0.2016& 0.4442& 0.2318& 0.0960& 0.1030& 0.3876& 0.2116& 0.0940& 0.0986& 0.3180& 0.2050& 0.0934 \\
				& & & (0.0436)& (0.3109)& (0.0929)& (0.0175)& (0.0167)& (0.2671)& (0.0898)& (0.0164)& (0.0168)& (0.1906)& (0.0891)& (0.0165) \\	
				& & III& 0.2082& 0.4065& 0.2781& 0.1080& 0.1031& 0.3869& 0.2697& 0.0961& 0.0915& 0.3304& 0.2532& 0.0957\\
				& & & (0.0486)& (0.3280)& (0.1916)& (0.0198)& (0.0175)& (0.2567)& (0.1083)& (0.0176)& (0.0168)& (0.1987)& (0.1112)& (0.0173)  \\
				\midrule
				(80,60)& 0.5& I& 0.1991& 0.4874& 0.2477& 0.1144& 0.0950& 0.2646& 0.1670& 0.1032& 0.0898& 0.2560& 0.1500& 0.0937  \\
				& & & (0.0432)& (0.3299)& (0.1038)& (0.0218)&  (0.0137)& (0.2032)& (0.0974)& (0.0161)& (0.0142)& (0.1932)& (0.0928)& (0.0157) \\
				& & II& 0.2015& 0.4391& 0.2182& 0.0926&  0.1029& 0.3367& 0.1897& 0.0924& 0.0995& 0.3158& 0.1828& 0.0911 \\
				& & & (0.0430)& (0.2653)& (0.0975)& (0.0172)& (0.0167)& (0.2439)& (0.0889)& (0.0160)& (0.0179)& (0.1921)& (0.0846)& (0.0156) \\
				& & III& 0.2072& 0.4079& 0.3035& 0.0967& 0.0980& 0.3374& 0.2854& 0.0907& 0.0910& 0.3240& 0.2754& 0.0898 \\
				& & & (0.0514)& (0.2819)& (0.1907)& (0.0196)& (0.0175)& (0.2033)& (0.1174)& (0.0173)& (0.0171)& (0.1952)& (0.1116)& (0.0167)  \\ [0.2 cm]
				& 1& I& 0.1979& 0.4622& 0.2389& 0.1035& 0.0861& 0.2412& 0.1356& 0.0902& 0.0867& 0.2292& 0.1470& 0.0877  \\
				& & & (0.0427)& (0.3064)& (0.0947)& (0.0195)& (0.0136)& (0.1941)& (0.0869)& (0.0153)& (0.0133)& (0.1859)& (0.0838)& (0.0148) \\
				& & II& 0.1940& 0.4270& 0.2053& 0.0960& 0.1008& 0.3301& 0.1771& 0.0916& 0.0935& 0.3044& 0.1695& 0.0862 \\
				& & & (0.0424)& (0.2623)& (0.0861)& (0.0162)& (0.0164)& (0.2433)& (0.0670)& (0.0146)& (0.0160)& (0.1825)& (0.0628)& (0.0141)  \\
				& & III& 0.1985& 0.3863& 0.2397& 0.0918& 0.0897& 0.2761& 0.2287& 0.0802& 0.0823& 0.2647& 0.2144& 0.0785 \\
				& & & (0.0428)& (0.2704)& (0.1861)& (0.0189)& (0.0171)& (0.2020)& (0.1028)& (0.0154)& (0.0165)& (0.1911)& (0.0987)& (0.0151)  \\
				
			   \bottomrule
				
			\end{tabular}}
		\end{center}
	\vspace{-0.5cm}
\end{table}

\begin{table}[htbp!]
	\begin{center}
		\caption{Simulated AWs and CPs (in parentheses) of ACIs and HPD credible intervals along with corresponding $G$ values for different values of $T$ when $q=0.1$.}
		\label{T2}
		\tabcolsep 7pt
		\small
		\scalebox{0.75}{
			\begin{tabular}{*{12}c*{11}{r@{}l}}
				\toprule
				\multicolumn{3}{c}{} &
				\multicolumn{4}{c}{ACI} & \multicolumn{4}{c}{HPD}  \\
				\cmidrule(lr){4-7}\cmidrule(lr){8-11}
				\multicolumn{1}{c}{$(n,m)$}& \multicolumn{1}{c}{$T$} & \multicolumn{1}{c}{Scheme} &
				\multicolumn{1}{c}{$\lambda_0$} & \multicolumn{1}{c}{$\lambda_{1}$}& \multicolumn{1}{c}{$\lambda_{2}$}  & \multicolumn{1}{c}{$\alpha$} & \multicolumn{1}{c}{$\lambda_0$} & \multicolumn{1}{c}{$\lambda_1$}& \multicolumn{1}{c}{$\lambda_{2}$}  & \multicolumn{1}{c}{$\alpha$} & \multicolumn{1}{c}{$G$} \\
				\midrule
				(50,30)& 0.5& I& 0.8257& 2.5102& 2.5598& 0.6901& 0.2916& 1.5883& 1.5579& 0.4513 & 0.9864 \\
				& & & (0.8574)& (0.9399)& (0.9507)& (0.9516)& (0.9030)& (0.9541)& (0.9735)& (0.9752) \\
				& & II& 0.7068& 1.7502& 1.7670& 0.5555& 0.2918& 1.3613& 1.3650& 0.4364 & 0.9892 \\
				& & & (0.8646)& (0.9397)& (0.9409)& (0.9527)& (0.8923)& (0.9483)& (0.9672)& (0.9621)  \\
				& & III& 0.7657& 2.1742& 2.2108& 0.6089& 0.2938& 1.4627& 1.4585& 0.4311 & 0.9925 \\
				& & & (0.8154)& (0.9393)& (0.9412)& (0.9375)& (0.8925)& (0.9531)& (0.9632)& (0.9562) \\ [0.2 cm]
				& 1& I& 0.8224& 2.5012& 2.5497& 0.6889& 0.2841& 1.2439& 1.2842& 0.4373 & 0.9951\\
				& & & (0.8603)& (0.9402)& (0.9514)& (0.9535)& (0.9089)& (0.9685)& (0.9776)& (0.9794)  \\
				& & II& 0.7009& 1.7460& 1.7306& 0.5459& 0.2892& 1.3608& 1.3584& 0.4324 & 0.9963\\
				& & & (0.8436)& (0.9378)& (0.9436)& (0.9524)& (0.8913)& (0.9259)& (0.9646)& (0.9628) \\
				& & III& 0.7413& 2.0253& 2.1649& 0.5980&  0.2905& 1.4471& 1.4389& 0.4276 & 0.9927\\
				& & & (0.8587)& (0.9425)& (0.9378)& (0.9497)& (0.8912)& (0.9583)& (0.9592)& (0.9582)  \\
				\midrule
				(50,40)& 0.5& I& 0.6621& 1.7138& 1.6679& 0.5569&  0.2780& 1.3737& 1.2928& 0.4365 & 1.0102\\
				& & & (0.8795)& (0.9349)& (0.9438)& (0.9533)& (0.9117)& (0.9614)& (0.9625)& (0.9848)  \\
				& & II& 0.6313& 1.4962& 1.4572& 0.4942& 0.2869& 1.2400& 1.2026&  0.4056 & 0.9975\\
				& & & (0.8789)& (0.9380)& (0.9401)& (0.9504)& (0.8927)& (0.9577)& (0.9606)& (0.9655) \\
				& & III& 0.6236& 1.5862& 1.5432& 0.5229& 0.2935& 1.2601& 1.2235& 0.4090 & 0.9938\\
				& & & (0.8609)& (0.9391)& (0.9322)& (0.9142)& (0.9088)& (0.9504)& (0.9629)& (0.9583)   \\ [0.2 cm]
				& 1& I& 0.6430& 1.6976& 1.6523& 0.5451& 0.2726& 1.2782& 1.2770& 0.4056 & 1.0423 \\
				& & & (0.8694)& (0.9345)& (0.9625)& (0.9495)& (0.8907)& (0.9473)& (0.9884)& (0.9798) \\
				& & II& 0.6301& 1.4898& 1.4351& 0.4923& 0.2843& 1.2347& 1.2008& 0.4042 & 0.9942 \\
				& & & (0.8783)& (0.9385)& (0.9488)& (0.9495)& (0.8944)& (0.9443)& (0.9644)& (0.9651)  \\
				& & III& 0.6201& 1.5631& 1.5290& 0.5126& 0.2869& 1.2395& 1.2101& 0.4017 & 1.0028\\
				& & & (0.8735)& (0.9409)& (0.9426)& (0.9514)&  (0.8962)& (0.9672)& (0.9591)& (0.9583) \\
				\midrule
				(80,40)& 0.5& I& 0.6543& 1.6950& 1.6208& 0.5319& 0.2672& 1.2802& 1.2762& 0.3559 & 0.9880\\
				& & & (0.8621)& (0.9438)& (0.9480)& (0.9532)& (0.9084)& (0.9542)& (0.9623)& (0.9695)  \\
				& & II& 0.6287& 1.4871& 1.4485& 0.4711& 0.2864& 1.2333& 1.2019& 0.3897 & 0.9973 \\
				& & & (0.8615)& (0.9169)& (0.9453)& (0.9538)& (0.8975)& (0.9532)& (0.9609)& (0.9712)  \\
				& & III& 0.6272& 1.5765& 1.5460& 0.4960& 0.2907& 1.2534& 1.2105& 0.3780& 0.9905 \\
				& & & (0.8735)& (0.9372)& (0.9072)& (0.9332)& (0.8925)& (0.9676)& (0.9351)& (0.9623) \\ [0.2 cm]
				& 1& I& 0.6306& 1.6841& 1.6118& 0.5298& 0.2346& 1.1070& 1.1245& 0.3260 & 0.9961 \\
				& & & (0.8705)& (0.9480)& (0.9431)& (0.9530)& (0.9059)& (0.9897)& (0.9776)& (0.9830) \\
				& & II& 0.6247& 1.4847& 1.4402& 0.4695& 0.2781& 1.2197& 1.1974& 0.3829 & 1.0214\\
				& & & (0.8721)& (0.9317)& (0.9457)& (0.9497)& (0.8929)& (0.9584)& (0.9627)& (0.9686) \\
				& & III& 0.6167& 1.5581& 1.5155& 0.4844& 0.2853& 1.2370& 1.2072& 0.3704 & 1.0325\\
				& & & (0.8731)& (0.9373)& (0.9567)& (0.9499)& (0.9016)& (0.9609)& (0.9772)& (0.9616)  \\
				\midrule
				(80,60)& 0.5& I& 0.5458& 1.3901& 1.3323& 0.4551& 0.2673& 1.1507& 1.1418& 0.3583 & 1.0403\\
				& & & (0.8054)& (0.9465)& (0.9379)& (0.9524)& (0.9023)& (0.9745)& (0.9698)& (0.9889) \\
				& & II& 0.5145& 1.1850& 1.1413& 0.3953& 0.2808& 1.0497& 1.0054& 0.3424 & 0.9945\\
				& & & (0.8040)& (0.8675)& (0.9430)& (0.9508)& (0.8854)& (0.9299)& (0.9605)& (0.9733) \\
				& & III& 0.5056& 1.2538& 1.2033& 0.4179& 0.2880& 1.0839& 1.0395& 0.3476 & 1.0210\\
				& & & (0.8567)& (0.9421)& (0.9067)& (0.9057)& (0.8691)& (0.9654)& (0.9599)& (0.9629) \\ [0.2 cm]
				& 1& I& 0.5381& 1.3862& 1.3280& 0.4517& 0.2579& 1.1462& 1.1191& 0.3524 & 0.9978\\
				& & & (0.8154)& (0.9416)& (0.9442)& (0.9510)& (0.8948)& (0.9536)& (0.9794)& (0.9806) \\
				& & II& 0.5064& 1.1789& 1.1369& 0.3910& 0.2781& 1.0464& 1.0008& 0.3407 & 0.9953 \\
				& & & (0.8652)& (0.9049)& (0.9512)& (0.9497)& (0.8835)& (0.9559)& (0.9731)& (0.9775)   \\
				& & III& 0.5025& 1.2368& 1.2018& 0.4159& 0.2811& 1.0569& 1.0136& 0.3373 & 1.0089\\
				& & & (0.8126)& (0.9317)& (0.9411)& (0.9456)&  (0.8921)& (0.9556)& (0.9653)& (0.9663) \\
				
				 \bottomrule
				
			\end{tabular}}
		\end{center}
	\vspace{-0.5cm}
\end{table}

\begin{table}[htbp!]
	\begin{center}
		\caption{Simulated ABs and MSEs (in parentheses) of the MLE, Bayes estimates under SELF and LLF for different values of $T$ with $q=0.3$.}
		\label{T3}
		\tabcolsep 7pt
		\small
		\scalebox{0.65}{
			\begin{tabular}{*{16}c*{15}{r@{}l}}
				\toprule
				\multicolumn{3}{c}{} &
				\multicolumn{4}{c}{MLE} & \multicolumn{4}{c}{SELF} & \multicolumn{4}{c}{LLF}  \\
				\cmidrule(lr){4-7}\cmidrule(lr){8-11} \cmidrule(lr){12-15}
				\multicolumn{1}{c}{$(n,m)$}& \multicolumn{1}{c}{$T$} & \multicolumn{1}{c}{Scheme} &
				\multicolumn{1}{c}{$\lambda_0$} & \multicolumn{1}{c}{$\lambda_{1}$}& \multicolumn{1}{c}{$\lambda_{2}$}  & \multicolumn{1}{c}{$\alpha$} & \multicolumn{1}{c}{$\lambda_0$} & \multicolumn{1}{c}{$\lambda_1$}& \multicolumn{1}{c}{$\lambda_{2}$}  & \multicolumn{1}{c}{$\alpha$} & \multicolumn{1}{c}{$\lambda_0$} & \multicolumn{1}{c}{$\lambda_{1}$}& \multicolumn{1}{c}{$\lambda_{2}$}  & \multicolumn{1}{c}{$\alpha$} \\
				\midrule
				(50,30)& 0.5& I& 0.2309& 0.6494& 0.4589& 0.1495& 0.1387& 0.5705& 0.2984& 0.1282& 0.1332& 0.5167& 0.2551& 0.1214  \\
				& & & (0.0649)& (0.8394)& (0.5048)& (0.0391)& (0.0290)& (0.4389)& (0.1841)& (0.0235)& (0.0280)& (0.3760)& (0.1701)& (0.0216) \\
				& & II& 0.2408& 0.5043& 0.3132& 0.1259& 0.1246& 0.4494& 0.2867& 0.1221& 0.1218& 0.4084& 0.2789& 0.1113 \\
				& & & (0.0617)& (0.3891)& (0.1559)& (0.0252)& (0.0237)& (0.3047)& (0.1295)& (0.0214)& (0.0239)& (0.2394)& (0.1158)& (0.0204) \\
				& & III& 0.2459& 0.5719& 0.4255& 0.1452&  0.1313& 0.5080& 0.3850& 0.1239& 0.1131& 0.4723& 0.3172& 0.1166  \\
				& & & (0.0650)& (0.6559)& (0.3820)& (0.0332)& (0.0198)& (0.4242)& (0.3088)& (0.0242)& (0.0193)& (0.4072)& (0.1679)& (0.0219) \\  [0.2 cm]
				& 1& I& 0.2320& 0.6365& 0.4382& 0.1491& 0.1171& 0.4780& 0.3657& 0.1207& 0.1094& 0.4363& 0.3221&  0.1175 \\
				& & & (0.0580)& (0.7425)& (0.4486)& (0.0382)& (0.0234)& (0.3637)& (0.1888)& (0.0276)& (0.0212)& (0.2470)& (0.1673)& (0.0259)  \\
				& & II& 0.2378& 0.4858& 0.2808& 0.1177& 0.1186& 0.4251& 0.2925& 0.1127& 0.1101& 0.4052& 0.2684& 0.1086   \\
				& & & (0.0581)& (0.3763)& (0.1442)& (0.0237)& (0.0189)& (0.3027)& (0.1308)& (0.0219)& (0.0184)& (0.2890)& (0.1201)& (0.0206)  \\
				& & III& 0.2320& 0.5629& 0.3943& 0.1361& 0.1222& 0.4855& 0.3114& 0.1165& 0.1161& 0.3805& 0.3019& 0.1129   \\
				& & & (0.0630)& (0.6195)& (0.3608)& (0.0328)& (0.0252)& (0.3558)& (0.1762)& (0.0215)& (0.0238)& (0.3041)& (0.1496)& (0.0206)  \\
				\midrule
				(50,40)& 0.5& I& 0.2162& 0.5477& 0.2881& 0.1383& 0.1164& 0.3440& 0.2290& 0.1218& 0.1191& 0.2952& 0.2109& 0.1083  \\
				& & &  (0.0518)& (0.4387)& (0.1474)& (0.0237)& (0.0202)& (0.2511)& (0.1128)& (0.0212)& (0.0192)& (0.2298)& (0.1089)& (0.0205)   \\
				& & II& 0.2120& 0.4790& 0.2543& 0.1092& 0.1229& 0.4335& 0.2483& 0.1047&  0.1213& 0.3633& 0.2366& 0.1011 \\
				& & & (0.0489)& (0.3482)& (0.1059)& (0.0211)& (0.0215)& (0.3033)& (0.1235)& (0.0202)& (0.0207)& (0.2566)& (0.1154)& (0.0197)  \\
				& & III& 0.2286& 0.4868& 0.3513& 0.1289& 0.1230& 0.4488& 0.3404& 0.1163& 0.1205& 0.3739& 0.3349& 0.1125  \\
				& & & (0.0582)& (0.3826)& (0.2748)& (0.0268)& (0.0216)& (0.3160)& (0.1615)& (0.0206)& (0.0208)& (0.2393)& (0.1368)& (0.0197)  \\ [0.2 cm]
				& 1& I& 0.2194& 0.5265& 0.2854& 0.1221& 0.1207& 0.3227& 0.2011& 0.1161&  0.1169& 0.2959& 0.1873& 0.1126 \\
				& & & (0.0535)& (0.4489)& (0.1494)& (0.0241)& (0.0209)& (0.2641)& (0.1290)& (0.0203)& (0.0196)& (0.2422)& (0.1210)& (0.0193) \\
				& & II& 0.2200& 0.4807& 0.2526& 0.1142& 0.1294& 0.4207& 0.2332& 0.1009& 0.1238& 0.3931& 0.2343& 0.0985  \\
				& & & (0.0526)& (0.3401)& (0.1061)& (0.0199)& (0.0274)& (0.2809)& (0.0995)& (0.0186)& (0.0262)& (0.2369)& (0.0969)& (0.0180)   \\
				& & III& 0.2162& 0.4977& 0.2881& 0.1183& 0.1164& 0.4040& 0.2790& 0.1062& 0.1121& 0.3725& 0.2709& 0.1083  \\
				& & & (0.0518)& (0.3847)& (0.2474)& (0.0257)& (0.0202)& (0.3151)& (0.1428)& (0.0182)& (0.0195)& (0.2487)& (0.1389)& (0.0188) \\
				\midrule
				(80,40)& 0.5& I& 0.2140& 0.5848& 0.3923& 0.1261& 0.1150& 0.3033& 0.1959& 0.1021& 0.1091& 0.2922& 0.1807& 0.0995   \\
				& & & (0.0532)& (0.4333)& (0.1416)& (0.0267)& (0.0252)& (0.2494)& (0.1238)& (0.0215)& (0.0219)& (0.2244)& (0.1159)& (0.0188)  \\
				& & II& 0.2115& 0.4623& 0.2428& 0.0986& 0.1121& 0.4137& 0.2384& 0.0961& 0.1054& 0.3615& 0.2106& 0.0949 \\
				& & & (0.0486)& (0.3395)& (0.1081)& (0.0187)& (0.0212)& (0.2951)& (0.1037)& (0.0179)& (0.0205)& (0.2347)& (0.1008)& (0.0175)  \\
				& & III& 0.2214& 0.4943& 0.3426& 0.1125& 0.1185& 0.4279& 0.3082& 0.1017& 0.1121& 0.4026& 0.2923& 0.0984 \\
				& & & (0.0554)& (0.3691)& (0.2106)& (0.0235)& (0.0208)& (0.2937)& (0.1418)& (0.0211)& (0.0189)& (0.2168)& (0.1357)& (0.0191)  \\ [0.2 cm]
				& 1& I& 0.2143& 0.5918& 0.2996& 0.1262& 0.1082& 0.2830& 0.1783& 0.1013&  0.1027& 0.2461& 0.1694& 0.0955 \\
				& & & (0.0534)& (0.4302)& (0.1371)& (0.0231)& (0.0196)& (0.2218)& (0.1059)& (0.0189)& (0.0184)& (0.1837)& (0.0974)& (0.0182)  \\
				& & II& 0.2223& 0.4507& 0.2419& 0.1218& 0.1137& 0.4098& 0.2249& 0.1048& 0.1089& 0.3552& 0.2113& 0.0985   \\
				& & & (0.0495)& (0.3293)& (0.1051)& (0.0191)& (0.0183)& (0.2752)& (0.0938)& (0.0173)& (0.0176)& (0.2250)& (0.0902)& (0.0169)  \\
				& & III& 0.2191& 0.4751& 0.3055& 0.1145& 0.1108& 0.4115& 0.2811& 0.1049& 0.1047& 0.3863& 0.2653& 0.0993  \\
				& & & (0.0517)& (0.3560)& (0.2187)& (0.0206)& (0.0215)& (0.2893)& (0.1210)& (0.0184)& (0.0205)& (0.2439)& (0.1158)& (0.0179)  \\
				\midrule
				(80,60)& 0.5& I& 0.2212& 0.6078& 0.4124& 0.1343& 0.1309& 0.3315& 0.2060& 0.1113& 0.1238& 0.3176& 0.1962& 0.1078  \\
				& & & (0.0542)& (0.4427)& (0.1634)& (0.0286)& (0.0261)& (0.2819)& (0.1325)& (0.0229)& (0.0240)& (0.2430)& (0.1270)& (0.0208) \\
				& & II& 0.2069& 0.4424& 0.2577& 0.1281& 0.1161& 0.3696& 0.1930& 0.1091& 0.1090& 0.3237& 0.1875& 0.0989   \\
				& & & (0.0452)& (0.2943)& (0.1036)& (0.0194)& (0.0191)& (0.2634)& (0.0928)& (0.0178)& (0.0185)& (0.2218)& (0.0899)& (0.0171)  \\
				& & III& 0.2193& 0.4270& 0.3197& 0.0981& 0.1039& 0.3645& 0.2890& 0.0945& 0.0984& 0.3317& 0.2805& 0.0921 \\
				& & & (0.0525)& (0.3055)& (0.2108)& (0.0215)& (0.0195)& (0.2431)& (0.1325)& (0.0196)& (0.0182)& (0.2039)& (0.1250)& (0.0189)  \\ [0.2 cm]
				& 1& I& 0.2007& 0.4699& 0.2414& 0.1095& 0.1016& 0.2736& 0.1406& 0.0982& 0.0956& 0.2385& 0.1524& 0.0917 \\
				& & & (0.0459)& (0.3183)& (0.1094)& (0.0207)& (0.0169)& (0.2108)& (0.0985)& (0.0179)& (0.0154)& (0.2017)& (0.0934)& (0.0168) \\
				& & II& 0.2067& 0.4385& 0.2276& 0.0985& 0.1043& 0.3583& 0.1923& 0.0933& 0.0986& 0.3212& 0.1794& 0.0910  \\
				& & & (0.0450)& (0.2883)& (0.1025)& (0.0184)& (0.0199)& (0.2551)& (0.0821)& (0.0169)& (0.0186)& (0.2105)& (0.0783)& (0.0159)  \\
				& & III& 0.2009& 0.4282& 0.2652& 0.0964& 0.1031& 0.3106& 0.2360& 0.0912& 0.0966& 0.2883& 0.2254& 0.0869 \\
				& & & (0.0440)& (0.3014)& (0.1954)& (0.0207)& (0.0186)& (0.2118)& (0.1095)& (0.0186)& (0.0178)& (0.1994)& (0.1023)& (0.0175)  \\

				\bottomrule
					\end{tabular}}
			\end{center}
		\vspace{-0.5cm}
	\end{table}

\begin{table}[htbp!]
	\begin{center}
		\caption{Simulated AWs and CPs (in parentheses) of of ACIs and HPD credible intervals along with corresponding $G$ values for different values of $T$ when $q=0.3$.}
		\label{T4}
		\tabcolsep 7pt
		\small
		\scalebox{0.75}{
			\begin{tabular}{*{12}c*{11}{r@{}l}}
				\toprule
				\multicolumn{3}{c}{} &
				\multicolumn{4}{c}{ACI} & \multicolumn{4}{c}{HPD}  \\
				\cmidrule(lr){4-7}\cmidrule(lr){8-11}
				\multicolumn{1}{c}{$(n,m)$}& \multicolumn{1}{c}{$T$} & \multicolumn{1}{c}{Scheme} &
				\multicolumn{1}{c}{$\lambda_0$} & \multicolumn{1}{c}{$\lambda_{1}$}& \multicolumn{1}{c}{$\lambda_{2}$}  & \multicolumn{1}{c}{$\alpha$} & \multicolumn{1}{c}{$\lambda_0$} & \multicolumn{1}{c}{$\lambda_1$}& \multicolumn{1}{c}{$\lambda_{2}$}  & \multicolumn{1}{c}{$\alpha$}& \multicolumn{1}{c}{$G$}  \\
				\midrule
				(50,30)& 0.5& I& 0.9562& 2.5978& 2.6565& 0.6979& 0.3434& 1.6724& 1.6901& 0.4808  & 0.9893\\
				& & & (0.8527)& (0.9489)& (0.9374)& (0.9498)& (0.8986)& (0.9693)& (0.9529)& (0.9802) \\
				& & II& 0.7909& 1.8249& 1.8625& 0.5857& 0.3186& 1.4035& 1.4543& 0.5166  & 0.9931\\
				& & & (0.8585)& (0.9495)& (0.9487)& (0.9530)& (0.8743)& (1.0000)& (0.9795)& (0.9885)  \\
				& & III& 0.8859& 2.2575& 2.3049& 0.6290&  0.3071& 1.8277& 1.8676& 0.4478 & 1.0047\\
				& & & (0.8244)& (0.9494)& (0.9495)& (0.9394)& (0.8974)& (0.9763)& (0.9690)& (0.9715)  \\ [0.2 cm]
				& 1& I& 0.9466& 2.5678& 2.6272& 0.6900& 0.3137& 1.6128& 1.6698& 0.4589  & 1.0225\\
				& & & (0.8578)& (0.9499)& (0.9578)& (0.9565)&  (0.8858)& (0.9728)& (0.9805)& (0.9872) \\
				& & II& 0.7891& 1.7606& 1.7934& 0.5559& 0.3062& 1.3845& 1.4140& 0.5026  & 1.0239\\
				& & & (0.8798)& (0.9397)& (0.9472)& (0.9487)& (0.9066)& (0.9691)& (0.9851)& (0.9899) \\
				& & III& 0.8638& 2.1073& 2.1547& 0.6166&  0.2903& 1.2284& 1.2525& 0.4259 & 0.9959\\
				& & & (0.8654)& (0.9460)& (0.9403)& (0.9480)& (0.8973)& (0.9691)& (0.9707)& (0.9694) \\
				\midrule
				(50,40)& 0.5& I& 0.7263& 1.7879& 1.7905& 0.5775&  0.3151& 1.3893& 1.4403& 0.4453 & 1.0187\\
				& & & (0.8199)& (0.9496)& (0.9417)& (0.9506)& (0.8937)& (0.9708)& (0.9695)& (0.9786) \\
				& & II& 0.6813& 1.5456& 1.5306& 0.4989& 0.3141& 1.3024& 1.2532& 0.4202 & 1.0306\\
				& & & (0.8368)& (0.9390)& (0.9473)& (0.9514)& (0.8732)& (0.9501)& (0.9796)& (1.0000)  \\
				& & III& 0.6835& 1.6496& 1.6481& 0.5277& 0.3158& 1.3038& 1.3123& 0.4214 & 0.9915\\
				& & & (0.8204)& (0.9375)& (0.9409)& (0.9501)& (0.8776)& (0.9610)& (0.9649)& (0.9764) \\ [0.2 cm]
				& 1& I& 0.7286& 1.7961& 1.8036& 0.5870& 0.3123& 1.4514& 1.4383& 0.4506  & 0.9930\\
				& & & (0.8140)& (0.9392)& (0.9477)& (0.9517)& (0.8749)& (0.9591)& (0.9603)& (0.9699)  \\
				& & II& 0.6956& 1.5847& 1.5856& 0.4983& 0.3358& 1.2633& 1.2520& 0.4181 & 1.0355\\
				& & & (0.8348)& (0.9417)& (0.9476)& (0.9492)& (0.8723)& (0.9705)& (0.9721)& (0.9799) \\
				& & III& 0.7263& 1.7879& 1.7905& 0.5575&  0.3151& 1.3893& 1.4403& 0.4153 & 0.9961\\
				& & & (0.8199)& (0.9496)& (0.9417)& (0.9506)& (0.8937)& (0.9608)& (0.9695)& (0.9786) \\ [0.2 cm]
				\midrule
				(80,40)& 0.5& I& 0.8104& 2.3216& 2.3345& 0.6007& 0.3266& 1.3214& 1.2846& 0.4021 & 0.9893\\
				& & & (0.8294)& (0.9459)& (0.9472)& (0.9567)& (0.8538)& (0.9697)& (0.9687)& (0.9798)  \\
				& & II& 0.6828& 1.5502& 1.5361& 0.4812& 0.3141& 1.2960& 1.2824& 0.4046 & 0.9947\\
				& & & (0.8492)& (0.9486)& (0.9517)& (0.9512)& (0.8745)& (0.9699)& (0.9668)& (0.9805)  \\
				& & III& 0.7256& 1.8515& 1.8544& 0.5150&  0.3123& 1.7063& 1.6556& 0.4317 & 0.9981\\
				& & & (0.8199)& (0.9494)& (0.9327)& (0.9461)& (0.8874)& (0.9583)& (0.9647)& (0.9762)  \\ [0.2 cm]
				& 1& I& 0.7140& 2.3356& 2.3471& 0.6010& 0.2862& 1.1717& 1.2151& 0.4178 & 1.0414\\
				& & & (0.8197)& (0.9509)& (0.9567)& (0.9535)& (0.8843)& (0.9714)& (0.9695)& (0.9772) \\
				& & II& 0.6910& 1.5728& 1.5745& 0.4722& 0.3010& 1.2718& 1.2769& 0.4190  & 1.0228\\
				& & & (0.8278)& (0.9391)& (0.9499)& (0.9518)& (0.8825)& (0.9644)& (0.9598)& (0.9750)  \\
				& & III& 0.6713& 1.6960& 1.7040& 0.5028& 0.3247& 1.3028& 1.2559& 0.3881 & 0.9937\\
				& & & (0.8171)& (0.9447)& (0.9541)& (0.9496)& (0.8717)& (0.9670)& (0.9705)& (0.9645)  \\
				\midrule
				(80,60)& 0.5& I& 0.6053& 1.4517& 1.4229& 0.4651& 0.2821& 1.2037& 1.2129& 0.3659 & 0.9915\\
				& & & (0.8374)& (0.9473)& (0.9593)& (0.9541)& (0.8601)& (0.9683)& (0.9789)& (0.9894)  \\
				& & II& 0.5551& 1.2368& 1.2126& 0.4165& 0.3023& 1.2418& 1.1750& 0.3711 & 0.9837 \\
				& & & (0.8476)& (0.9315)& (0.9528)& (0.9508)& (0.8783)& (0.9604)& (0.9760)& (0.9718)  \\
				& & III& 0.5598& 1.3175& 1.2921& 0.4284& 0.2831& 1.1426& 1.0765& 0.3685 & 0.9945 \\
				& & & (0.8196)& (0.9456)& (0.9523)& (0.9518)& (0.8995)& (0.9672)& (0.9863)& (0.9905)  \\ [0.2 cm]
				& 1& I& 0.6070& 1.4553& 1.4261& 0.4549& 0.2823& 1.1215& 1.0667& 0.3543  & 1.0124 \\
				& & & (0.8289)& (0.9453)& (0.9513)& (0.9503)& (0.9055)& (0.9641)& (0.9787)& (0.9706) \\
				& & II& 0.5730& 1.2584& 1.2374& 0.4060& 0.2937& 1.1039& 1.0783& 0.3519 & 1.0309\\
				& & & (0.8357)& (0.9414)& (0.9539)& (0.9508)& (0.9032)& (0.9624)& (0.9759)& (0.9730)   \\
				& & III& 0.6032& 1.4234& 1.3960& 0.4371& 0.3034& 1.1356& 1.0987& 0.3423 & 0.9959\\
				& & & (0.8481)& (0.9477)& (0.9523)& (0.9561)& (0.8926)& (0.9678)& (0.9889)& (1.0000)  \\
				
				\bottomrule
			\end{tabular}}
		\end{center}
	\vspace{-0.5cm}
\end{table}

\section{Optimal censoring scheme}
In reliability analysis, to obtain sufficient information about unknown parameters, a practitioner seeks out the optimal censoring plan within a class of all feasible chosen schemes. In our study, in order to get enough details about model parameters, choosing ideal time $T$ with a progressive censoring scheme $(R_1,\cdots, R_D)$, where $D=m$ for case-I and $D=J$ for case-II, for prefixed values of $n,m$, may be of interest to an experimenter. However, comparing two (or more) different censoring plans has gained a lot of attention in past few years by several authors. For instance one can see \cite{ng2004optimal}, \cite{kundu2008bayesian} and \cite{singh2015estimating}.  Here, three commonly used criteria have been considered to compare between the chosen class of progressive censoring schemes, see Table $\ref{T5}$.\\
\begin{table}[htbp!]
	\begin{center}
		\caption{Different optimality criterion.}
		\label{T5}
		\tabcolsep 7pt
		\small
		\scalebox{0.8}{
			\begin{tabular}{*{3}c*{2}{r@{}l}}
				\toprule
				\multicolumn{1}{c}{Criterion}&&& & &\multicolumn{1}{c}{Goal}  \\
				\midrule
				A-optimality &&& & & minimum trace $(I^{-1}(\widehat{\Theta}))$&\\
				D-optimality &&& & &minimum det $(I^{-1}(\widehat{\Theta}))$~~~& \\
				F-optimality &&&&& maximum  trace $(I(\widehat{\Theta}))$~~~&\\
				\bottomrule
			\end{tabular}}
		\end{center}
	\vspace{-0.5cm}
\end{table}
 $A$-optimality and $D$-optimality have been considered as criterion I and II, respectively. These two criteria are widely used in the statistical literature and intend to minimize the trace and determinant of the variance-covariance (V-C) matrix based on the observed Fisher information matrix associated with $\widehat{\Theta}$, respectively. The trace of the V-C matrix equals to the sum of the diagonal elements of $I^{-1}(\widehat{\Theta})$. According to $F$-optimality, which has been considered as criterion III, the trace of the observed Fisher information matrix has been minimized. The related optimal progressive censoring plans have been taken into consideration based on the minimum values of the $A$- and $D$- optimality criteria, and maximum value of the $F$-optimality criterion (see Table $9$).

\section{Real data analysis}
To illustrate the usefulness of the proposed methods, a real life soccer game data set from \cite{meintanis2007test} has been considered in this section. This data set represents that at least one goal was scored by home team, and at least one goal was scored directly by a kick (i.e., penalty, foul, or any other kick) by any team. Let $Y_1$ be the time of the first goal scored by any team, and $Y_2$ be the time of the first goal scored by the home team. Since the soccer game is 90 minutes long, the data is divided by 90. The complete failure data with the observed failure risks $\delta_{i}$ is tabulated in Table $\ref{T6}$.

\begin{table}[htbp!]
	\begin{center}
		\caption{Complete failure data with corresponding risks of the soccer game.}
		\label{T6}
		\tabcolsep 7pt
		\small
		\scalebox{0.8}{
			\begin{tabular}{*{13}c*{11}{r@{}l}}
				\toprule
				\mbox{Original soccer game data ($Y_1$,$Y_2$)}\\
				\midrule
				(0.0222, 0.0222)~~~(0.4556, 0.0333)~~~(0.4667, 0.0333)~~~(0.6000, 0.0778)~~~(0.0889, 0.0889)~~~(0.2778, 0.1000)\\
				(0.6111, 0.1222)~~~(0.4889, 0.1444)~~~(0.2444, 0.1556)~~~(0.2778, 0.1556)~~~(0.7111, 0.1667)~~~(0.1778, 0.1778)\\
				(0.1778, 0.8333)~~~(0.7000, 0.2000)~~~(0.2000, 0.2000)~~~(0.2111, 0.2111)~~~(0.2889, 0.2222)~~~(0.2667, 0.2667)\\
				(0.2889, 0.5333)~~~(0.3000, 0.5222)~~~(0.3111, 0.3111)~~~(0.5667, 0.3111)~~~(0.4889, 0.3333)~~~(0.3778, 0.3778)\\
				(0.4000, 0.5778)~~~(0.4333, 0.4333)~~~(0.5889, 0.4333)~~~(0.4444, 0.4444)~~~(0.4667, 0.4667)~~~(0.9111, 0.5333)\\
				(0.5444, 0.5444)~~~(0.5444, 0.5444)~~~(0.7333, 0.6889)~~~(0.8444, 0.7111)~~~(0.7333, 0.9444)~~~(0.7667, 0.7889)\\
				(0.8000, 0.8000)~~~~~~~~~~~~~~~~~~~~~~~~~~~~~~~~~~~~~~~~~~~~~~~~~~~~~~~~~~~~~~~~~~~~~~~~~~~~~~~~~~~~~~~~~~~~~~~~~~~~~~~~~~~~~~~~~~~~~~~ \\
				\midrule
				Competing risks data ($Y_i$,$\delta_{i}$)\\
				\midrule
				(0.0222,0)~~~~~~~~~~(0.0333,2)~~~~~~~~~~(0.0333,2)~~~~~~~~~~(0.0778,2)~~~~~~~~~~(0.0889,0)~~~~~~~~~~(0.1000,2)\\
				(0.1222,2)~~~~~~~~~~(0.1444,2)~~~~~~~~~~(0.1556,2)~~~~~~~~~~(0.1556,2)~~~~~~~~~~(0.1667,2)~~~~~~~~~~(0.1778,0)\\
				(0.1778,1)~~~~~~~~~~(0.2000,2)~~~~~~~~~~(0.2000,0)~~~~~~~~~~(0.2111,0)~~~~~~~~~~(0.2222,2)~~~~~~~~~~(0.2667,0)\\
				(0.2889,1)~~~~~~~~~~(0.3000,1)~~~~~~~~~~(0.3111,0)~~~~~~~~~~(0.3111,2)~~~~~~~~~~(0.3333,2)~~~~~~~~~~(0.3778,0)\\
				(0.4000,1)~~~~~~~~~~(0.4333,0)~~~~~~~~~~(0.4333,2)~~~~~~~~~~(0.4444,0)~~~~~~~~~~(0.4667,0)~~~~~~~~~~(0.5333,2)\\
				(0.5444,0)~~~~~~~~~~(0.5444,0)~~~~~~~~~~(0.6889,2)~~~~~~~~~~(0.7111,2)~~~~~~~~~~(0.7333,1)~~~~~~~~~~(0.7667,1)\\
				(0.8000,0)~~~~~~~~~~~~~~~~~~~~~~~~~~~~~~~~~~~~~~~~~~~~~~~~~~~~~~~~~~~~~~~~~~~~~~~~~~~~~~~~~~~~~~~~~~~~~~~~~~~~~~~~~~~~~~~~~~~~\\		
				\bottomrule
		\end{tabular}}
	\end{center}
	\vspace{-0.5cm}
\end{table}

\begin{table}[htbp!]
	\begin{center}
		\caption{ Three censoring schemes for AT-II PHC data with competing risks are used as an illustration.}
		\label{T7}
		\tabcolsep 7pt
		\small
		\scalebox{0.8}{
			\begin{tabular}{*{13}c*{11}{r@{}l}}
				\toprule
				\mbox{Data I: $n = 37$, $m = 28$, $T= 0.4$ and $R=(0*27,9)$}\\
				\midrule
				(0.0222,0)~~~~~~~~~~(0.0333,2)~~~~~~~~~~(0.0333,2)~~~~~~~~~~(0.0778,2)~~~~~~~~~~(0.0889,0)~~~~~~~~~~(0.1000,2)\\ (0.1222,2)~~~~~~~~~~(0.1444,2)~~~~~~~~~~(0.1556,2)~~~~~~~~~~(0.1556,2)~~~~~~~~~~(0.1667,2)~~~~~~~~~~(0.1778,0)\\ (0.1778,1)~~~~~~~~~~(0.2000,2)~~~~~~~~~~(0.2000,0)~~~~~~~~~~(0.2111,0)~~~~~~~~~~(0.2222,2)~~~~~~~~~~(0.2667,0)\\ (0.2889,1)~~~~~~~~~~(0.3000,1)~~~~~~~~~~(0.3111,0)~~~~~~~~~~(0.3111,2)~~~~~~~~~~(0.3333,2)~~~~~~~~~~(0.3778,0)\\
				(0.4000,1)~~~~~~~~~~(0.4333,0)~~~~~~~~~~(0.4333,2)~~~~~~~~~~(0.4444,0)~~~~~~~~~~~~~~~~~~~~~~~~~~~~~~~~~~~~~~~~~~~~~~~\\
				\midrule
				\mbox{Data II: $n = 37$,  $m = 28$,  $T = 0.4$ and $R=(9,0*27)$}\\
				\midrule
				(0.0222,0)~~~~~~~~~~(0.0333,2)~~~~~~~~~~(0.0333,2)~~~~~~~~~~(0.0778,0)~~~~~~~~~~(0.1000,2)~~~~~~~~~~(0.1222,2)\\ (0.1444,2)~~~~~~~~~~(0.1556,2)~~~~~~~~~~(0.1667,2)~~~~~~~~~~(0.1778,1)~~~~~~~~~~(0.2000,2)~~~~~~~~~~(0.2000,0)\\
				(0.2111,0)~~~~~~~~~~(0.2222,2)~~~~~~~~~~(0.2667,0)~~~~~~~~~~(0.2889,1)~~~~~~~~~~(0.3000,1)~~~~~~~~~~(0.3111,0)\\
			    (0.3111,2)~~~~~~~~~~(0.3333,2)~~~~~~~~~~(0.4333,0)~~~~~~~~~~(0.4333,2)~~~~~~~~~~(0.4444,0)~~~~~~~~~~(0.4667,0)\\
				(0.5333,2)~~~~~~~~~~(0.5444,0)~~~~~~~~~~(0.6889,2)~~~~~~~~~~(0.7111,2)~~~~~~~~~~~~~~~~~~~~~~~~~~~~~~~~~~~~~~~~~~~~~~~\\
				\midrule
				\mbox{Data III: $n = 37$,  $m = 28$,  $T= 0.4$ and $R=(1*9,0*19)$}\\
				\midrule
				(0.0333,2)~~~~~~~~~~(0.0333,2)~~~~~~~~~~(0.0778,0)~~~~~~~~~~(0.1000,2)~~~~~~~~~~(0.1222,2)~~~~~~~~~~(0.1444,2)\\
				(0.1556,2)~~~~~~~~~~(0.1667,2)~~~~~~~~~~(0.1778,1)~~~~~~~~~~(0.2000,2)~~~~~~~~~~(0.2000,0)~~~~~~~~~~(0.2111,0)\\
				(0.2222,2)~~~~~~~~~~(0.2667,0)~~~~~~~~~~(0.2889,1)~~~~~~~~~~(0.3000,1)~~~~~~~~~~(0.3111,0)~~~~~~~~~~(0.3111,2)\\
				(0.3333,2)~~~~~~~~~~(0.4333,0)~~~~~~~~~~(0.4333,2)~~~~~~~~~~(0.4444,0)~~~~~~~~~~(0.4667,0)~~~~~~~~~~(0.5333,2)\\
				(0.5444,0)~~~~~~~~~~(0.6889,2)~~~~~~~~~~(0.7111,2)~~~~~~~~~~(0.7333,1)~~~~~~~~~~~~~~~~~~~~~~~~~~~~~~~~~~~~~~~~~~~~~~~\\
				\bottomrule
		\end{tabular}}
	\end{center}
	\vspace{-0.5cm}
\end{table}

\begin{table}[htbp!]
	\begin{center}
		\caption{Point and interval estimates of the unknown parameters based on generated samples from the real life data set.}
		\label{T8}
		\tabcolsep 7pt
		\small
		\scalebox{0.9}{
			\begin{tabular}{*{12}c*{11}{r@{}l}}
				\toprule
				\multicolumn{1}{c}{Data} & \multicolumn{1}{c}{$\Theta$}&
				\multicolumn{1}{c}{MLE}& \multicolumn{1}{c}{ACI} & \multicolumn{1}{c}{Bayes}& \multicolumn{1}{c}{HPD}  \\
				\midrule
				I& $\lambda_{0}$& 1.2739& (0.4204,2.5274)[2.1070]& 1.5833& (1.2779,1.8847)[0.6068] \\
				& $\lambda_{1}$& 0.5159& (0.0317,1.1037)[1.0720]& 0.6482& (0.5628,0.7054)[0.1426] \\
				& $\lambda_{2}$& 2.0098& (0.7108,3.3089)[2.5981]& 2.2199& (1.7334,2.7056)[0.9722]   \\
				& $\alpha$& 1.3846& (0.9609,1.8082)[0.8473]& 1.5704& (1.3746,1.8544)[0.4798]  \\	
				\midrule
				II& $\lambda_{0}$& 1.5688& (0.4711,2.6665)[2.1954]& 1.9522& (1.5872,2.3524)[0.7652]  \\
				& $\lambda_{1}$& 0.6275& (0.0203,1.2753)[1.2550]& 0.5812& (0.5253,0.6364)[0.1111] \\
				& $\lambda_{2}$& 2.5101& (1.0347,3.9856)[2.9509]& 2.3674& (1.8621,2.8271)[0.9650] \\
				& $\alpha$&  1.4685&  (1.0578,1.8793)[0.8215]& 1.5056& (1.3367,1.6168)[0.2801]   \\
				\midrule
				III& $\lambda_{0}$& 1.5100& (0.4596,2.5605)[2.1009]& 1.5457& (1.4224,1.6723)[0.2499] \\
				& $\lambda_{1}$& 0.6040& (0.0179,1.2260)[1.2081]& 0.4944& (0.2626,0.6696)[0.4070]    \\
				& $\lambda_{2}$& 2.4161& (1.0076,3.8246)[2.8170]& 2.5246& (2.1011,2.9496)[0.8485]  \\
				& $\alpha$&  1.6001& (1.1654,2.0346)[0.8692]& 1.6663& (1.4724,1.8089)[0.3365]      \\
			\bottomrule
			\end{tabular}}
		\end{center}
	\vspace{-0.5cm}
\end{table}

\noindent Prior to conducting data analysis, we examine how well this set of data can be fitted with the MOBW distribution using the Kolmogorov-Smirnov (K-S) statistic. Since there is no such
goodness-of-fit test for bivariate distributions as univariate distributions, here we have investigated whether the marginals $Y_1$, $Y_2$ and $min(Y_1,Y_2)$ fit Weibull distribution or not. The K-S distances and the corresponding p-values (given in bracket) of  $Y_1$, $Y_2$ and $min(Y_1,Y_2)$ are $0.0834$ $(0.9602)$, $0.1055$ $(0.8051)$ and $0.0689$ $(0.9947)$, respectively. Further, the empirical CDF (ECDF), probability-probability (P-P) and quantile-
quantile (Q-Q) plots along with associated fitted Weibull distribution are provided in Figures $2-4$. From these observations, we conclude that this data set can be considered to model Weibull distribution for the marginals $Y_1$, $Y_2$ and their minimum. This yields that MOBW distribution can be used for analyzing this bivariate data set given in Table $\ref{T6}$.

\begin{figure}[htbp!]
	\begin{center}
		\subfigure[]{\label{c1}\includegraphics[height=1.5in,width=2.2in]{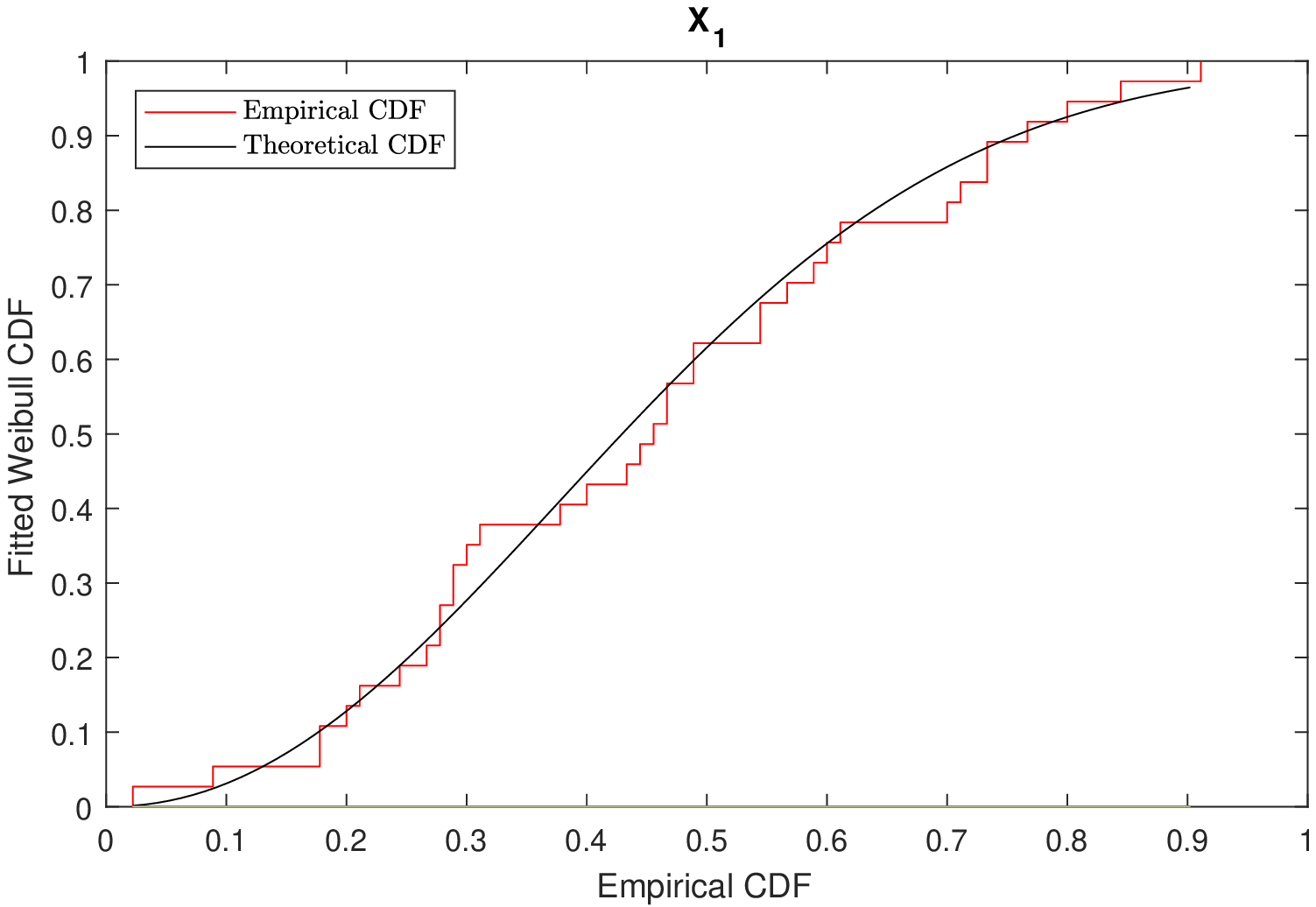}}
		\subfigure[]{\label{c1}\includegraphics[height=1.5in,width=2.2in]{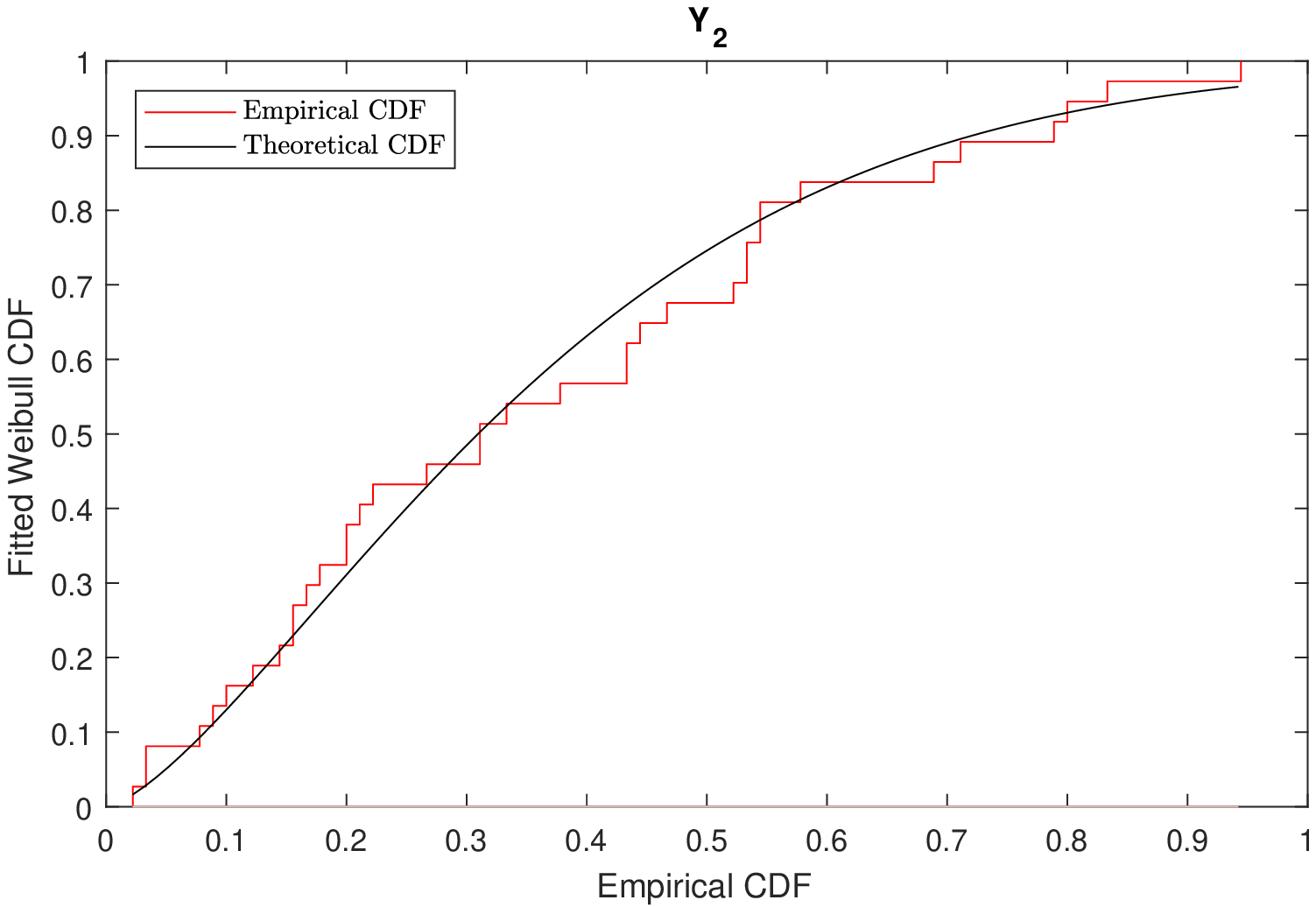}}
		\subfigure[]{\label{c1}\includegraphics[height=1.5in,width=2.2in]{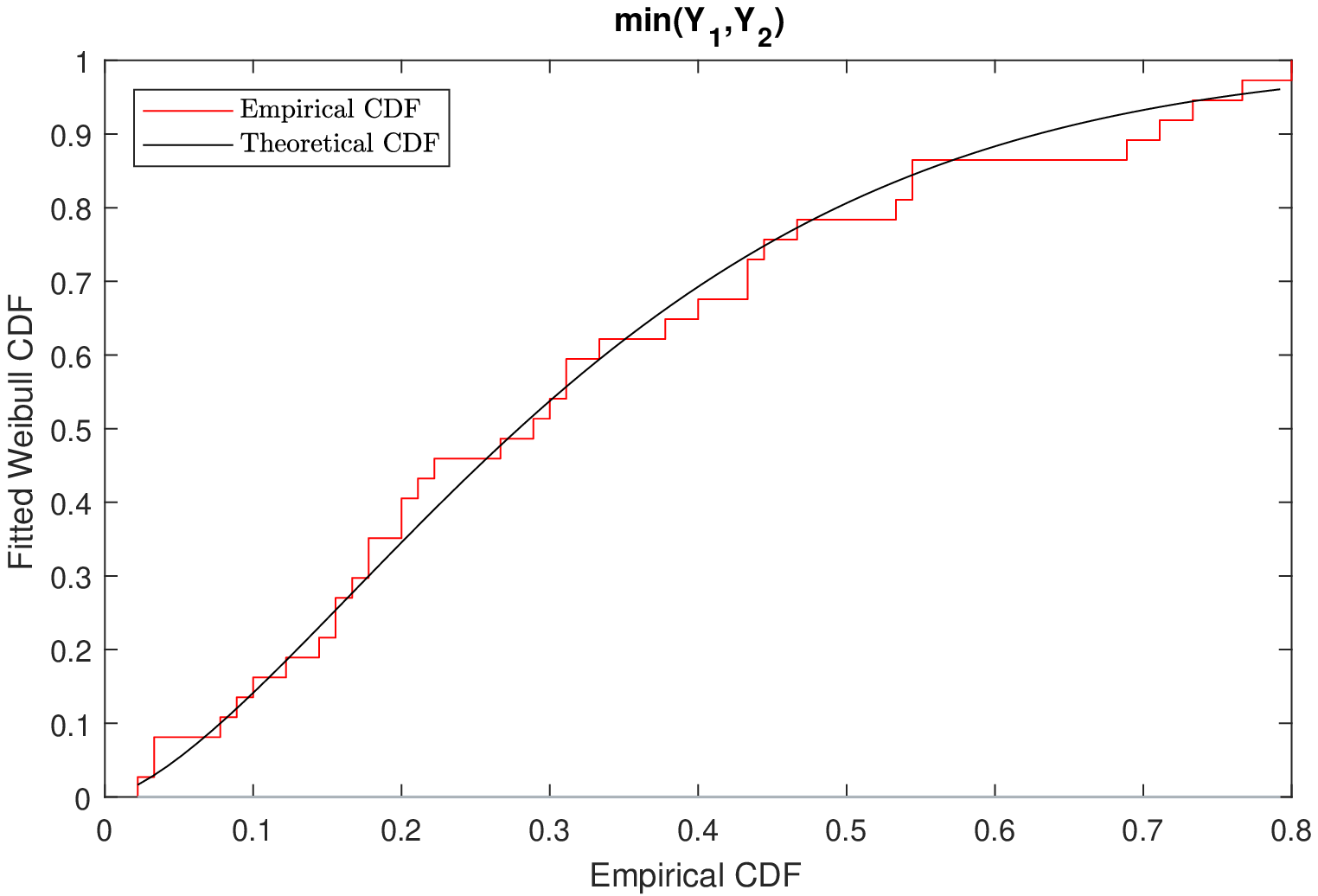}}
		\caption{ECDF plots along with fitted Weibull models for real data set.}
	\end{center}
\end{figure}

\begin{figure}[htbp!]
	\begin{center}
		\subfigure[]{\label{c1}\includegraphics[height=1.5in,width=2.2in]{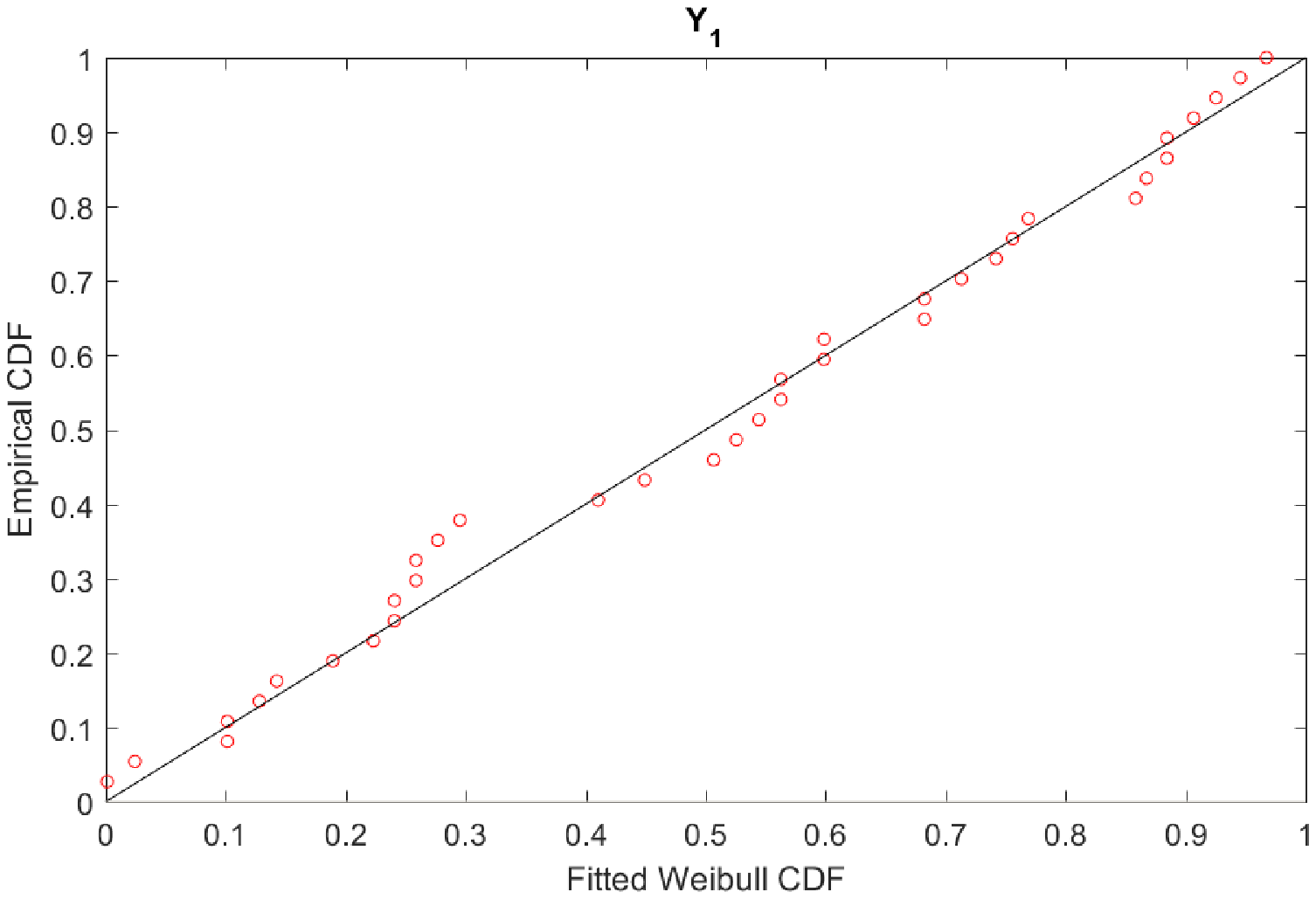}}
		\subfigure[]{\label{c1}\includegraphics[height=1.5in,width=2.2in]{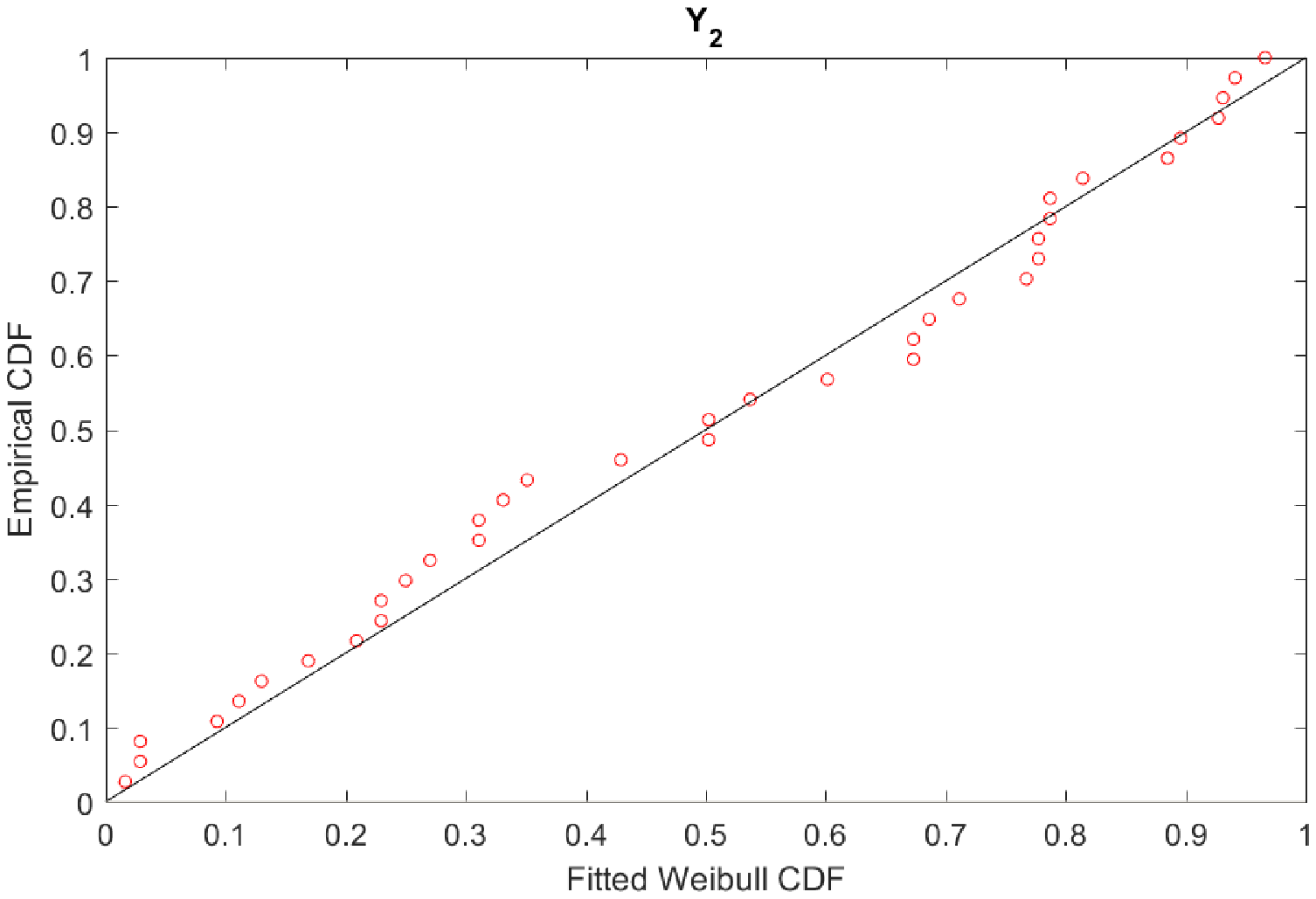}}
		\subfigure[]{\label{c1}\includegraphics[height=1.5in,width=2.2in]{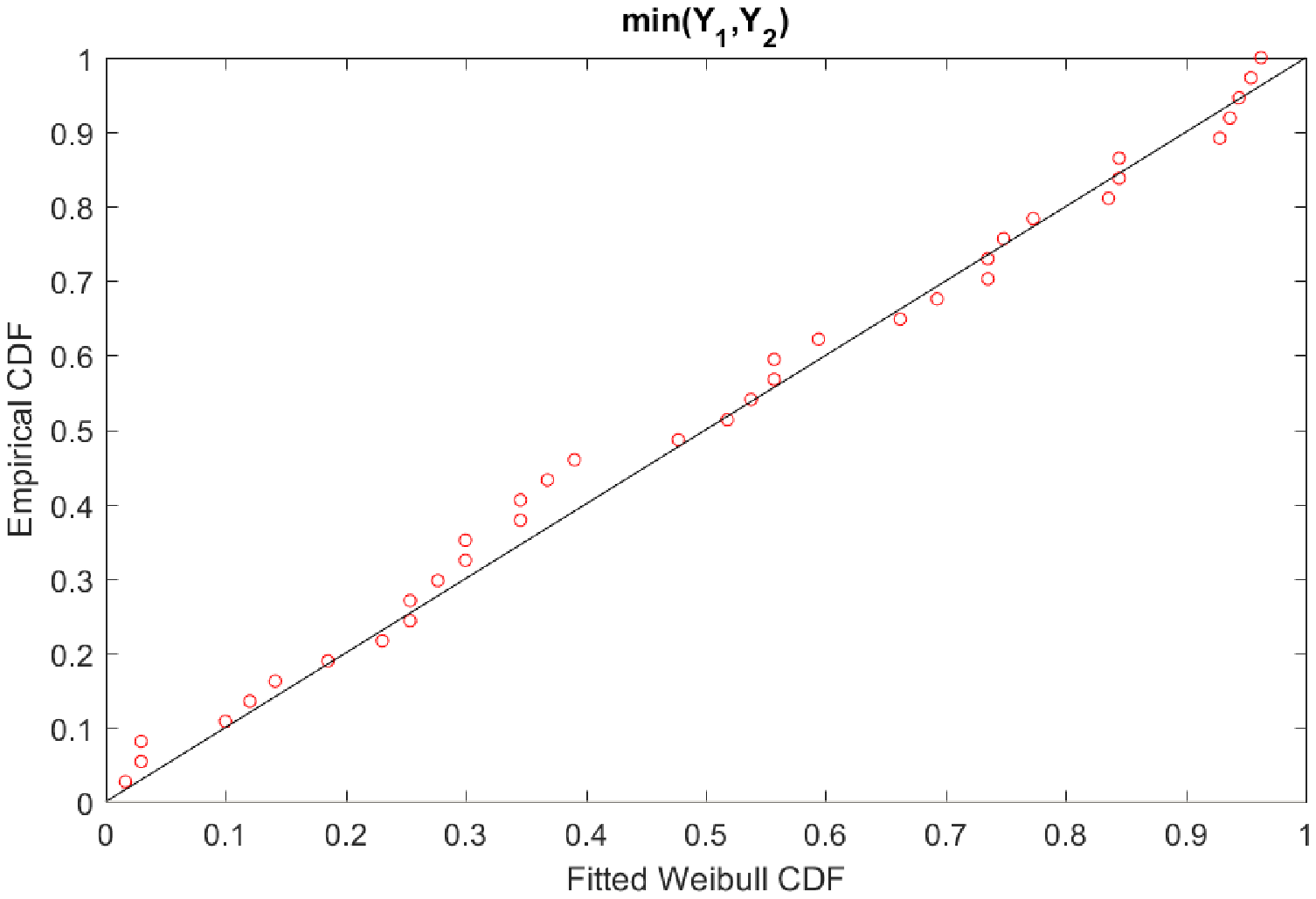}}
		\caption{P-P plots along with fitted Weibull models for real data set.}
	\end{center}
\end{figure}

\begin{figure}[htbp!]
	\begin{center}
		\subfigure[]{\label{c1}\includegraphics[height=1.5in,width=2.2in]{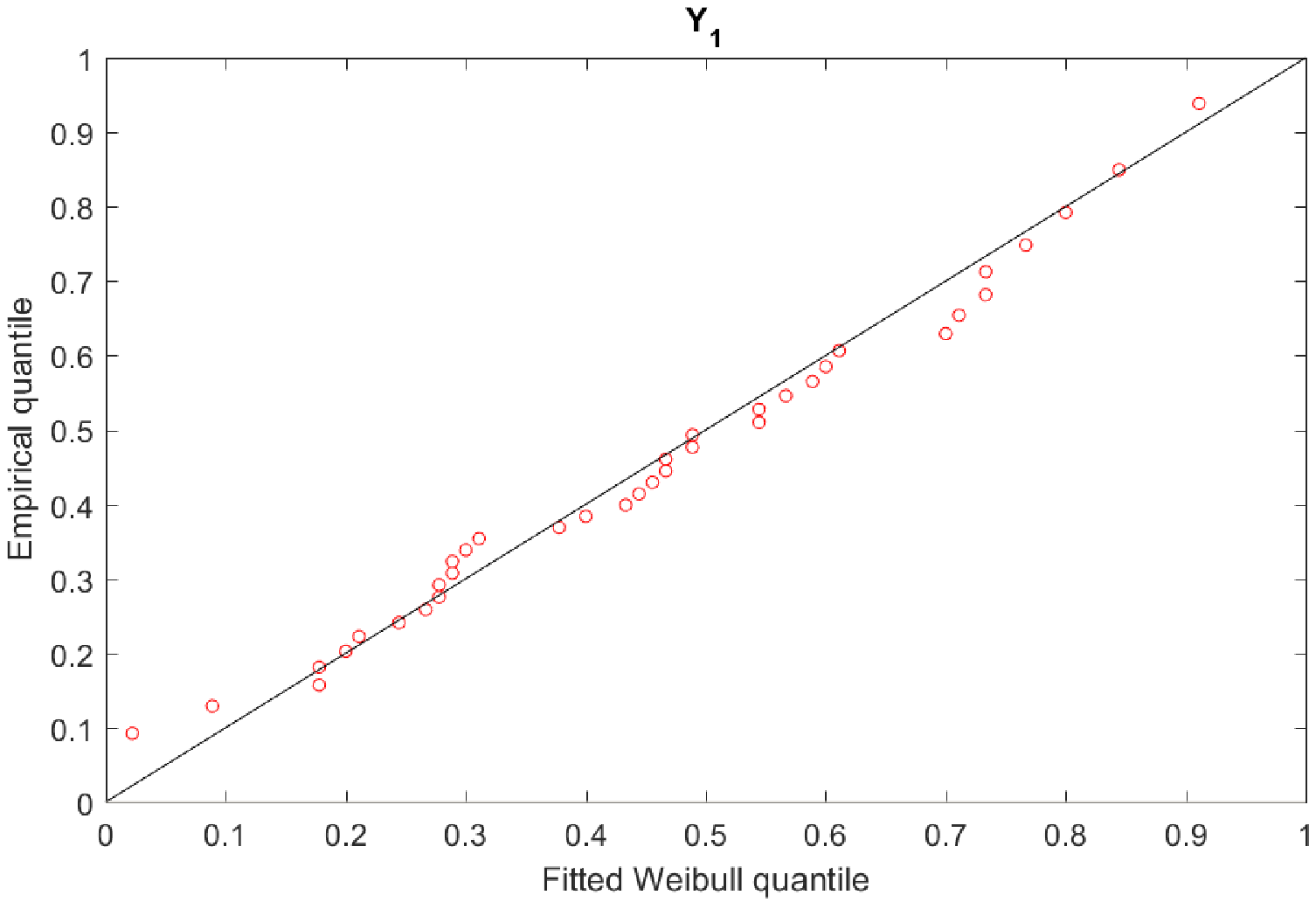}}
		\subfigure[]{\label{c1}\includegraphics[height=1.5in,width=2.2in]{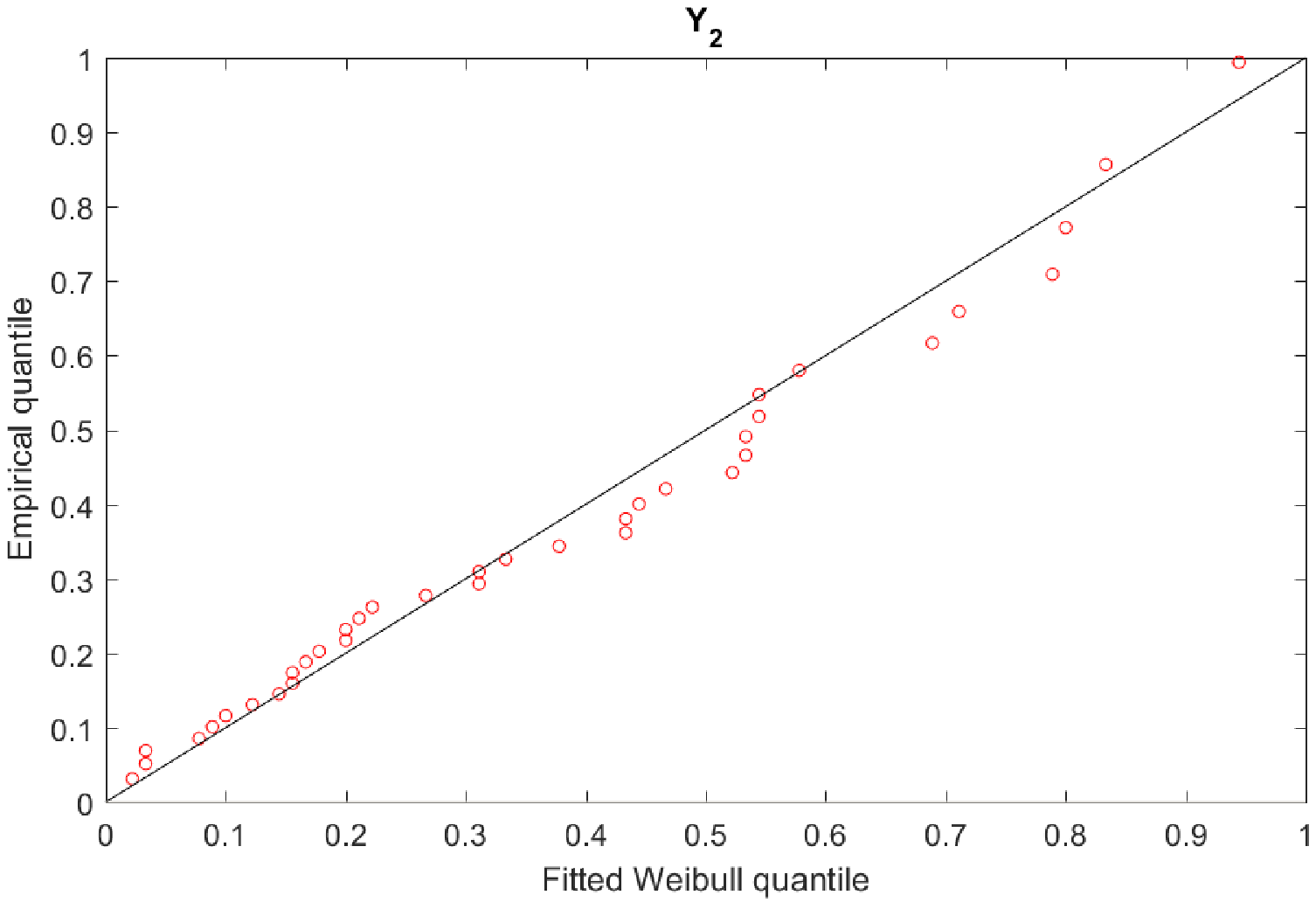}}
		\subfigure[]{\label{c1}\includegraphics[height=1.5in,width=2.2in]{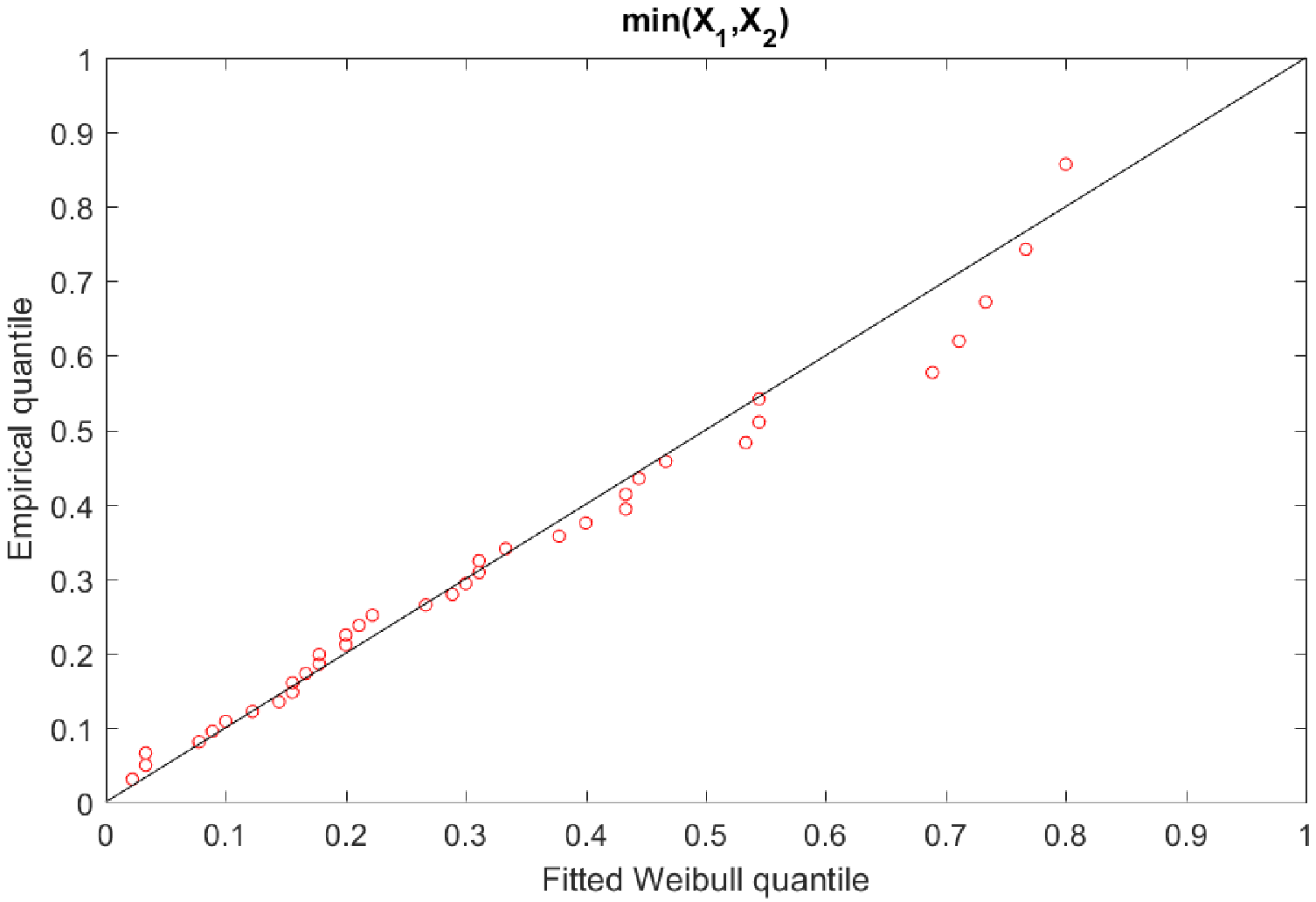}}
		\caption{Q-Q plots along with fitted Weibull models for real data set.}
	\end{center}
\end{figure}

Based on these data set, three different censored data sets have been generated under different censoring schemes for min $(Y_1,Y_2)$ with $q=0$, see Table $\ref{T7}$. The point and interval estimates are calculated and tabulated in Table $\ref{T8}$, where the interval widths are reported in square brackets. Since there is no prior information, the Bayes estimates have been computed under non-informative prior with hyperparameters $0.0001$. Based on three above mentioned data sets, MCMC samples for model parameters are plotted in Figures $5-7$. It has been noticed from Table $\ref{T8}$ that the MLEs and Bayes estimates are extremely comparable. HPD intervals performs better than ACIs in terms of width, which yields that the results are consistent with the previously simulated results. From Table $\ref{T9}$, it has been observed that the censoring scheme $R=(0*29,7)$ is optimal corresponding to the other proposed censored schemes according to $A$-, $D$-, and $F$-optimality criteria given in Table $\ref{T5}$. It is apparent that the suggested algorithm is adequate for processing AT-II PHC  dependent competing risks data.

\begin{figure}[htbp!]
	\begin{center}
		\subfigure[]{\label{c1}\includegraphics[height=1.2in,width=1.6in]{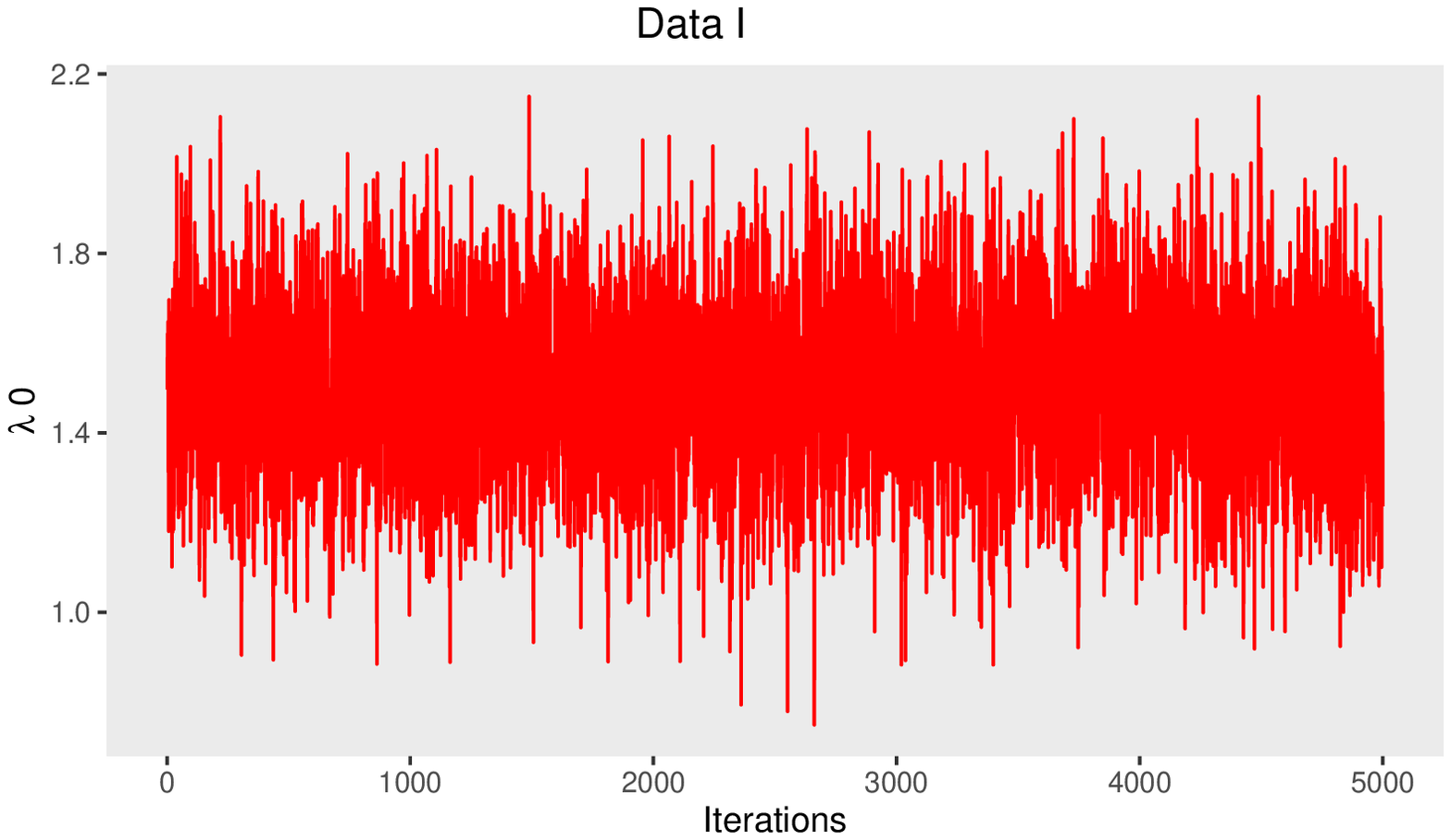}}
		\subfigure[]{\label{c1}\includegraphics[height=1.2in,width=1.6in]{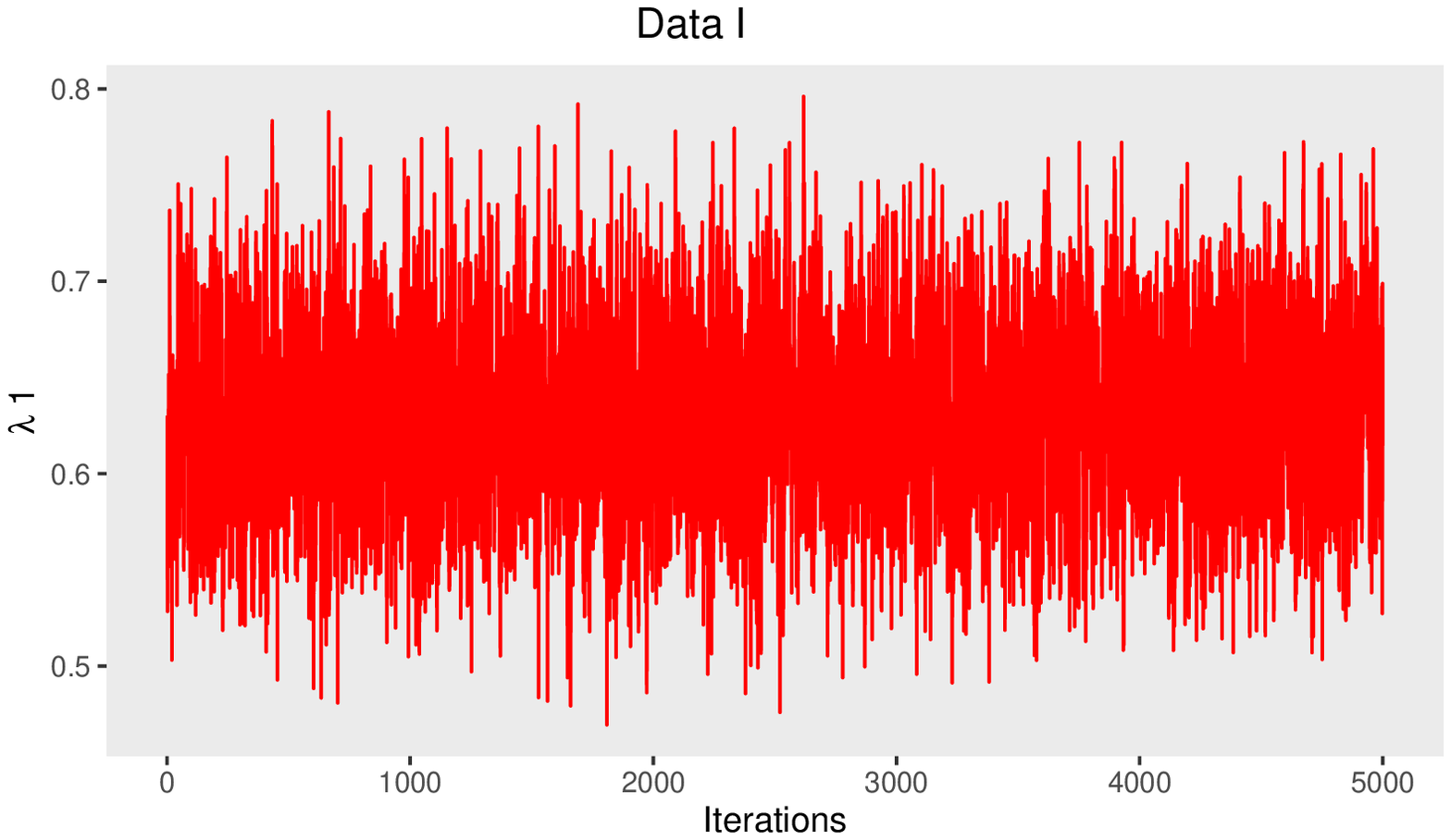}}
		\subfigure[]{\label{c1}\includegraphics[height=1.2in,width=1.6in]{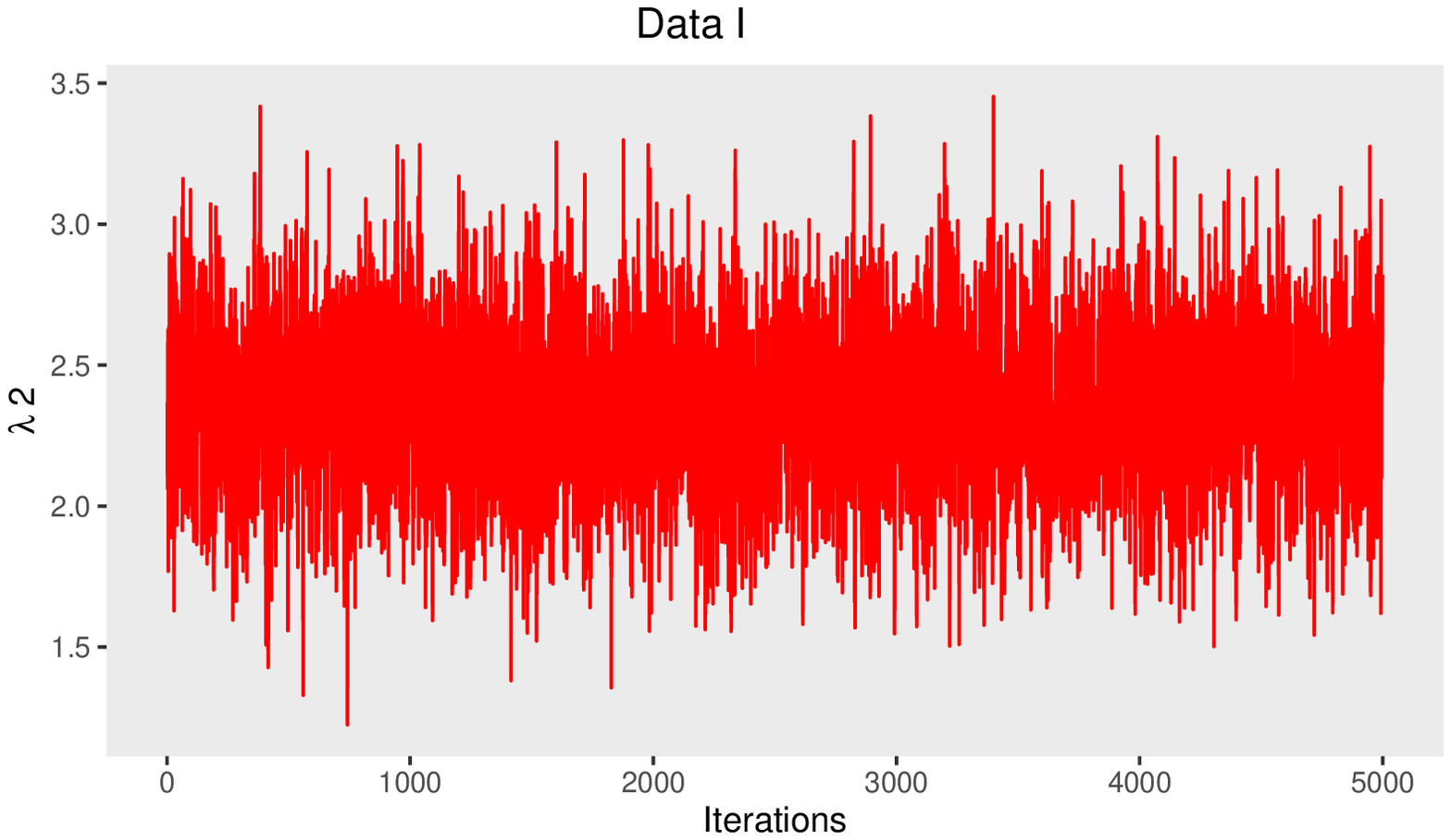}}
		\subfigure[]{\label{c1}\includegraphics[height=1.2in,width=1.6in]{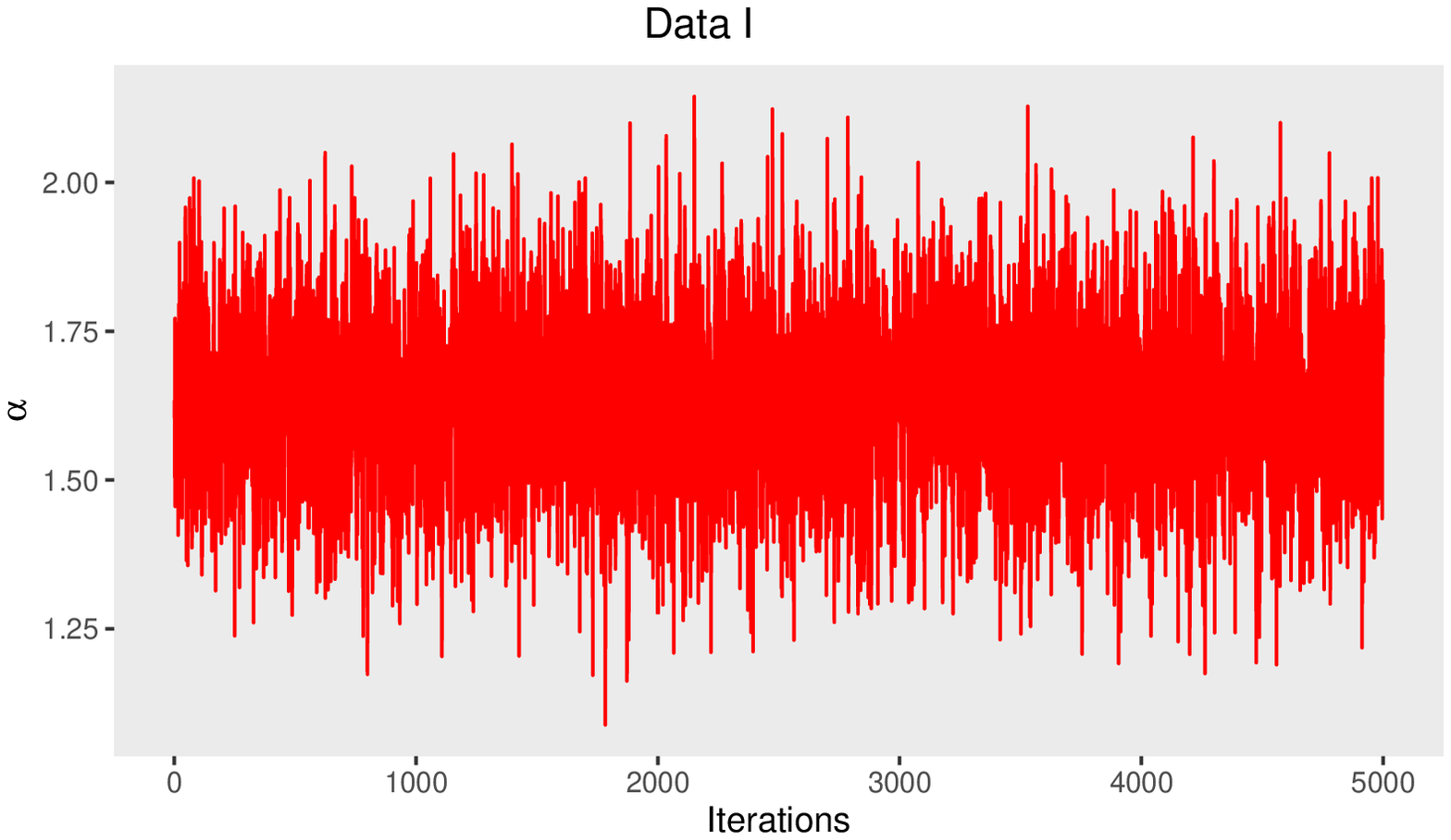}}
		\caption{MCMC samples of (a) $\lambda_{0}$, (b) $\lambda_{1}$, (c) $\lambda_2$ and (d) $\alpha$ for Data I.}
	\end{center}
\end{figure}

\begin{figure}[htbp!]
	\begin{center}
		\subfigure[]{\label{c1}\includegraphics[height=1.2in,width=1.6in]{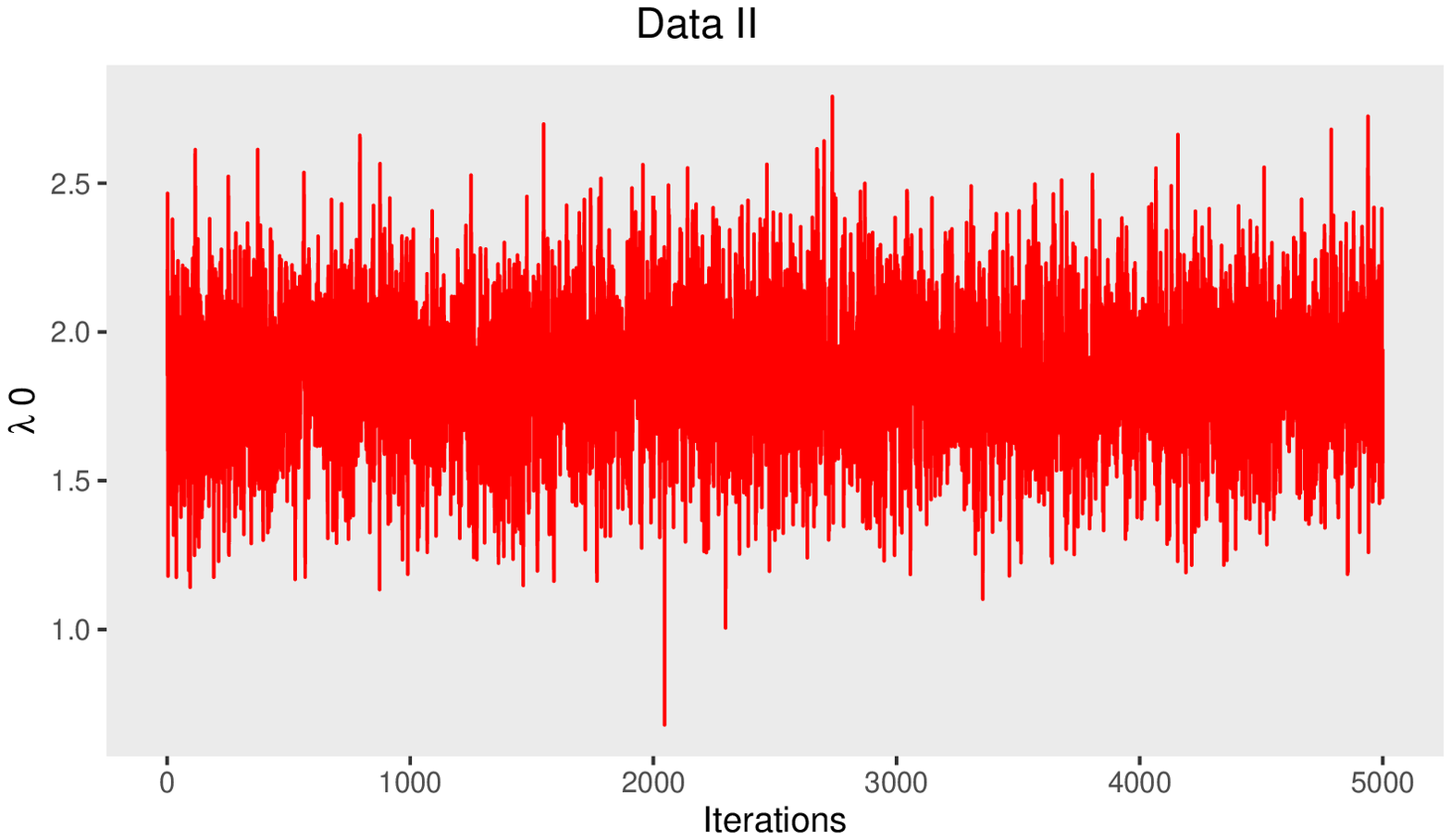}}
		\subfigure[]{\label{c1}\includegraphics[height=1.2in,width=1.6in]{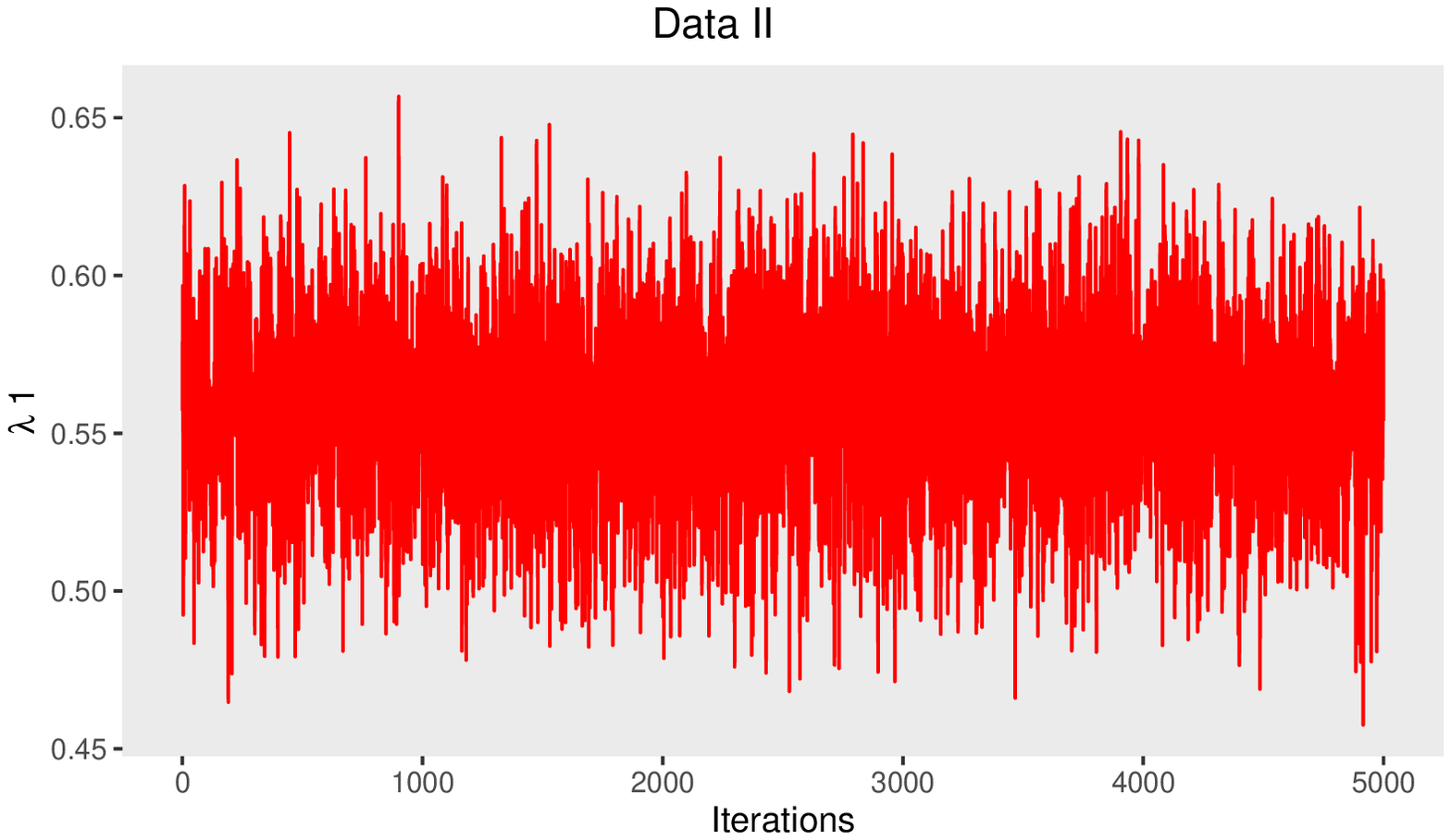}}
		\subfigure[]{\label{c1}\includegraphics[height=1.2in,width=1.6in]{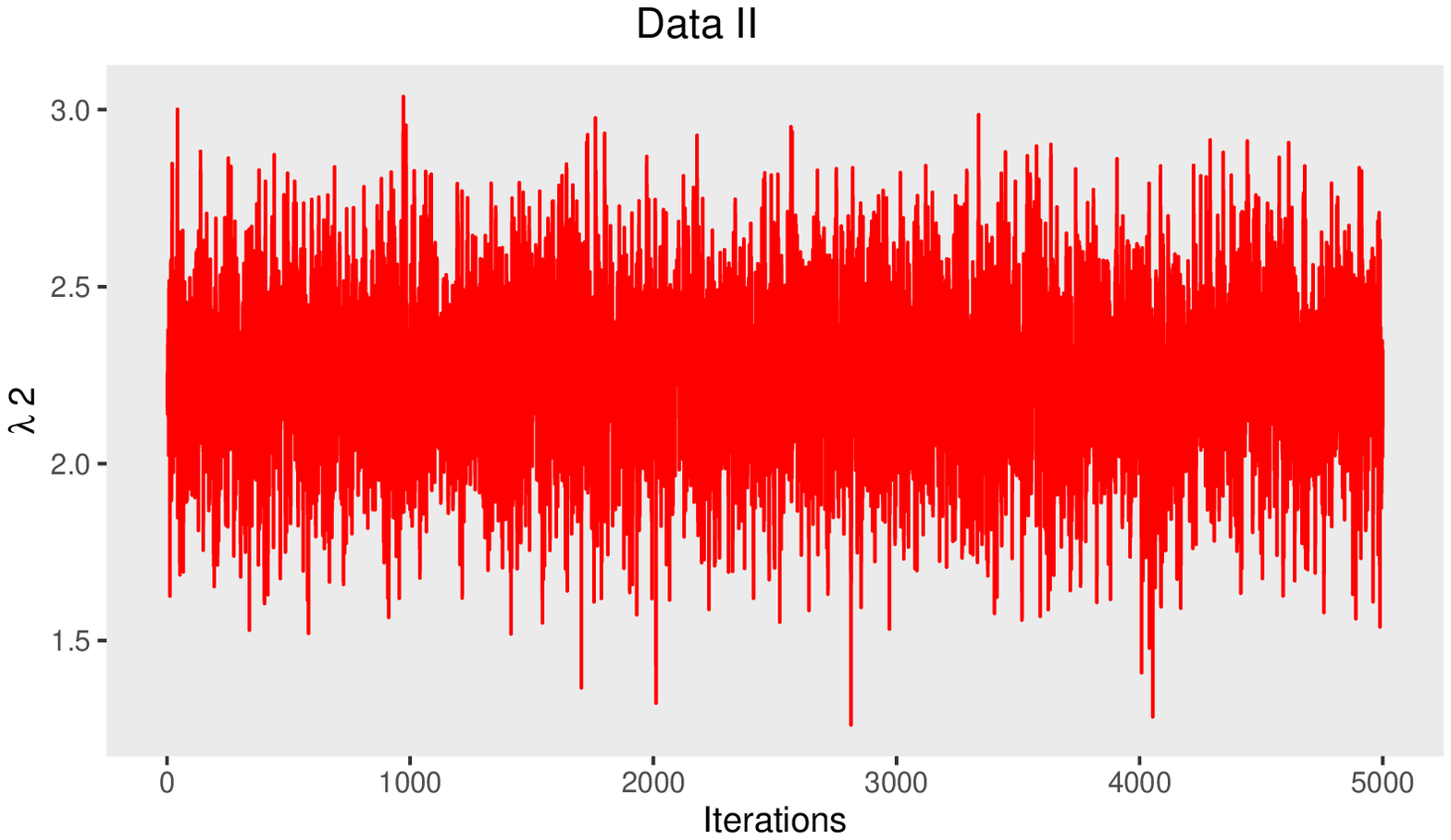}}
		\subfigure[]{\label{c1}\includegraphics[height=1.2in,width=1.6in]{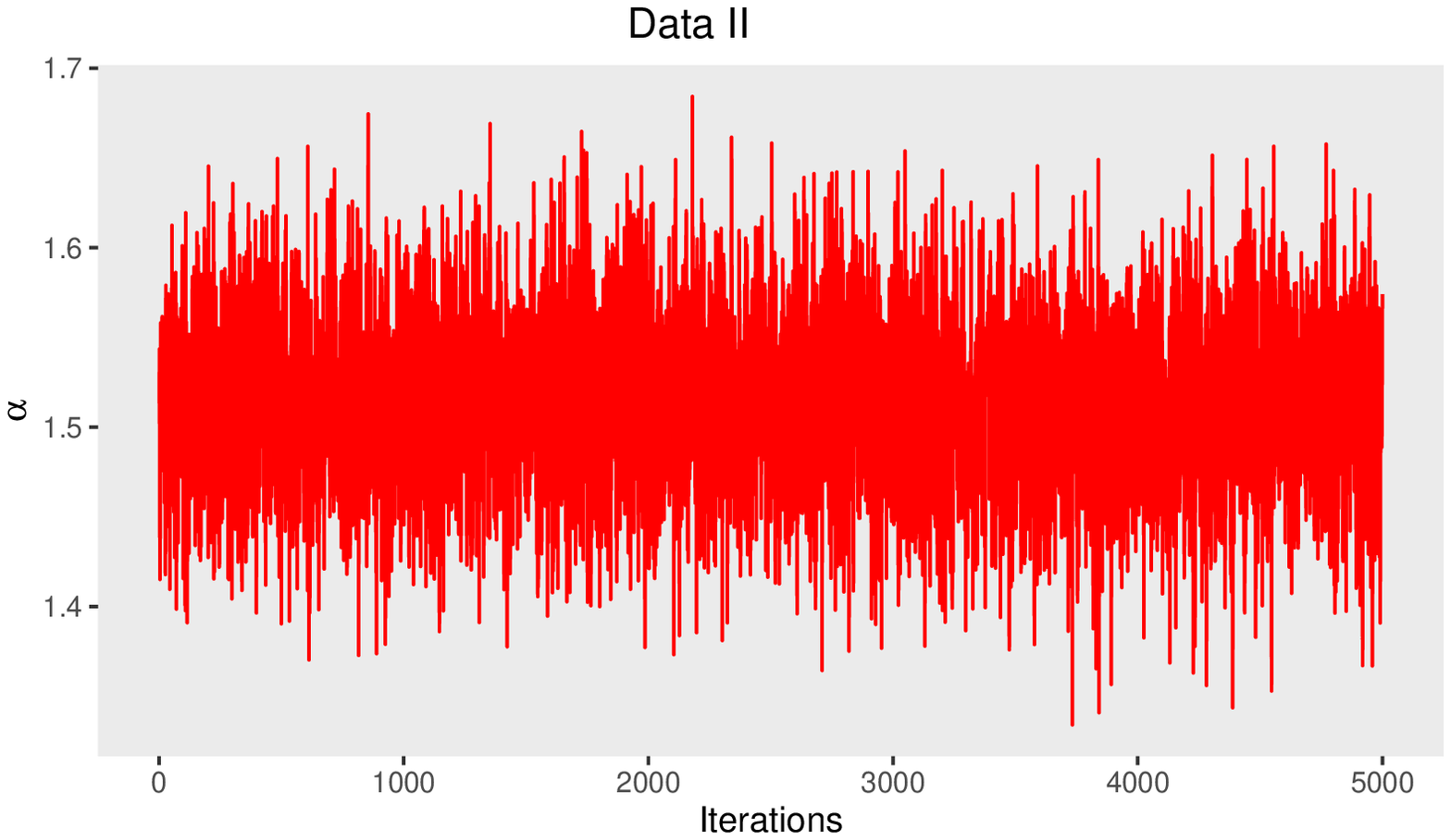}}
		\caption{MCMC samples of (a) $\lambda_{0}$, (b) $\lambda_{1}$, (c) $\lambda_2$ and (d) $\alpha$ for Data II.}
	\end{center}
\end{figure}

\begin{figure}[htbp!]
	\begin{center}
		\subfigure[]{\label{c1}\includegraphics[height=1.2in,width=1.6in]{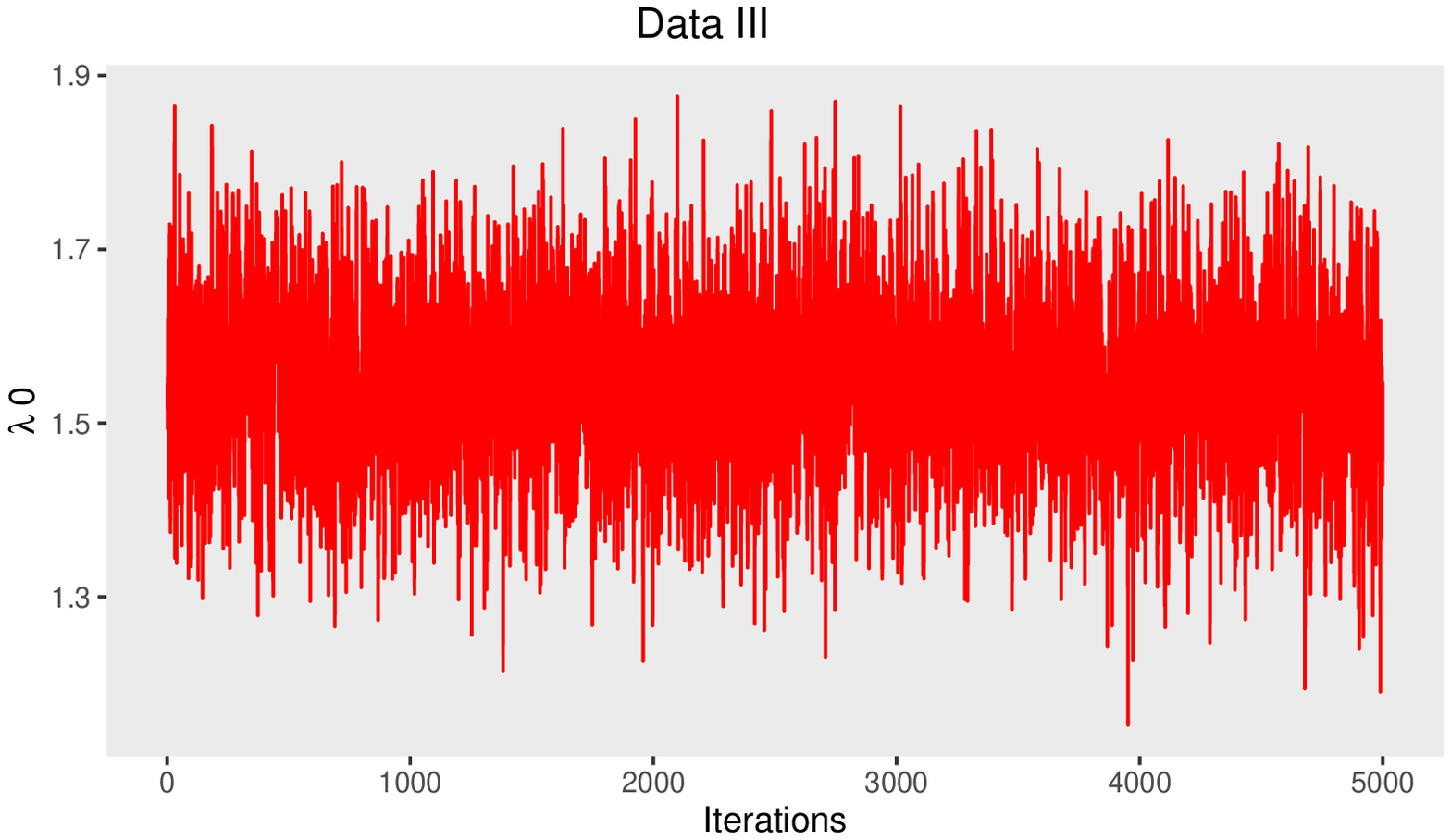}}
		\subfigure[]{\label{c1}\includegraphics[height=1.2in,width=1.6in]{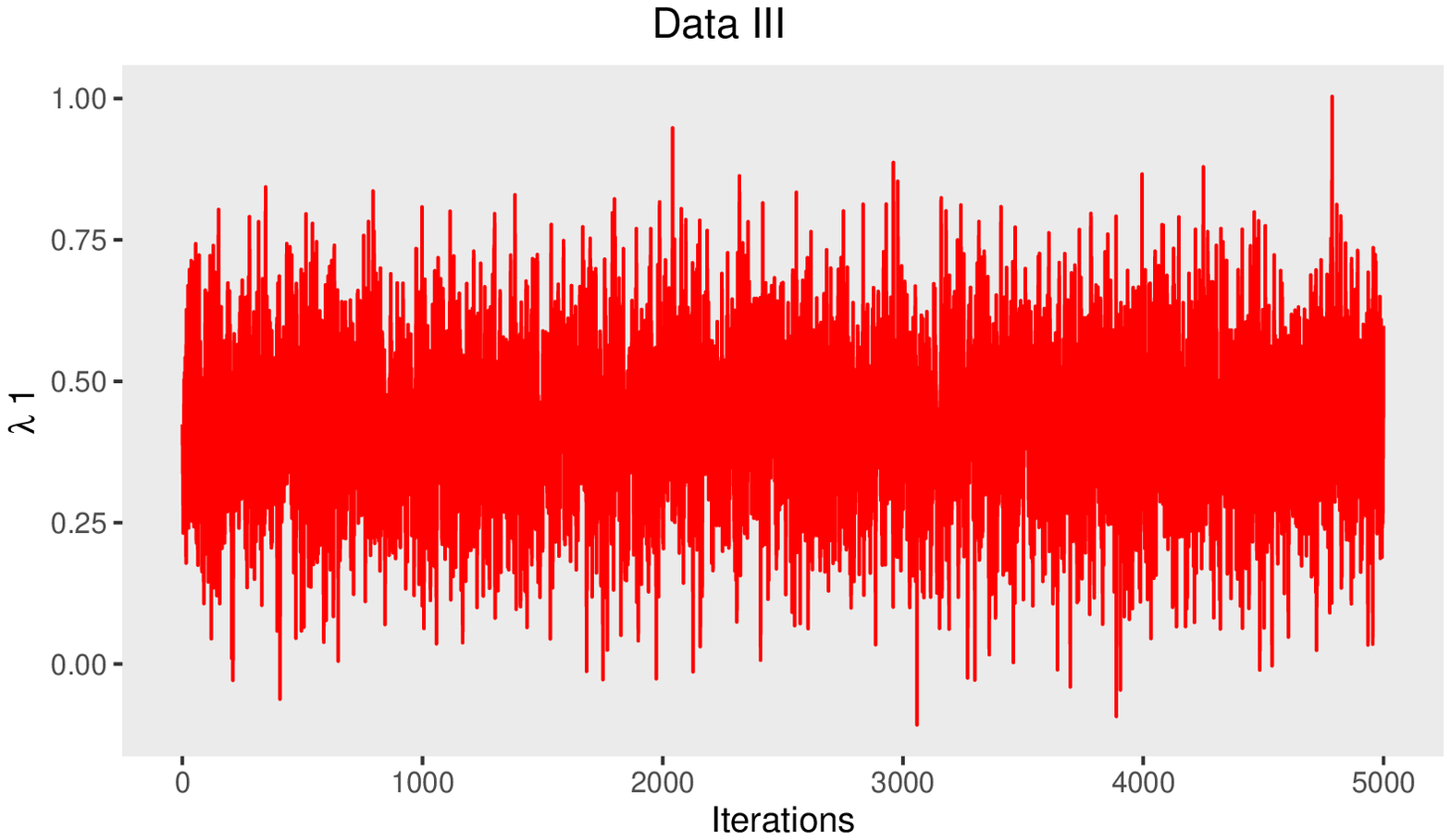}}
		\subfigure[]{\label{c1}\includegraphics[height=1.2in,width=1.6in]{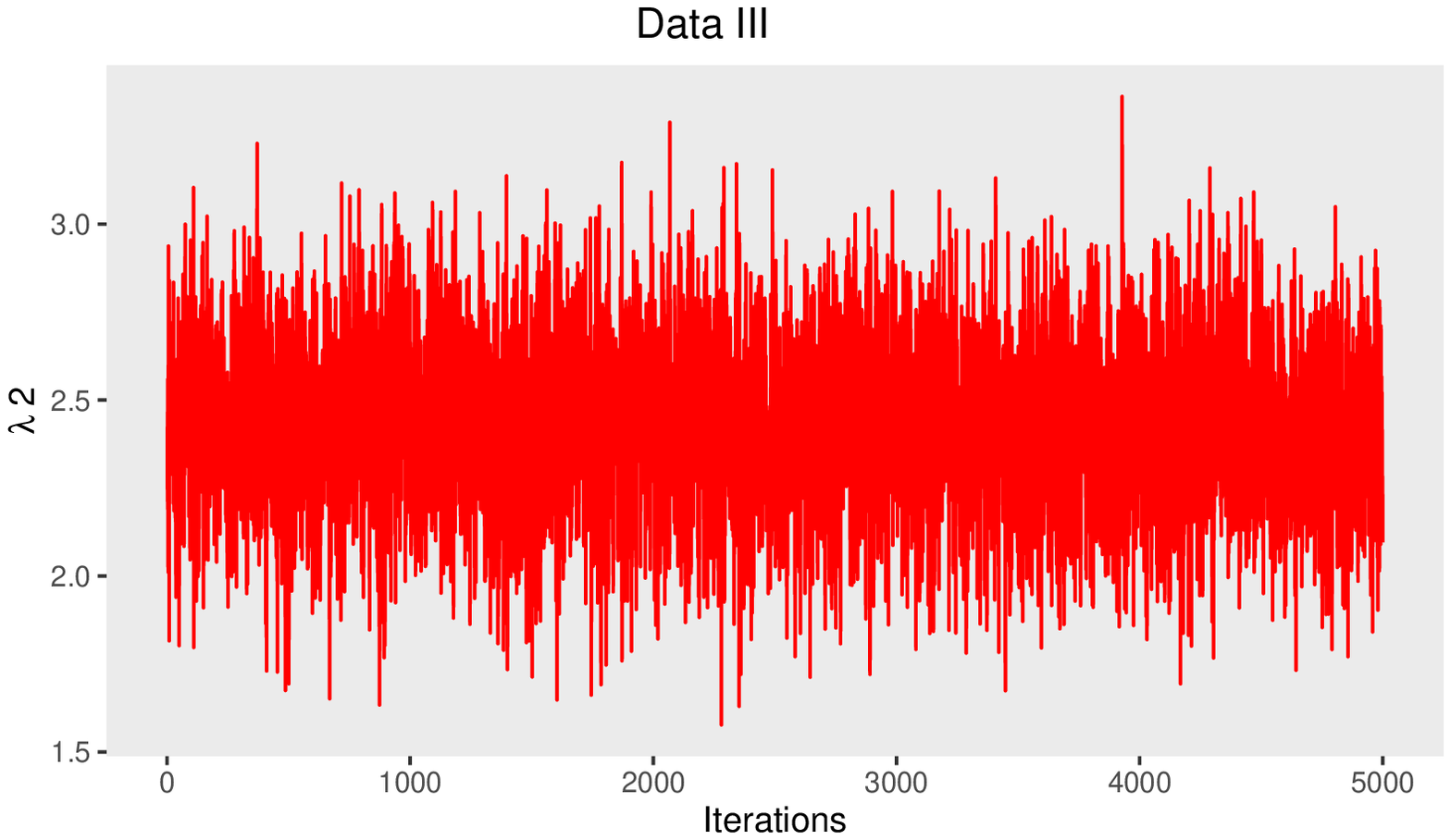}}
		\subfigure[]{\label{c1}\includegraphics[height=1.2in,width=1.6in]{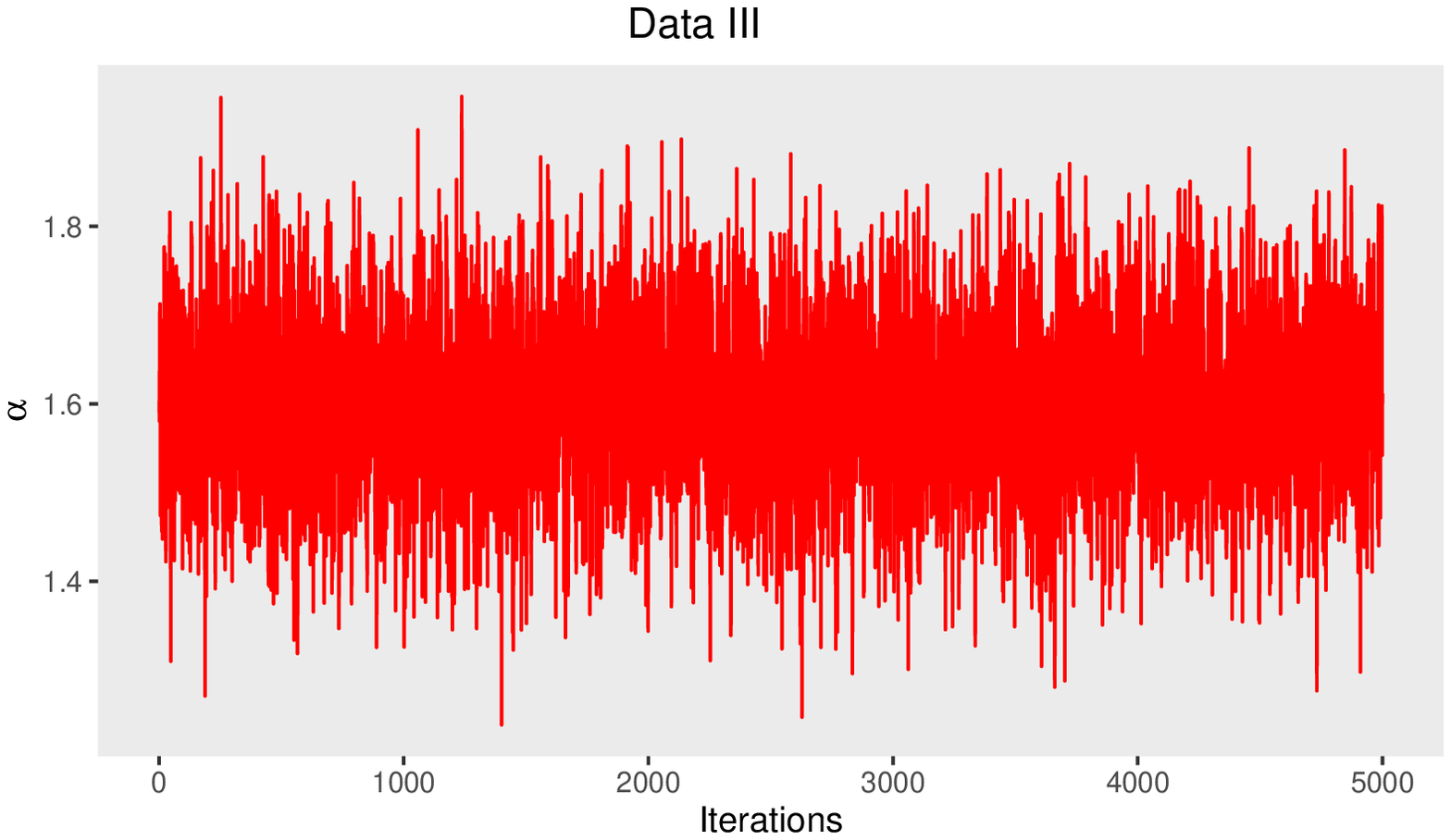}}
		\caption{MCMC samples of (a) $\lambda_{0}$, (b) $\lambda_{1}$, (c) $\lambda_2$ and (d) $\alpha$ for Data III.}
	\end{center}
\end{figure}

\begin{table}[htbp!]
	\begin{center}
		\caption{Optimal censoring scheme using three AT-II PHCS based on real life data set.}
		\label{T9}
		\tabcolsep 7pt
		\small
		\scalebox{1}{
			\begin{tabular}{*{12}c*{11}{r@{}l}}
				\toprule
				\multicolumn{1}{c}{Data} & \multicolumn{1}{c}{Criterion I}&
				\multicolumn{1}{c}{Criterion II}& \multicolumn{1}{c}{Criterion III}   \\
				\midrule
				I& \textbf{0.8587}& \textbf{1.7845} $\times 10^{-4}$& \textbf{71.1219} \\
				II& 1.0333& 4.1915 $\times 10^{-4}$& 58.2624 \\
				III& 0.9533& 3.7313 $\times 10^{-4}$& 54.2736   \\
				\bottomrule
			\end{tabular}}
		\end{center}
	\vspace{-0.5cm}
\end{table}

\section{Conclusion}
In this article, statistical inference for a MOBW distribution based on adaptive Type II progressive hybrid censored data under a dependent competing risk model is developed. Based on both frequentist and Bayesian approaches, point and interval estimation procedures are proposed. Since MLEs of the unknown parameters cannot be derived explicitly, Newton's iterative method has been implemented for this purpose. The existence of MLEs of the unknown parameters are also established. Similarly, due to complex form, Bayes estimates have been obtained by using the MCMC method. Convergence of this MCMC technique has been tested. To evaluate the effectiveness of the suggested techniques, a Monte Carlo simulation has been carried out. It has been discovered that the Bayesian technique produces superior outcomes to the frequentist approach. In the case of interval estimation, HPD credible intervals perform better than the ACIs in terms of their average width. A real-life data set has been analyzed in order to demonstrate the usefulness of the offered methodologies in practice. An optimal progressive censoring plan has been suggested by using different optimality criteria.\\
\\
\textbf{Acknowledgement:}
The author S. Dutta, expresses gratitude to the Council of Scientific and Industrial Research (C.S.I.R.
Grant No. 09/983(0038)/2019-EMR-I), India, for providing financial support received for this research project. Both authors thank the Department of Mathematics, National Institute of Technology Rourkela, India, for providing research facilities.\\\\
\textbf{Conflicts of interest:} Both authors declare that they have no conflict of interest. \\\\

 \bibliography{myref1}
 \renewcommand{\theequation}{A.1.\arabic{equation}}
 \section*{Appendix A.1 Proof of Theorem \ref{th2.1}}
 When $(Y_1,Y_2)\sim MOBW(\lambda_0,\lambda_1,\lambda_2,\alpha)$, then the joint survival function of $(Y_1,Y_2)$ can be expressed as \\
\begin{align}
	\nonumber S_{Y_1,Y_2}(y_1,y_2)= &~ P(Y_1>y_1,Y_2>y_2)\\
	\nonumber =&~P(V_1>y_1,V_2>y_2,V_0>max\{y_1,y_2\})\\
		\nonumber =&~  S_{WE}(y_1;\alpha, \lambda_1) S_{WE}(y_2;\alpha,\lambda_2) S_{WE}(max\{y_1,y_2\};\alpha,\lambda_0),
\end{align}
 where $max\{y_1,y_2\} \in (0,\infty)\times(0,\infty)$. Furthermore, the joint survival function can be expressed as 
 \begin{align*}
 	S_{Y_1,Y_2}(y_1,y_2)
 	= \begin{cases}
 		S_{WE}(y_1;\alpha, \lambda_1)~ S_{WE}(y_2;\alpha, \lambda_{02}), ~~~\mbox{if}~~ y_1<y_2\\
 		S_{WE}(y_1;\alpha, \lambda_{01})~ S_{WE}(y_2;\alpha, \lambda_2),~~~\mbox{if}~~ y_2<y_1 \\
 		S_{WE}(y;\alpha, \lambda_{012}),~~~~~~~~~~~~~~~~~~~~~~~\mbox{if}~~y=y_1=y_2.
 	\end{cases} 	
 \end{align*} 
\section*{Appendix A.2 Proof of Corollary \ref{cor1}}
Using the joint survival function from Theorem $\ref{th2.1}$, the joint PDF can be derived from the expression $-\frac{\partial^2 S_{(Y_1,Y_2)}(y_1,y_2)}{\partial y_1 \partial y_2}$ when $y_1<y_2$ and $y_2<y_1$, respectively.\\
We know that, 
\begin{align*}
	\int_{0}^{\infty} \int_{0}^{y_2} f_{WE}(y_1;\alpha,\lambda_{1}) f_{WE}(y_2;\alpha,\lambda_{02})~dy_1 dy_2=&~ \frac{\lambda_{1}}{\lambda_{012}},\\
	\int_{0}^{\infty} \int_{0}^{y_2} f_{WE}(y_1;\alpha,\lambda_{01}) f_{WE}(y_2;\alpha,\lambda_{2})~dy_1 dy_2=&~ \frac{\lambda_{2}}{\lambda_{012}},
\end{align*}
and further using the full probability formula, when $y_1=y_2=y$ the joint PDF $f_{(Y_1,Y_2)}(y_1,y_2)$ can be derived as $\frac{\lambda_{0}}{\lambda_{012}}f_{WE}(y;\alpha,\lambda_{012})$. Hence the result is provided.  
 \section*{Appendix A.3 Proof of Theorem \ref{th3.1}}
  %\textbf{Proof of Theorem \ref{th3.1}}\\\\
 From $(\ref{3.4})$, we already have
 \begin{align}
 \frac{\partial \log L}{\partial \lambda_j}= &~ \frac{m_j}{\lambda_j} + \frac{m_3}{\lambda_{012}}-A(\alpha), ~~\mbox{for}~ j=0,1,2. \label{A.1.1}
 \end{align}
 For a fixed $\alpha$ equating (\ref{A.1.1}) to zero, one can easily obtain that $\lambda_j= \frac{m_j}{m_{012}} \frac{m}{A(\alpha)}$, for $j=0,1,2$. Furthermore, it can be proved that the associated conditional MLEs maximize the log-likelihood function $\log L (\alpha, \lambda_0,\lambda_1,\lambda_2)$.
 
 We know that, $\log v \leq v-1$ and using these for $v=\frac{\lambda_j}{\widehat{\lambda}_j}$, where $j=0,1,2$ and $v= \frac{\lambda_{012}}{\widehat{\lambda}_{012}}$ we have
 \begin{align}
 m_{j} \log \lambda_j \leq m_j \frac{\lambda_j}{\widehat{\lambda}_j} - m_j + m_j \log \widehat{\lambda}_j = \frac{m_{012}\lambda_j A(\alpha)}{m}- m_j + m_j \log \widehat{\lambda}_j, \label{A.1.2}
 \end{align}
 and
 \begin{align}
 m_3 \ln \lambda_{012} \leq m_3 \frac{\lambda_{012}}{\widehat{\lambda}_{012}} -m_3 + m_3 \log \widehat{\lambda}_{012} = \frac{m_3 \lambda_{012} A(\alpha)}{m} -m_3+ m_3 \log \widehat{\lambda}_{012}. \label{A.1.3}
 \end{align}
 From Equation $(\ref{3.3})$, we can write that
 \begin{align}
 \log L = & ~ m \log \alpha + \sum_{j=0}^{2} m_j \log \lambda_j+ m_3 \log \lambda_{012} +(\alpha-1) \sum_{i=1}^{m}\log y_i - \lambda_{012} A(\alpha).\label{A.1.4}
 \end{align}
 Using (\ref{A.1.2}), (\ref{A.1.3}) and (\ref{A.1.4}) we have
 \begin{align}
 \log L (\lambda_0,\lambda_1,\lambda_2,\alpha) \leq m \log \alpha + \sum_{j=0}^{2} m_j \log \widehat{\lambda}_j + m_3 \log \widehat{\lambda}_{012} +(\alpha-1) \sum_{i=1}^{m}\log y_i -m.
 \end{align}
 Since $m=\widehat{\lambda}_{012}A(\alpha)$, then we have
 \begin{align}
 \nonumber \log L (\lambda_0,\lambda_1,\lambda_2,\alpha) \leq &~ m \log \alpha + \sum_{j=0}^{2} m_j \log \widehat{\lambda}_j + m_3 \log \widehat{\lambda}_{012} +(\alpha-1) \sum_{i=1}^{m}\log y_i - \widehat{\lambda}_{012}A(\alpha)\\
 \nonumber =&~ \log L (\widehat{\lambda}_0,\widehat{\lambda}_1,\widehat{\lambda}_2,\alpha),
 \end{align}
 and equality holds iff $(\lambda_0,\lambda_1,\lambda_2)= (\widehat{\lambda}_0,\widehat{\lambda}_1,\widehat{\lambda}_2)$.

 \renewcommand{\theequation}{A.2.\arabic{equation}}
 \section*{Appendix A.4 Proof of Theorem \ref{th4.1}}
 %\noindent \textbf{Appendix $A.2$. Proof of Theorem \ref{th4.1}}\\\\
 The marginal posterior density of $\alpha$ from $(\ref{4.5})$ can be obtained as follows
 \begin{align}
 \nonumber \pi_{1}^{*}(\alpha|data) &= \int_{0}^{\infty}\int_{0}^{\infty}\int_{0}^{\infty} \pi(\lambda_{0},\lambda_{1},\lambda_{2},\alpha|data) ~d\lambda_{0}~d\lambda_{1}~d\lambda_{2}\\
  &\nonumber \propto \pi_2(\alpha) \alpha^m \bigg(\prod_{i=1}^{m}y_{i}^{\alpha-1}\bigg)\int_{0}^{\infty}\int_{0}^{\infty}\int_{0}^{\infty} \lambda_{012}^{m_3+a-d_{012}} \bigg(\prod_{j=0}^{2} \lambda_{j}^{m_j+d_j-1}\bigg)\\
  & ~~\times e^{-\lambda_{012}\big(b+A(\alpha)\big)}~d\lambda_{0}~d\lambda_{1}~d\lambda_{2}.
 \end{align}
 Let us assume that $u_1=\lambda_{012}$, $u_2=\frac{\lambda_{1}}{\lambda_{012}}$, and $u_3=\frac{\lambda_{2}}{\lambda_{012}}$. The transformation from $(\lambda_{0},\lambda_{1},\lambda_{2})$ to $(u_1,u_2,u_3)$ is one-to-one transformation such that
 \begin{equation*}
 \lambda_{1}= u_1u_2,~~\lambda_{2}=u_1u_3,~~ \mbox{and}~ \lambda_{0}= u_1(1-u_2-u_3),
 \end{equation*}
 where $0<u_1<\infty$, $0<u_2+u_3<1$. The Jacobian of the above transformation is
 \begin{equation*}
 G=\frac{\partial (\lambda_{0},\lambda_{1},\lambda_{2})}{\partial (u_1,u_2,u_3)}= {\begin{bmatrix}
 	1-u_2-u_3 & -u_1 & -u_1\\
 	u_2 & u_1 & 0\\
 	u_3 & 0 & u_3\\
 	\end{bmatrix}}, ~~ \mbox{and}~~ det(G)= u_1^{2}.
 \end{equation*}
 Therefore, we have
 \begin{align*}
 \int_{0}^{\infty}&\int_{0}^{\infty}\int_{0}^{\infty} \lambda_{012}^{m_3+a-d_{012}} \bigg(\prod_{j=0}^{2} \lambda_{j}^{m_j+d_j-1}\bigg)e^{-\lambda_{012}\big(b+A(\alpha)\big)}~d\lambda_{0}~d\lambda_{1}~d\lambda_{2}\\
 =&\int_{0}^{\infty}\int_{0<u_2+u_3<1} u_{1}^{a+m-1} u_{2}^{m_1+d_1-1} u_{3}^{m_2+d_2-1} (1-u_2-u_3)^{m_0+d_0-1} e^{-u_1(b+A(\alpha))} ~du_1~du_2~du_3\\
 =& \int_{0}^{\infty} u_{1}^{a+m-1} e^{-u_1(b+A(\alpha))} du_1 \times\int_{0<u_2+u_3<1} u_{2}^{m_1+d_1-1} u_{3}^{m_2+d_2-1} (1-u_2-u_3)^{m_0+d_0-1} ~du_2~du_3 \\
 =& \Gamma(a+m) \big(b+A(\alpha)\big)^{-(a+m)}\cdot B(m_1+d_1,m_0+d_0+m_2+d_2)\cdot B(m_2+d_2,m_0+d_0),
 \end{align*}
 where $B(\cdot,\cdot)$ represents the Beta function. This completes the proof of the first part. The joint posterior density of $(\lambda_{0},\lambda_{1},\lambda_{2})$ for given data and $\alpha$, can be written as follows
 \begin{align*}
 \pi_{2}^{*}(\lambda_{0},\lambda_{1},\lambda_{2}|\alpha,data) \propto  \pi(\lambda_{0},\lambda_{1},\lambda_{2},\alpha|data)/ \pi(\alpha|data)
 \end{align*}
 which completes the proof of second part of the theorem.

 \renewcommand{\theequation}{A.3.\arabic{equation}}
 \section*{Appendix A.5 Proof of Theorem \ref{th4.2}}
 %\noindent \textbf{Appendix $A.3$. Proof of Theorem \ref{th4.2}}\\\\
 The marginal posterior density of $\alpha$ from $(\ref{4.13})$ can be written as
  \begin{align}
 	\pi_1^{*}(\alpha|data) \propto \alpha^{m+a_1-1} e^{-\alpha\big(b_1-\sum_{i=1}^{m}\log y_i\big)} \big(b+A(\alpha)\big)^{-(a+m)}. \label{A.3.7}
 \end{align}
Now taking logarithm on both sides of (\ref{A.3.7}) we get
\begin{align}
	\log \pi_1^{*}(\alpha|data) = (m+a_1-1) \log \alpha -\bigg(b_1-\sum_{i=1}^{m}\log y_i\bigg) \alpha -(a+m) \log \big(b+A(\alpha)\big). \label{A.3.8}
\end{align}
After differentiating (\ref{A.3.8}) partially with respect to $\alpha$ twice we get
\begin{align*}
	\frac{\partial^{2} \log \pi_1^{*}}{\partial \alpha^2} = -\frac{m+a_1-1}{\alpha^2} - (a+m) \frac{(b+A(\alpha))A^{\prime\prime}(\alpha)- \big(A^{\prime}(\alpha)\big)^2}{(b+A(\alpha))^2},
\end{align*}
 where $A^{\prime\prime}(\alpha)=\sum_{i=1}^{m} y_i^{\alpha} \log^2 y_i + \sum_{i=1}^{J}R_i y_i^{\alpha} \log^2 y_i + R^{*} y_m^{\alpha} \log^2 y_m \geq 0$. 
 According to Cauchy-Schwartz inequality, we know that $A(\alpha)A^{\prime\prime}(\alpha)-\big(A^{\prime}(\alpha)\big)^2 \geq 0$ and this implies that  $\frac{\partial^{2} \log \pi_1^{*}}{\partial \alpha^2} \leq 0$. Hence the result is proved.

\end{document}